\documentclass[reprint,superscriptaddress,nofootinbib,aps,prl,floatfix]{revtex4-2}

\usepackage[utf8]{inputenc}
\usepackage{CJKutf8}
\usepackage{amsmath}
\usepackage{amssymb}
\usepackage{amsthm}
\usepackage{mathrsfs}
\usepackage{bbm}
\usepackage{bm}
\usepackage{tikz}
\usepackage{amsmath,amssymb}
\usepackage[percent]{overpic}
\usepackage[
    colorlinks=true,
    linkcolor=blue,
    citecolor=red,
    urlcolor=blue,
    breaklinks=true,
    linktoc=page,
    bookmarksdepth=2
]{hyperref}

\usetikzlibrary{arrows.meta,calc,patterns}

\definecolor{surfaceblue}{HTML}{E4EFF6}
\definecolor{meshblue}{HTML}{86A9C2}
\definecolor{boundaryorange}{HTML}{E8893A}
\definecolor{alphablue}{HTML}{3E6D91}
\definecolor{betared}{HTML}{B45C5C}
\definecolor{phaseviolet}{HTML}{A991D4}
\definecolor{axisgray}{HTML}{4F5962}

\usepackage{quantikz}
\usepackage{xcolor}

\usepackage[capitalize]{cleveref}
\usepackage[ruled,vlined]{algorithm2e}

\DeclareMathOperator{\Tr}{Tr}

\renewcommand{\d}{\ensuremath{\mathrm{d}}}

\newcommand{\m}{\textemdash}
\renewcommand {\braket} [2] {\ensuremath{\left \langle #1 | #2 \right \rangle}}
\renewcommand{\proj}[1]{| #1 \rangle\!\langle #1 |}

\newcommand{\Var}{\mathrm{Var}}
\newcommand{\Bias}{\mathrm{Bias}}
\newcommand{\MSE}{\mathrm{MSE}}

\def \iu {\mathrm{i}}
\def \bra#1{\langle#1|}
\def \ket#1{|#1\rangle}
\def \ketbra#1#2{|#1\rangle\! \langle #2|}
\def \braket#1#2{\langle #1|#2\rangle}

\def \id {\mathbbm{I}}

\def \poly {\mathrm{poly}}
\def \polylog {\mathrm{polylog}}

\makeatletter
\newcommand{\newreptheorem}[2]{\newtheorem*{rep@#1}{\rep@title}\newenvironment{rep#1}[1]{\def\rep@title{#2 \ref*{##1}}\begin{rep@#1}}{\end{rep@#1}}}
\makeatother

\makeatletter
\newcommand*{\rom}[1]{\expandafter\@slowromancap\romannumeral #1@}
\makeatother

\newtheorem{theorem}{Theorem}
\newreptheorem{theorem}{Theorem}
\newtheorem{lemma}[theorem]{Lemma}
\newreptheorem{lemma}{Lemma}
\newtheorem{proposition}[theorem]{Proposition}
\newreptheorem{proposition}{Proposition}
\newtheorem{definition}{Definition}
\newtheorem{remark}{Remark}

\makeatletter
\renewcommand*\l@subsubsection[2]{}
\makeatother

\newif\ifarxivSupplementToc
\makeatletter
\AtBeginDocument{%
  \let\arxiv@addtocontents\addtocontents
  \long\def\addtocontents#1#2{%
    \def\arxiv@extension{#1}%
    \def\arxiv@toc{toc}%
    \ifx\arxiv@extension\arxiv@toc
      \ifarxivSupplementToc\arxiv@addtocontents{#1}{#2}\fi
    \else
      \arxiv@addtocontents{#1}{#2}%
    \fi
  }%
}
\makeatother

\begin{document}

\title{Unbounded separation between definite and indefinite causal order in finite-dimensional quantum metrology}

\author{Yanglin Hu(胡杨林)}
\email{yanglin.hu@u.nus.edu}
\affiliation{QICI Quantum Information and Computation Initiative, School of Computing and Data Science, The University of Hong Kong, Pokfulam Road, Hong Kong SAR, China}

\author{Zi-Shen Li(李子绅)}
\email{zishen@connect.hku.hk}
\affiliation{QICI Quantum Information and Computation Initiative, School of Computing and Data Science, The University of Hong Kong, Pokfulam Road, Hong Kong SAR, China}

\author{Giulio Chiribella}
\email{giulio@cs.hku.hk}
\affiliation{QICI Quantum Information and Computation Initiative, School of Computing and Data Science, The University of Hong Kong, Pokfulam Road, Hong Kong SAR, China}

\author{Yuxiang Yang(杨宇翔)}
\email{yuxiang@cs.hku.hk}
\affiliation{QICI Quantum Information and Computation Initiative, School of Computing and Data Science, The University of Hong Kong, Pokfulam Road, Hong Kong SAR, China}

\date{October 1, 2026}
\begin{abstract}

    Indefinite causal order is known to offer enhancements in quantum metrology, notably including an unbounded advantage in  the measurement of  a geometric phase of a harmonic oscillator. This advantage, however, is specific to  the infinite dimensional setting,  and its finite dimensional analogue remains elusive, with recent findings suggesting that  advantages in finite dimensions may be fundamentally limited to bounded constant factors.   Here we show that, in fact,  arbitrarily large advantages arise for finite dimensional systems in the finite sample regime. Specifically, we  establish an unbounded separation between definite and indefinite causal order in the estimation of a geometric phase associated to two sets of $N$ displacements generated by discrete position and momentum operators on a $d$-dimensional quantum system:  for any given constant $R$, there exist values   of  $N$ and $d=\Omega(N^2)$ such that a strategy with indefinite order uses an initial probe with  $R$ times less energy than the probe required by every strategy with definite order achieving the same mean squared error, whenever the number of measurement shots $\nu$ is bounded as  $\nu=\mathcal{O}(\exp(\pi d/16)/\poly(d))$.  In other words, indefinite  order offers an energy saving that grows arbitrarily large with the parameters of the problem.  
    To prove this result, we establish an approximate Weyl relation for discrete Gaussian wavepackets, which is of independent technical interest. 
\end{abstract}

\begin{CJK*}{UTF8}{gbsn}
\maketitle
\end{CJK*}

\textit{Introduction.}\m
While the standard model of quantum circuits  assumes that all gates take place a definite causal order (DCO), quantum theory is in principle compatible  with scenarios where the order becomes indefinite~\cite{Hardy_2007,Hardy_2009,chiribella_2009,Oreshkov_2012,Chiribella_2013,Baumeler_2014}. This observation has stimulated intense research on both theory~\cite{araujo_2015,Oreshkov_2016,Wechs_2021,Cao_2023,VanderLugt_2023,Liu_2025,delaHamette_2025} and experiments~\cite{Procopio_2015,Rubino_2017,Goswami_2018,Wei_2019,Guo_2020,Rubino_2022,Yin_2023,Deng_2025,Guo_2026,Richter_2026}. Initially explored as a foundational idea, indefinite causal order (ICO) has emerged as a powerful resource for quantum information processing, impacting a variety of tasks such as quantum channel discrimination~\cite{Chiribella_2012,Bavaresco_2021},  computation~\cite{Colnaghi_2012,Araujo_2014,Taddei_2021,Liu_2024}, communication complexity~\cite{Feix_2015, Guerin_2016, Wei_2019, Guo_2020, Liu_2023}, quantum communication~\cite{Ebler_2018,Goswami_2020,Sazim_2021,Chiribella_2021q,Chiribella_2021i}, and thermodynamics~\cite{Felce_2020, Rubino_2021, Nie_2022, Cao_2022, Zhu_2023, Simonov_2025,Xue_2026}.

Recently, ICO has been identified as a promising resource for quantum metrology. In~\cite{Zhao_2020}, Zhao, Yang, and Chiribella (ZYC) introduced the problem of estimating the product of the average displacements along two orthogonal directions in the phase space of a quantum harmonic oscillator. There, ZYC proved an unbounded separation between strategies with ICO and DCO: as the number of displacements $N$ grows, the number of measurement shots required to estimate the product with ICO is $N$ times smaller than the number required by any DCO strategy achieving the same mean squared error with the same amount of energy. More broadly, the ZYC separation  can be viewed as an advantage in estimating the phase-space geometric phase generated by the Weyl relation, as shown  in~\cite{Yin_2023}. A broader algebraic framework was developed in~\cite{Kong_2026}, attributing such enhancements to the noncommutativity of the encoding generators. Related causal-order-induced enhancements have also been explored in quantum metrology through indefinite-time-direction encoding~\cite{Xia_2024,Xia_2025} and in distributed sensing through causal-order switching~\cite{Xia_2026}.  Besides establishing fundamental separations, ICO-based schemes have been also used  to achieve practical precision enhancements in photonic setups \cite{Chen_2026n}.

All the above examples of unbounded advantages, however, are specific to the infinite dimensional setting. For finite dimensional systems, instead, the first findings gave a completely different picture. Numerical methods reported a finite separation in quantum Fisher information (QFI) between ICO and DCO strategies in specific finite-dimensional systems~\cite{Liu_2023}. Subsequent analytical methods have demonstrated no generic asymptotic advantage of ICO over DCO in the presence of noise~\cite{Kurdzialek_2023}. Moreover, bounds on  the spectral diameter of generators suggested that the separation between ICO and DCO should be bounded by dimensional pre-factors~\cite{Kong_2026}.  The stark contrast between the infinite-dimensional case, where an unbounded separation has been rigorously proven \cite{Zhao_2020}, and the finite-dimensional case, where bounded separations or no separation at all has been found,  poses a fundamental puzzle to our understanding of ICO's role in quantum metrology. Here, a crucial open question is:  \emph{can ICO achieve an unbounded advantage over DCO in finite-dimensional quantum metrology?}

In this letter, we answer the above question affirmatively, by identifying a scenario of finite-dimensional quantum metrology where ICO achieves an unbounded energy-efficiency advantage over DCO. Specifically, we consider the task of estimating the product of two sets of $N$ phase-space displacements on generated  by discrete position and momentum operators on a $d$-dimensional quantum system (see Figure~\ref{fig:phase_space}). Although prior works suggested a bounded advantage of ICO over DCO asymptotically for any fixed $d$, here we show that an unbounded advantage survives in the pre-asymptotic regime: for every arbitrarily large factor $R$, there exists a metrological task with finite $N$ and $d = \Omega(N^2)$ such that an ICO strategy requires $R$ time less energy in the initial probe state than every DCO strategy achieving the  same precision (measured by the mean squared error) in the pre-asymptotic regime where the number of measurement shots is $\nu = \mathcal{O}(\exp[\pi d/16]/\poly(d))$. As $N$ and $d$ increase, this results proves that the advantage of ICO can be arbitrarily large.  Notably, this separation depends only on  the initial probe energy, and remains valid even if more energy is injected by the DCO strategy through subsequent control operations. In other words, the advantage at the level of initial probe energy persists even if the DCO strategy is allowed to ``cheat'' by pumping additional energy at later steps. To the best of our knowledge, this is the first unbounded separation in the initial probe energy required by indefinite and definite causal order strategies across finite-dimensional instances.

As a technical contribution used to prove our results, we establish a finite-dimensional approximation to the canonical commutation relations, in its unitary form due to Hermann Weyl. To this purpose, we exploit   discrete Gaussian wavepacket (DGW) framework,   originally developed for finite-sized quasi-ideal quantum clocks~\cite{Woods_2019,Woods_2021,Woods_2022}.  The DGW framework provides states on which the discrete position and momentum operators approximately satisfy the canonical commutation relation. Here we extend the approximation from the self-adjoint generators to the  unitary displacements generated by them, thereby providing an approximation of the Weyl commutation relation for displacement operators. This result is likely to be useful beyond the scope of the present paper, in the same way as the DGW approximations have found applications in quantum error correction~\cite{Woods_2020}, quantum nondemolition measurements~\cite{Boulebnane_2021} and autonomous quantum devices~\cite{Woods_2023}.

\begin{figure}[t]
\includegraphics[width=0.75\linewidth]{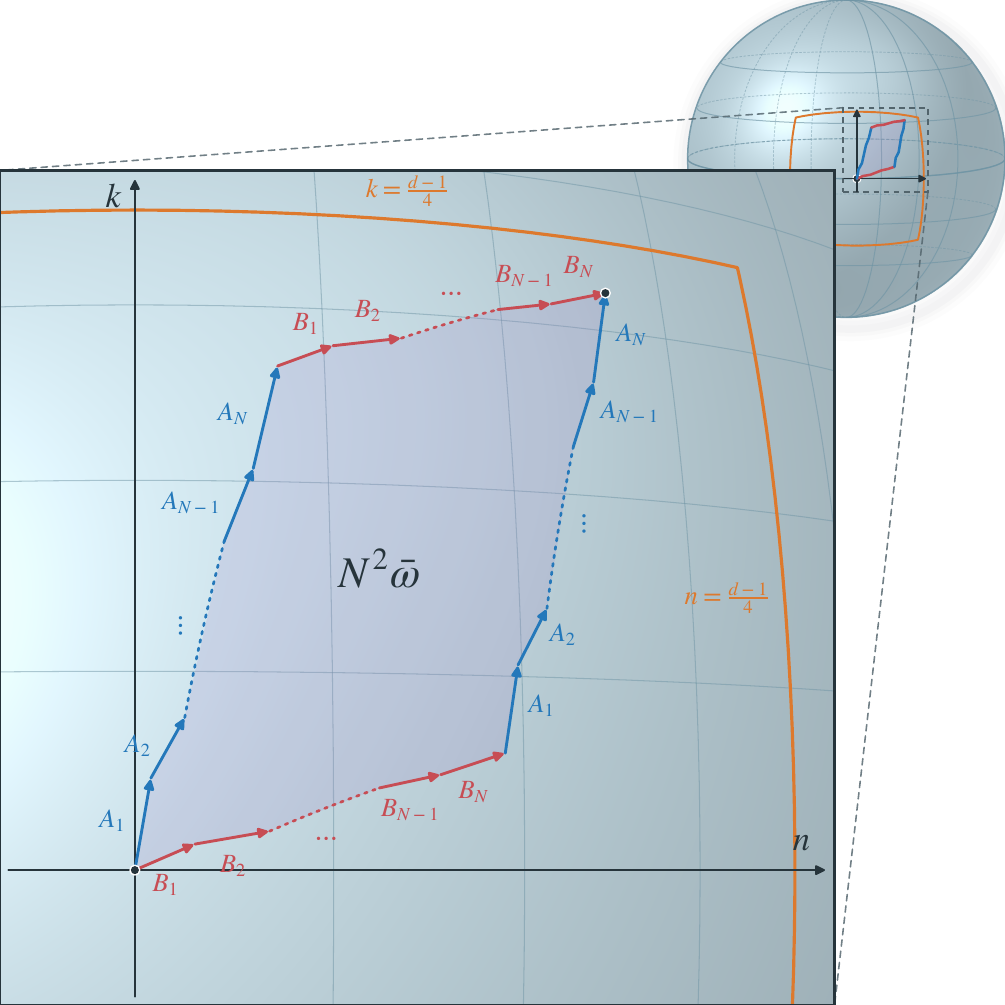}
\caption{The approximate accumulation of the geometric phase in a finite-dimensional phase space. Two trajectories are formed by applying two sets of $N$ unknown phase-space vectors $\alpha_1, \dots, \alpha_N$ corresponding to $A_1, \dots, A_N$ (blue arrows) and $\beta_1, \dots, \beta_N$  corresponding to $B_1, \dots, B_N$ (red arrows) via $\Gamma = \sqrt{d/2\pi} \gamma$ for $\gamma=\alpha_i,\beta_i$ in two opposite causal orders. These two trajectories form a closed phase-space loop enclosing an area $\omega$. By the approximate Weyl relation in Theorem~\ref{thm:general_weyl_relation}, the two states obtained by the two trajectories approximately acquire a relative geometric phase $N^2\bar{\omega}$. The approximation remains valid provided that the centers encountered in these two trajectories are strictly confined to the bulk region, bounded by $|n|, |k| \leq \frac{d-1}{4}$ (orange lines).}

\label{fig:phase_space}
\end{figure}

\textit{Finite-dimensional approximation of the Weyl canonical commutation relation.}\m
Our metrological problem involves a phase-space representation of finite dimensional quantum systems.  For convenience, we take the dimension $d$ to be an odd number ($d=2m+1$ for some positive integer), and we label the canonical basis vectors as $\{\ket{x_n}\}_{n=-m}^m$.  In this setting, the discretized  position  and momentum operators are   $X:= \sqrt{2\pi/d}  \sum_{n} n\, \ketbra{x_n}{x_n}$ and $  P := \sqrt{2\pi/d}  \sum_{k}  k\, \ketbra{p_k}{p_k}$, respectively, where $\{\ket{x_n}\}_{n=-m}^m$ is an orthonormal basis and $\{\ket{p_k} \}_{k=-m}^m$ is its Fourier transform, defined by $\ket{p_k}  :=\sqrt{1/d}\sum_n e^{\iu 2\pi nk/d}\ket{x_n}$.   For a phase-space vector $\gamma=(\gamma_x,\gamma_p)  \in  \mathbb{R}^2$, we  define the unitary displacement operator  $G(\gamma)=e^{\iu g(\gamma)}$, where  $g(\gamma) :=  \gamma_p X-\gamma_x  P$ is the generator of the displacement.   

If $X$ and $P$ were position and momentum operators in infinite dimension, they would satisfy the canonical commutation relation $[X,P ]=\iu \id$, where $\id$ is the identity operator. In a finite-dimensional system, however, an exact realization of the canonical commutation relations is impossible in finite dimensions, due to the Stone-von Neumann theorem~\cite{Stone_1930, Neumann_1931} which states that all Hilbert spaces supporting an irreducible representation of the canonical commutation relations must be isomorphic, and therefore infinite-dimensional (see also \cite{Hall_2013} for a modern textbook presentation).

An approximate version of the canonical commutation relation is satisfied by the DGW states, defined as follows. For a phase space point $(n_0,k_0)  \in  \mathbb{R}^2$, the DGW state centered at $(n_0,k_0)$ is       
\begin{equation}
    \ket{\Psi(n_0,k_0)} \propto \sum_{n  =  \lfloor n_0 \rfloor -m }^{ \lfloor n_0 \rfloor +m}  \,  e^{-\frac{\pi}{\sigma^2}(n-n_0)^2} e^{\iu\frac{2\pi }{d}k_0(n-n_0)}\ket{x_n}  \, ,
\end{equation}
up to a normalization factor.  Now, let $D_X(n)=m-|n|$ and $D_P(k)=m-|k|$ be the distances of a phase-space center from the position and momentum boundaries. We call a trajectory \emph{bulk} if every phase-space center $(n,k)$ encountered along it satisfies $|n|,|k|\leq (d-1)/4$ or equivalently $D_X(n),D_P(k)\geq (d-1)/4$. Along such a trajectory, the finite-dimensional dynamics reproduce infinite-dimensional phase-space displacements with exponentially small error. More precisely, for a phase-space vector $\gamma$, define $\Gamma = \sqrt{d/2\pi} \gamma$, then $G(\gamma)$ shifts a DGW state from $(n,k)$ to $(n+\Gamma_x, k+ \Gamma_p)$ and induces a phase $\phi(\gamma;n,k)$ up to boundary errors. This approximate displacement property yields the approximate Weyl relation below. The proof adapting techniques from~\cite{Woods_2019} is given in~Theorem~\ref{SI-thm:general_weyl_relation} of the Supplemental Material.

\begin{theorem}[{[\textit{Informal}]} Approximate Weyl relation for arbitrary phase-space displacements]
\label{thm:general_weyl_relation}
For a DGW state $\ket{\Psi(n_0,k_0)}$, suppose that the two phase-space trajectories generated by applying two phase-space vectors $\alpha$ and $\beta$ in the two possible orders both remain in the bulk. Then
\begin{equation}
\begin{aligned}
    & G(\beta)G(\alpha) \ket{\Psi(n_0,k_0)}  \\
    & = e^{\iu\omega(\alpha,\beta)}G(\alpha)G(\beta) \ket{\Psi(n_0,k_0)} + \ket{\delta},
\end{aligned}
\end{equation}
where $\omega(\alpha,\beta) =  \alpha_x\beta_p - \alpha_p\beta_x$ is the geometric phase equal to the phase-space area enclosed by the two trajectories, and $\ket{\delta}$ is the boundary error that satisfies $\|\ket{\delta}\|_2 = \mathcal{O}( \poly(d)e^{-\frac{\pi d}{16}} )$ if both the minimum position and momentum boundary distances $D_{X,*}$ and $D_{P,*}$ encountered along the two trajectories are at least $(d-1)/4$. 
\end{theorem}

The consequence of the finite-dimensional approximation of the infinite-dimensional Weyl relation is clear, as is shown in Figure~\ref{fig:phase_space}: the two causal orders of the same phase-space displacements bring a DGW state to two DGW states with the same final phase-space center but with a relative geometric phase. When the two causal orders are in quantum superposition, this relative phase can be read out by Ramsey interferometry, which will be essential to the metrological task below. The only price is the boundary term $\ket{\delta}$, which remains exponentially small until the accumulated displacement approaches the edge of the finite-dimensional phase space.

\begin{figure}[t]
\centering
\begin{tikzpicture}[
    scale=0.75,
    box black/.style={fill=black, text=white, minimum width=0.9cm, minimum height=0.8cm, align=center, font=\small, thick},
    box cyan/.style={fill=cyan!30, draw=black, thick, minimum width=0.7cm, minimum height=2.3cm, align=center, font=\small},
    mid arrow blue/.style={postaction={decorate,decoration={markings,mark=at position 0.55 with {\arrow[blue!50, line width=1.5pt]{Latex[length=2.5mm, width=2.5mm]}}}}, draw=blue!50, line width=1.5pt},
    mid arrow orange/.style={postaction={decorate,decoration={markings,mark=at position 0.55 with {\arrow[orange!50, line width=1.5pt]{Latex[length=2.5mm, width=2.5mm]}}}}, draw=orange!50, line width=1.5pt}
    ]

    \def\yTop{2}
    \def\yBot{0}
    \def\yCtrl{-1.2}
    \def\yMid{1} 
    \def\yOff{0.25}

    \path[fill=red!50, draw=black, thick, rounded corners=1pt] (-1, \yMid-0.4) -- (-1.5, \yMid-0.4) arc (270:90:0.4 and 0.4) -- (-1, \yMid+0.4) -- cycle;

    \path[fill=red!50, draw=black, thick, rounded corners=1pt] (-1, \yCtrl-0.4) -- (-1.5, \yCtrl-0.4) arc (270:90:0.4 and 0.4) -- (-1, \yCtrl+0.4) -- cycle;
    
    \node[draw=black, fill=white, circle, inner sep=1.5pt, thick] (c1) at (-0.5, \yMid) {};
    \node[draw=black, fill=black, circle, inner sep=1.5pt] (c2) at (-0.5, \yCtrl) {};
    \draw[thick] (c1) -- (c2);
    \draw[thick] (-1, \yMid) -- (c1);
    \draw[thick] (-1, \yCtrl) -- (c2);

    \node[box black] (dx1) at (1.5, \yTop) {$G(\lambda_\alpha \alpha_1)$};
    \node[box black] (dxN) at (4.5, \yTop) {$G(\lambda_\alpha \alpha_N)$};
    \node[box black] (dp1) at (1.5, \yBot) {$G(\lambda_\beta \beta_1)$};
    \node[box black] (dpN) at (4.5, \yBot) {$G(\lambda_\beta \beta_N)$};

    \node at (3, \yTop) {\textbf{\dots}};
    \node at (3, \yBot) {\textbf{\dots}};

    \node[draw=black, fill=white, circle, inner sep=1.5pt, thick] (c3) at (6.5, \yMid) {};
    \node[draw=black, fill=black, circle, inner sep=1.5pt] (c4) at (6.5, \yCtrl) {};
    \draw[thick] (c3) -- (c4);
    \draw[thick] (c2) -- (c4); 
    \draw[thick] (c3) -- (7, \yMid);
    \draw[thick] (c4) -- (7, \yCtrl);

    \draw [mid arrow blue] (c1) .. controls (0.0, \yMid) and (0.0, \yTop+\yOff) .. ($(dx1.west) + (0, \yOff)$);

    \draw [mid arrow blue] ($(dx1.east) + (0, \yOff)$) -- ($(dxN.west) + (0, \yOff)$);

    \draw [mid arrow blue, rounded corners=6pt] ($(dxN.east) + (0, \yOff)$) -- (5.85, \yTop+\yOff) -- (5.85, 1.3) -- (0.15, 1.3) -- (0.15, \yBot+\yOff) -- ($(dp1.west) + (0, \yOff)$);

    \draw [mid arrow blue] ($(dp1.east) + (0, \yOff)$) --  ($(dpN.west) + (0, \yOff)$);

    \draw [mid arrow blue] ($(dpN.east) + (0, \yOff)$) .. controls (6.1, \yBot+\yOff) and (6.1, \yMid) .. (c3);

    \draw [mid arrow orange] (c1) .. controls (0.0, \yMid) and (0.0, \yBot-\yOff) .. ($(dp1.west) + (0, -\yOff)$);

    \draw [mid arrow orange] ($(dp1.east) + (0, -\yOff)$) -- ($(dpN.west) + (0, -\yOff)$);

    \draw [mid arrow orange, rounded corners=6pt] ($(dpN.east) + (0, -\yOff)$) -- (5.95, \yBot-\yOff) -- (5.95, 0.7) -- (0.25, 0.7) -- (0.25, \yTop-\yOff) -- ($(dx1.west) + (0, -\yOff)$);

    \draw [mid arrow orange] ($(dx1.east) + (0, -\yOff)$) -- ($(dxN.west) + (0, -\yOff)$);

    \draw [mid arrow orange] ($(dxN.east) + (0, -\yOff)$) .. controls (6.1, \yTop-\yOff) and (6.1, \yMid) .. (c3);

    \path[fill=gray!30, draw=black, thick, pattern=north east lines, rounded corners=1pt] (7, \yCtrl-0.4) -- (7, \yCtrl+0.4) arc (90:-90:0.5 and 0.5) -- cycle;

    \node at (-1.5, 2.75) {(a)};

    \def\yTopB{-4}
    \def\yBotB{-5.5}

    \path[fill=red!50, draw=black, thick, rounded corners=1pt] (-1, \yBotB-0.4) -- (-1.5, \yBotB-0.4) arc (270:90:0.4 and 1.15) -- (-1, \yTopB+0.4) -- cycle;

    \draw[thick] (-1, \yTopB) -- (2.4, \yTopB);
    \draw[thick] (3.2, \yTopB) -- (6.95, \yTopB);
    \draw[thick] (-1, \yBotB) -- (2.4, \yBotB);
    \draw[thick] (3.2, \yBotB) -- (6.95, \yBotB);

    \node[box black] at (0.15, \yTopB) {$G(\lambda_{\alpha} \alpha_1)$};
    \node[box black] at (4.4, \yTopB) {$G(\lambda_{\beta} \beta_{N})$};

    \node[box cyan] at (1.75, -4.75) {$V_1$};
    \node[box cyan] at (6.15, -4.75) {$V_{2N}$};

    \node at (2.8, \yTopB) {\textbf{\dots}};
    \node at (2.8, \yBotB) {\textbf{\dots}};

    \path[fill=gray!30, draw=black, thick, pattern=north east lines, rounded corners=1pt] (6.95, \yBotB-0.4) -- (6.95, \yTopB+0.4) arc (90:-90:0.5 and 1.15) -- cycle;

    \node at (-1.5, -2.75) {(b)};
\end{tikzpicture}

\caption{Metrological strategies for phase-space displacements. (a) The strategy with indefinite causal order. A quantum switch routes the probe through two sets of displacements, $G(\lambda_\alpha \alpha_1), \ldots, G(\lambda_\alpha \alpha_N)$ and $G(\lambda_\beta \beta_1), \ldots, G(\lambda_\beta \beta_N)$, in a quantum superposition of two causal orders (blue and orange routines). The accumulated geometric phase is subsequently extracted via a Ramsey measurement on the control qubit. (b) The strategy with definite causal order. The $2N$ displacements are applied in a fixed causal order, interleaved with arbitrary adaptive unitaries $V_1, \ldots, V_{2N}$ prior to the final measurement. The adaptive unitaries are not assumed to preserve energy.
}

\label{fig:strategy}

\end{figure}
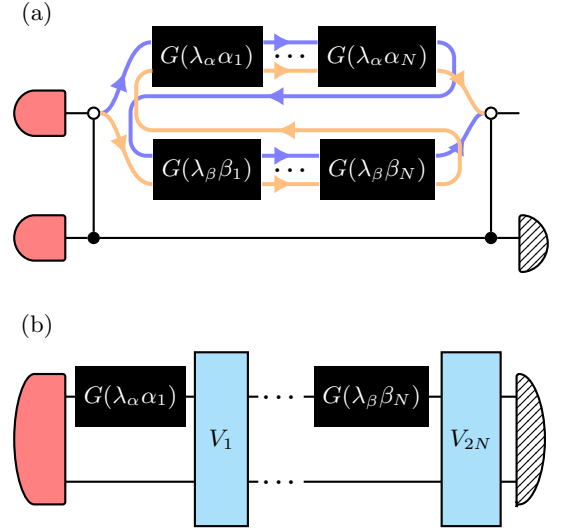

\textit{Metrological Task.}\m
We now formulate the finite-dimensional phase-space metrological task. An agent has tunable oracle access to two sets of $N$ unknown phase-space displacement oracles $G(\lambda_\gamma\gamma_j)$  where $\gamma=\alpha,\beta$. The unknown phase-space vectors $\gamma_j=(\gamma_{j,p},\gamma_{j,x})$, with $\gamma=\alpha,\beta$, have two nonnegative components, and an infinity norm $\|\gamma\|_\infty$ which lies in fixed intervals $[\gamma_{\min},\gamma_{\max}]$ bounded away from zero independent of $N$ and $d$. The attenuation factors $\lambda_\alpha,\lambda_\beta$ are known control parameters chosen by the agent, with values in $(0,1]$. Each call to an attenuated oracle counts as one use of the corresponding physical sample, independently of the attenuation strength. All attenuated stages automatically satisfy the same bulk condition once the final unattenuated stage does. The attenuation factors are only introduced to perform phase unwrapping used in Ramsey experiments, see~Section~\ref{SI-sec:metrology_bounds} of the Supplemental Material.  The target parameter is the geometric phase given by the symplectic product $\bar{\omega}(\bar{\alpha},\bar{\beta})= \bar{\alpha}_x\bar{\beta}_p -\bar{\alpha}_p\bar{\beta}_x $ of the two average displacements $\bar{\gamma}=\frac{1}{N}\sum_{j=1}^N\gamma_j$ where $\gamma=\alpha,\beta$. We assume a coarse prior interval $[\omega_-,\omega_+]$ of constant width. The probe energy $E =\bra{\psi}H\ket{\psi}$ for a probe state $\ket{\psi}$ is measured with respect to the finite-dimensional oscillator Hamiltonian $ H= X^2 + P^2$.

Figure~\ref{fig:strategy} illustrates the two types of  strategies  considered in this work. On the one hand, we consider  an ICO strategy where a quantum switch coherently routes the same probe through the two sets of oracles in the two opposite causal orders. On the other hand, we consider all  possible DCO strategies, in which the same oracle calls are arranged in a fixed causal order, possibly interleaved with arbitrary known unitaries, ancillas, adaptive controls, and a final joint measurement. In this case, no energy-preservation constraint is imposed on the adaptive controls. In particular, they may inject energy after the first oracle call. All preceding processing are included in probe state preparation. The two classes have access to the same oracle samples and the same non-oracle quantum processing. The only resource present in the ICO strategy and absent in the DCO strategy is the coherent superposition of causal orders. 

We now establish an unbounded advantage of the ICO strategy over all DCO strategies.  In general, metrological advantages  can take place in different ways: either as a reduced number of queries to an unknown quantum process, as in the case of \cite{Liu_2023,Kurdzialek_2023,Kong_2026,Wei_2026}, or as a more efficient use of physical resources, such as energy or time. 
Probe energy, in particular, is the prototypical  physical resource in quantum metrology, as illustrated in the famous applications of NOON states~\cite{Lee_2002,Slussarenko_2017} and squeezed states of light~\cite{Caves_1981,Pezze_2008}. Recently, a bounded energy advantage of ICO over DCO was established in a channel estimation problem when accounting for the energy costs of probe preparation, intermediate controls, and final measurement~\cite{Chen_2026o}. In the following, instead, we will establish an unbounded advantage that holds independently of how much energy the DCO strategy pumps into in the intermediate controls and final measurement.

For the ICO strategy, we choose the DGW state $\ket{\Psi(0,0)}$ as the probe state $\ket{\psi_{\rm ICO}}$ with the energy $E_{\rm ICO}=1+ \mathcal{O}(\poly(d)e^{-\frac{\pi d}{4}})$.  The quantum switch applies the two aggregate displacements $\prod_{j=1}^N G(\lambda_\alpha\alpha_j)$ and $\prod_{j=1}^N G(\lambda_\beta\beta_j)$ in a coherent superposition of the two possible causal orders. As long as the centers encountered during the evolution remain far from the boundary, the approximate Weyl relation implies that the two ordered branches bring the probe to the same final DGW wavepacket with a relative geometric phase up to an exponentially small boundary error. When both components of the phase-space displacement vectors are bounded by $x_{\max}$ and $p_{\max}$ respectively, a sufficient condition ensuring that the center $(n,k)$ during the trajectory remains within the bulk $|n|,|k|\leq (d-1)/4$ is
\begin{equation}
    d \geq \frac{8}{\pi} \max\big\{|x_{\max}|^2,|p_{\max}|^2\big\}N^2 .
\end{equation}
In this regime, the approximate Weyl relation gives 
\begin{equation}
\begin{aligned}
    & \prod_{j=1}^N G(\lambda_\beta\beta_j)  \prod_{k=1}^N G(\lambda_\alpha\alpha_k)  \ket{\Psi(0,0)}  -   \ket{\delta} \\
    & =  e^{\iu\lambda_\alpha\lambda_\beta N^2 \bar{\omega}}\prod_{k=1}^N G(\lambda_\alpha\alpha_k)   \prod_{j=1}^N G(\lambda_\beta\beta_j) \ket{\Psi(0,0)} ,
\end{aligned}
\end{equation}
where the boundary error is $\|\ket{\delta}\|_2 = \mathcal{O}(\poly(d)e^{-\pi d/16})$ and the target geometric phase is $\bar{\omega} = \bar{\alpha}_x\bar{\beta}_p-\bar{\alpha}_p\bar{\beta}_x  $. The control qubit therefore carries a Ramsey signal whose gain is enhanced by $N^2$, and this signal can be learned by Ramsey interferometry, see Section~\ref{SI-sec:metrology_with_qswitch} of the Supplemental Material. For total number of shots up to $\nu=\mathcal{O}(e^{\pi d/16}/\poly(d))$, the statistical estimation error of the Ramsey experiment dominates the DGW approximation error, leading to a MSE of $\mathcal{O}(\frac{1}{\nu N^4})$.

It worth noting that for each fixed pair $(N,d)$, the DGW approximation contributes a fixed boundary error, determined by the minimum distance to the boundary in the trajectory. Taking the number of shots $\nu$ to infinity reduces the statistical error while leaves the DGW approximation error unchanged. Thus the $\mathcal{O}(\frac{1}{\nu N^4})$ scaling is observed in the finite-sample window where the statistical error dominates the exponentially small DGW boundary error. In our bulk regime, this window extends up to total number of shots of order $\mathcal{O}(e^{\pi d/16}/\poly(d))$.

By contrast, a DCO strategy cannot directly convert noncommutativity into a single observable relative geometric phase. To see the limitation, consider two sets of conjugate phase-space displacements $\alpha_j=(p_j,0)$ and $\beta_j=(0,x_j)$ for which the target geometric phase is $\bar{\omega}=-\bar{x}\bar{p}$. The DCO strategy accesses these $2N$ displacements in the same fixed causal order in every shot.  We lower bound the DCO error by a side information argument. Suppose that all displacements are known except the actual first oracle called by the DCO protocol (we assume this to be $x_1$; other $x$ and $p$ parameters work similarly). This simplifies the estimation problem, and any lower bound for the simplified problem also applies to the original DCO task. The parameter $x_1$ is encoded through $e^{-\iu x_1 P}$, whose generator is $P$.  For a probe state $\ket{\psi_{\rm DCO}}$ with a probe energy $E_{\rm DCO}$, the quantum Fisher information satisfies $F_Q(x_1) = 4\Var_{\psi_{\rm DCO}}(P) \leq 4E_{\rm DCO}$ where we used $\Var(P)\leq \langle P^2\rangle\leq \langle H\rangle$. By the quantum Cram\'er-Rao bound, given total number of shots $\nu$, the MSE for $x_1$ is $\Omega(\frac{1}{\nu E_{\rm DCO}})$. As $\partial_{x_1} \bar{\omega} = -\bar{p}/N$, propagating this error to $\bar{\omega}$ yields a MSE of $\Omega(\frac{1}{\nu E_{\rm DCO} N^2})$. Here $E_{\rm DCO}$ denotes the initial energy of the probe state entering the first unknown displacement oracle. Any pre-processing before this oracle is regarded as part of the probe preparation. This lower bound thus depends only on the initial energy of the probe state, regardless of whether more energy is injected into the system by adaptive control. A detailed proof for the lower bound covering both unbiased and biased estimators is presented in Section~\ref{SI-sec:metrology_without_ico} of the Supplemental Material.

Combining the achievability result of an ICO strategy and the impossibility result of a DCO strategy gives the unbounded separation between strategies with ICO and DCO in the following sense: 
\begin{theorem}[{[\textit{Informal}]} Unbounded separation across finite-dimensional instances]
    \label{thm:finite_dimensional_separation}
    For every prescribed factor $R$, there is a finite-dimensional phase-space metrological task with finite $N$ as well as dimension $d=\Omega(N^2)$ such that for the same total number of shots within the finite-sample window $\nu=\mathcal{O}(e^{\pi d/16}/\poly(d))$, an ICO strategy reaches the mean squared error $\MSE=O(\frac{1}{\nu N^4})$  with a constant initial probe energy $E_{\rm ICO}$, whereas every DCO strategy with the same fixed causal order in every shot reaching the same mean squared error requires a larger initial probe energy $E_{\rm DCO}$ by a factor at least $R$: 
    \begin{align}
        \frac{E_{\rm DCO}}{E_{\rm ICO}} \geq R. 
    \end{align}
\end{theorem}

For details, see~Theorem~\ref{SI-thm:finite_dimensional_separation} of the Supplemental Material. Theorem~\ref{thm:finite_dimensional_separation} establishes an unbounded separation in the probe energy between ICO and DCO strategies across a family of finite-dimensional metrological tasks. 


The finite-sample separation in Theorem~\ref{thm:finite_dimensional_separation} is compatible with fixed-dimensional asymptotic obstructions such as~\cite[Appendix~\rom{5}]{Kong_2026}. These obstructions do not imply that ICO can never be useful in finite-dimensional metrology. Rather, the relevant statement is that, at fixed finite dimensions, the variance of the local generator is bounded by finite-dimensional spectral diameter, preventing strategies with quantum switches from sustaining the unbounded asymptotic separation in infinite dimensions.

To make the connection with the present metrological task explicit, we prove the following energy-resolved version of the fixed-dimensional bound. Let the energy be measured with respect to $ H= X^2+ P^2$. We focus again on two sets of $N$ conjugate phase space displacements $\alpha_j=(p_j,0)$ and $\beta_j=(0,x_j)$ where $x_j$ and $p_j$ bounded away from zero for which the target geometric phase is $\bar{\omega}=-\bar{x}\bar{p}$. We allow an arbitrary strategy with quantum switches which coherently controls different causal orders of these $2N$ unknown displacement oracles together with arbitrary known interleaving operations, ancillas and final measurements. The argument follows this routine: differentiating the output state with respect to the average displacement inserts one conjugated generator at each oracle slot; the variance of each inserted generator is bounded by the energy of the state entering that slot; summing over the $N$ slots gives an $N^2$ factor. This leads to the following proposition, whose details are in~Proposition~\ref{SI-prop:lower_bound_with_ico} of the Supplemental Material:

\begin{proposition}[{[\textit{Informal}]} Energy-resolved fixed-dimensional bound]
\label{prop:lower_bound_with_ico}
Let $\pi$ be a permutation which specifies a causal order of displacement oracle queried in each step, used to denote a causal order branch. For each causal order $\pi$ and each step $i$, let $E_{\pi,i}$ be the intermediate energy of the intermediate state entering the $i$th oracle, and define $E_{\rm int}:=\max_{\pi,i}E_{\pi,i}$. Then every locally unbiased estimator of $\bar{\omega}$ satisfies $\MSE = \Omega(\frac{1}{\nu N^2 E_{\rm int}})$. 
\end{proposition}

At fixed dimension, Proposition~\ref{prop:lower_bound_with_ico} recovers the spectral diameter obstruction in~\cite{Kong_2026}. Since $ H= X^2+ P^2$ has operator norm $\|H\|_\infty=O(d)$ in our normalization, every strategy satisfies $E_{\rm int}=O(d)$, yielding $\MSE=\Omega(\frac{1}{\nu N^2 d})$. This shows why Theorem~\ref{thm:finite_dimensional_separation} does not contradict Proposition~\ref{prop:lower_bound_with_ico}. Theorem~\ref{thm:finite_dimensional_separation} concerns the \emph{initial probe energy} required by ICO and DCO strategies in the low energy and finite sample regime, whereas Proposition~\ref{prop:lower_bound_with_ico} constrains strategies according to the \emph{largest intermediate energy} encountered during the entire process. 

In our ICO protocol, the initial probe state has constant energy, but the switched trajectories displace the wavepacket center by an amount proportional to $N$. Since the oscillator energy is quadratic in the phase-space displacement, the intermediate energy scales as $E_{\rm int}=O(N^2)$ throughout the bulk regime. Substituting this into Proposition~\ref{prop:lower_bound_with_ico} gives the lower bound $\Omega(\frac{1}{\nu N^4})$, matching the scaling achieved by our ICO protocol in the finite sample regime. Therefore, the two results correspond to different resource accountings: Theorem~\ref{thm:finite_dimensional_separation} demonstrates a low energy and finite sample separation, while Proposition~\ref{prop:lower_bound_with_ico} governs the asymptotic behavior through the maximal intermediate energy, or equivalently the finite spectral diameter.

\textit{Discussion.}\m
We have established an unbounded separation between metrological strategies with ICO and DCO for estimating finite-dimensional phase-space displacements. The separation occurs in the finite-sample regime: the ICO protocol uses a probe state close to the finite-dimensional vacuum, while every DCO strategy achieving the same precision with the same total number of shots requires larger probe energy.  High-energy states are typically costly to prepare, whereas the near vacuum states can be more easily prepared by dissipative dynamics~\cite{Paz_2004, Cormick_2013, Wang_2025, Rubies-Bigorda_2025, Zhan_2026} or polynomial-depth circuits~\cite{Rattew_2021,Huggins_2025, Chen_2025,Xie_2025}.

The finite-dimensional phase space considered here also has natural physical realizations. The DGW framework can be embedded into discrete or quasi continuous variable architectures~\cite{Kitaev_2008,Opatrny_2017}.  In particular, the Hamiltonian $H=X^2+P^2$ is well approximated in the low-energy sector of collective-spin oscillators~\cite{Hakioglu_2000}, which perfectly coincides with our required bulk condition. Candidate platforms include atomic ensembles~\cite{Opatrny_2017,Opatrny_2023} and nanophotonic waveguide arrays~\cite{Weimann_2016}.

In addition, the tunable attenuation assumed in the oracle model corresponds to calibrating the interaction time or coupling strength before the displacement is applied. These can be realized in optomechanical~\cite{Law_1995,Aspelmeyer_2014}, electro-optic~\cite{Tsang_2010,Tsang_2011,Ma_2011,Ma_2012,Chan_2013}, and superconducting systems~\cite{Makhlin_1999,Blais_2003,Averin_2003,Peropadre_2013,Baust_2015} through standard tunable-coupling mechanisms.

In the End Matter, we numerically verify the predicted scalings of the mean squared error and the initial probe energy ratio for ICO and DCO strategies. We also propose possible experimental realizations of the metrological task.

\textit{Acknowledgments.}\m
This work is supported by the National Natural Science Foundation of China via the Excellent Young Scientists Fund (Hong Kong and Macau) Project 12322516, the National Natural Science Foundation of China (NSFC)/Research Grants Council (RGC) Joint Research Scheme via Project N\_HKU7107/24, the Guangdong Provincial Quantum Science Strategic Initiative via Project GDZX2503001, and the Hong Kong
Research Grant Council (RGC) through the General Research Fund (GRF) Grant 17302724.  G.C. acknowledges support from the Chinese Ministry of Science and Technology (MOST) through grant 2023ZD0300600, from the Hong Kong Research Grants Council (RGC) through grants 17310725 and SRFS2021-7S02, and from the State Key Laboratory of Quantum Information Technologies and Materials, Chinese University of Hong Kong. Research at the Perimeter Institute is supported by the Government of Canada through the Department of Innovation, Science and Economic Development Canada and by the Province of Ontario through the Ministry of Research, Innovation and Science.

Y.H. and Z.-S. Li contributed equally and share first authorship.

\bibliographystyle{ultimate}
\bibliography{references}

\section*{End Matter}

\textit{Physical realization.}\m
The DGW framework describes a Gaussian wavepacket over $d$ discrete position states, with the discrete Fourier basis playing the role of momentum. A closely related $d$-dimensional finite oscillator is provided by the Kravchuk construction~\cite{Hakioglu_2000}. In the low-excitation, large-$d$ regime, its wavefunctions approach the Hermite--Gaussian eigenfunctions of the harmonic oscillator, while its appropriately rescaled position and momentum operators approach conjugate quadratures. Accordingly, DGW states and the quadratic Hamiltonian $H=X^2+P^2$ locally approximate their harmonic-oscillator counterparts. The bulk condition specifies the regime of validity by requiring the wavepacket to remain localized away from the edges of the phase-space chart and hence insensitive to periodic wraparound. A possible physical realization uses a large atomic ensemble polarized near a spin-coherent state, where two collective-spin components transverse to the mean polarization act, after rescaling, as oscillator position and momentum~\cite{Opatrny_2017,Opatrny_2023}. Closely related discrete fractional Fourier transformations have also been demonstrated in engineered arrays of coupled optical waveguides~\cite{Weimann_2016}. Alternatively, on a gate-based quantum processor, the Gaussian amplitude profile of a DGW state can be prepared efficiently in a qubit register~\cite{Kitaev_2008,Rattew_2021,Xie_2025}. The detailed correspondence between the DGW framework and the Kravchuk construction is given in Section~\ref{SI-sec:DGW_realization} of the Supplemental Material.

\begin{figure}[t]
    \centering
    \begin{overpic}[width=0.75\linewidth]{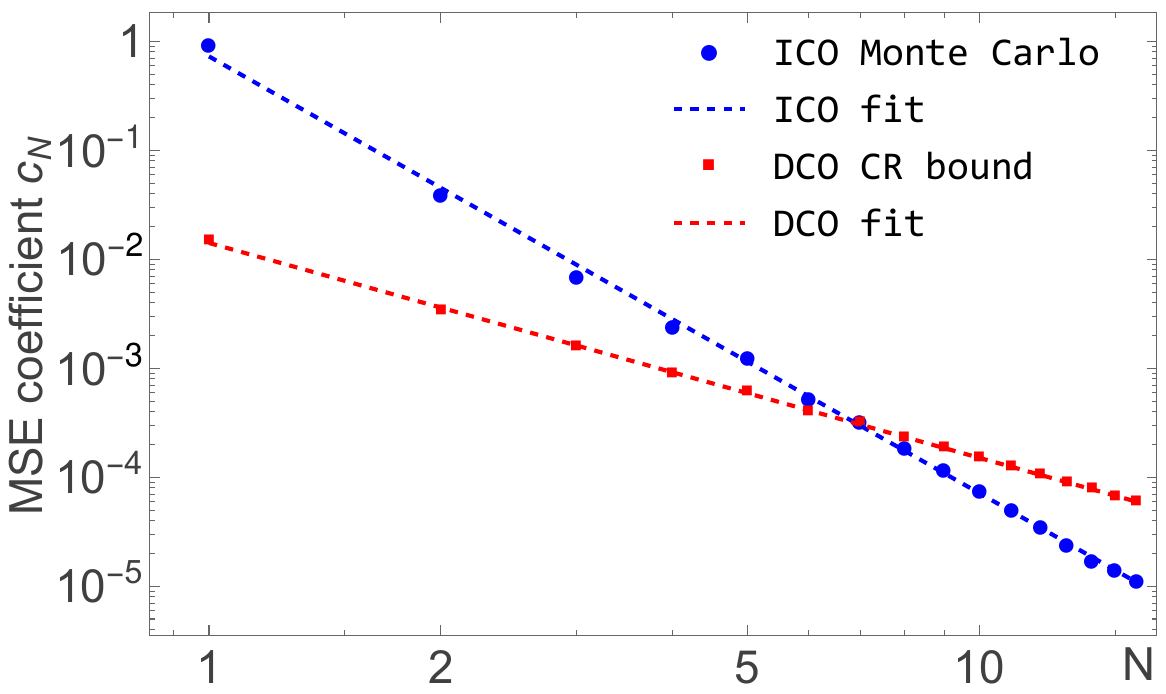}
        \put(0, 55){\textbf{(a)}}
    \end{overpic}

    \vspace{1em}

    \begin{overpic}[width=0.75\linewidth]{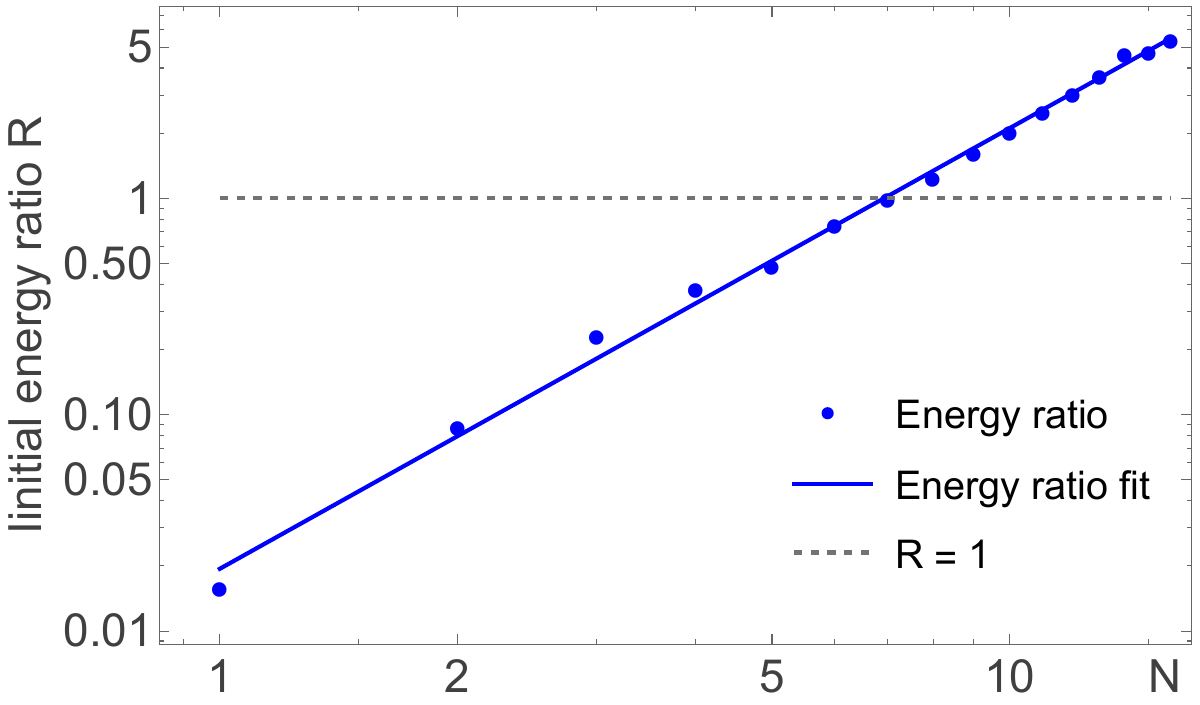}
        \put(0, 55){\textbf{(b)}}
    \end{overpic}

    \caption{\textbf{(a)} Scaling of the finite-sample MSE coefficient $c_N$, defined by $\MSE\simeq c_N/\nu$, with respect to $N$ for $d\simeq N^2$. For each $N$, the blue circle represents the ICO coefficient, extracted by fitting the simulated $\MSE$ against $\nu^{-1}$ and averaging over randomly generated sets of displacements in the range $[0.4,0.6]$. The red square represents the DCO coefficient, obtained by computing the Cram\'er--Rao lower bound at the same initial probe energy and averaging over the same sets of displacements  as the ICO strategy. The dashed lines are log--log power-law fits. The fitted slopes are $-4.0087$ for ICO and $-1.9724$ for DCO, and the fitted Pearson correlation coefficients are $-0.99928$ for ICO and $-0.99985$ for DCO, consistent with the $N^{-4}$ scaling for the ICO MSE and the $N^{-2}$ scaling for the DCO MSE, respectively. \textbf{(b)} Scaling of the ratio $R$, between the initial probe energies required for the DCO and ICO strategies to achieve the same mean squared error, with respect to $N$ for $d\simeq N^2$. Using the same parameters as in (a), we determine $R$ by first extracting the ICO coefficient, then matching the DCO coefficient by tuning the initial probe energy in the Cram\'er--Rao lower bound, and finally taking the ratio between the initial probe energies for the DCO and ICO strategies. The blue circles denote the ratio $R$, and the blue line is the log--log power-law fit. The fitted slope is $2.0363$ and the fitted Pearson correlation coefficient is $0.99797$, well reflecting the desired $N^2$ scaling. The dashed line indicates the reference ratio $R=1$. The fitted curves in both (a) and (b) exhibit a crossover near $N\simeq 7$, after which the ICO strategy outperforms the DCO strategy.}
    \label{fig:N_scaling}
\end{figure}

The attenuated oracle calls used in our estimator correspond operationally to applying the unknown displacement interaction for a shorter duration or with a weaker calibrated coupling. Thus, if the full interaction implements $G(\gamma)$, applying a known fraction $0<\lambda\leq1$ of it implements $G(\lambda\gamma)$ without requiring knowledge of $\gamma$. This assumes controllable access to the underlying interaction rather than only to an indivisible black-box unitary. Attenuation is used solely for phase unwrapping. Because the high-gain Ramsey fringe determines the amplified phase only modulo $2\pi$, the low-gain stages enlarge the fringe period and identify the correct phase branch before the final unattenuated stage provides the full local sensitivity. The required interaction control is available in several analog platforms. In cavity optomechanics, radiation pressure couples an optical cavity mode to a mechanical mode~\cite{Law_1995,Aspelmeyer_2014}, while external modulation of the cavity parameters or intermode coupling provides control over the effective interaction~\cite{Liao_2015,Liao_2016,Wang_2019}. In electro-optic systems, a driven nonlinear medium produces a beam-splitter-like interaction between electromagnetic modes~\cite{Tsang_2010,Tsang_2011}, whose mixing strength can be controlled through the optical drive and voltage applied to electro-optic modulators~\cite{Ma_2011,Ma_2012,Chan_2013}. Controllable interactions are likewise available in superconducting-circuit architectures~\cite{Makhlin_1999,Blais_2003,Averin_2003} where Josephson-junction or SQUID couplers allow the coupling between microwave resonators to be varied by an external magnetic flux~\cite{Peropadre_2013}, as demonstrated experimentally using a flux-qubit-mediated tunable coupler~\cite{Baust_2015}. Further platform-specific details are given in ~Section~\ref{SI-sec:attenuation} of the Supplemental Material.

\textit{Numerical verification.}\m
We numerically test the predicted finite-sample scaling using the exact finite-dimensional DGW dynamics for $N=1,\ldots,16$ and $d\simeq N^2$. For each $N$, we extract the ICO coefficient $c_N$, defined by $\MSE = c_N/\nu$, via a Monte Carlo method by fitting the $\MSE$ with respect to $\nu$. For comparison, we use a lower bound estimate of the DCO coefficient obtained by the Cram\'er--Rao bound at the same initial probe energy.  As shown in Fig.~\ref{fig:N_scaling}(a), we plot the resulting ICO coefficient versus $N$, scaling as $N^{-4.0087}$, as well as the resulting DCO coefficient versus $N$, scaling as $N^{-1.9724}$. These numerical results are in close agreement with the predicted $N^{-4}$ and $N^{-2}$ behaviors. Fig.~\ref{fig:N_scaling}(b) illustrates this coefficient separation in terms of the ratio $R_N$ between the initial probe energies required for the DCO and ICO strategies, which is extracted by tuning the initial probe energy so that the DCO coefficient matches the ICO coefficient. The fitted relation between $R_N$ and $N$ scales as $0.01941 N^{2.0363}$, crossing the equal-energy threshold near $N\simeq 7$ and subsequently grows approximately quadratically. These results numerically support the predicted increasing initial-energy separation along the family of finite-dimensional instances with $d\simeq N^2$.

\clearpage
\onecolumngrid

\arxivSupplementToctrue

\setcounter{section}{0}
\setcounter{subsection}{0}
\setcounter{subsubsection}{0}
\setcounter{equation}{0}
\setcounter{figure}{0}
\setcounter{table}{0}
\setcounter{theorem}{0}
\setcounter{definition}{0}
\setcounter{remark}{0}
\setcounter{algocf}{0}
\setcounter{secnumdepth}{2}

\renewcommand{\thesection}{S\arabic{section}}
\renewcommand{\theequation}{S\arabic{equation}}
\renewcommand{\thefigure}{S\arabic{figure}}
\renewcommand{\thetable}{S\arabic{table}}
\renewcommand{\thetheorem}{S\arabic{theorem}}
\renewcommand{\thealgocf}{S\arabic{algocf}}
\renewcommand{\thesubsection}{\thesection.\arabic{subsection}}

\renewcommand{\theHsection}{SI.\arabic{section}}
\renewcommand{\theHequation}{SI.\arabic{equation}}
\renewcommand{\theHfigure}{SI.\arabic{figure}}
\renewcommand{\theHtable}{SI.\arabic{table}}
\renewcommand{\theHtheorem}{SI.\arabic{theorem}}
\renewcommand{\theHdefinition}{SI.\arabic{definition}}
\renewcommand{\theHremark}{SI.\arabic{remark}}
\renewcommand{\theHsubsection}{\theHsection.\arabic{subsection}}
\renewcommand{\theHsubsubsection}{\theHsubsection.\arabic{subsubsection}}
\makeatletter
\renewcommand{\p@subsection}{}
\makeatother

\phantomsection
\pdfbookmark[0]{Supplemental Material}{supplemental-material}
\begin{center}
    {\large\bfseries Supplemental Material for}\\[0.5em]
    {\large\bfseries Unbounded separation between definite and indefinite causal order in finite-dimensional quantum metrology}
\end{center}

\tableofcontents

\section{Discrete Gaussian Wavepacket Framework}
\label{SI-sec:DGW_framework}

\subsection{Definition}
The DGW framework was initially proposed for quasi-ideal quantum clocks~\cite{Woods_2019}. Now we find its application in quantum metrology with super-Heisenberg scalings. 
\begin{definition}[Discrete Gaussian Wavepacket Framework]
\label{SI-def:DGW_framework}
    We consider a discrete phase space of odd dimension $d = 2m+1$. The normalized momentum and position operators are defined as:
    \begin{align}
        \hat{P} = \sqrt{\frac{2\pi}{d}} \sum_{k=-m}^m k \ketbra{p_k}{p_k}, \quad \hat{X} = \sqrt{\frac{2\pi}{d}} \sum_{n=-m}^m n \ketbra{x_n}{x_n}, 
    \end{align}
    with the discrete Fourier transform relating the mutually unbiased bases: 
    \begin{align}
        \ket{p_k}=\frac{1}{\sqrt{d}}\sum_{n=-m}^{m} e^{\iu \frac{2\pi}{d} n k} \ket{x_n}. 
    \end{align}
    The Hamiltonian is given by 
    \begin{align}
        \hat{H} = \hat{X}^2 + \hat{P}^2. 
    \end{align}
    The DGW state with width $\sigma \approx \sqrt{d}$ centered at $(n_0, k_0)$ is:
    \begin{align}
        \ket{\Psi(n_0,k_0)} = A \sum_{n\in \mathcal{S}_d(n_0)} e^{-\frac{\pi}{\sigma^2}(n-n_0)^2} e^{\iu\frac{2\pi }{d}k_0(n-n_0)}\ket{x_n},
    \end{align}
    where $\mathcal{S}_d(n_0) =\lbrace \lfloor n_0 \rfloor -m, \ldots, \lfloor n_0 \rfloor + m \rbrace$ is the physical support window and $A$ normalizes the state. 
\end{definition}

To quantify the breakdown of the approximation, we define the absolute distances of the state's position and momentum centers from the periodic boundaries $\pm m$:
\begin{definition}
    The absolute distances of the state's position and momentum centers from the boundaries $\pm m$ is 
    \begin{align}
        D_X(n_0) = m - |n_0|, \quad D_P(k_0) = m  - |k_0|.
    \end{align}
\end{definition}


\subsection{Normalization}

As a remark, the normalization factors $A$ for DGW states are approximately the same: 
\begin{lemma}[Normalization factors for adjacent wavepackets~{\cite[Theorem E.1]{Woods_2019}}]
    \label{SI-lem:close_norm}
    Let $\sigma^2 \approx d$. Let $\lfloor n_0\rfloor =\lfloor n_0'\rfloor$. Let normalization factors for $\ket{\Psi(n_0,k_0)}$ and $\ket{\Psi(n_0',k_0)}$ be, respectively,
    \begin{align}
        A(n_0) = \Big(\sum_{n\in\mathcal{S}_d(n_0)}|\psi(n_0,k_0;n)|^2 \Big)^{-\frac{1}{2}}, \quad A(n_0') = \Big(\sum_{n\in\mathcal{S}_d(n_0')}|\psi(n_0',k_0;n)|^2 \Big)^{-\frac{1}{2}}, 
    \end{align}
    We have 
    \begin{align}
        |A(n_0)-A(n_0')| \leq \mathcal{O}\big(\sigma^{-\frac{1}{2}} e^{-\frac{\pi \sigma^2}{2}} \big)+ \mathcal{O}\big(\sigma^{-\frac{3}{2}}e^{-\frac{\pi d^2}{2\sigma^2}}\big).
    \end{align}
\end{lemma}

\begin{proof}
    Since $\lfloor n_0 \rfloor = \lfloor n_0' \rfloor$, the discrete support sets defining the states are perfectly identical $\mathcal{S}_d(n_0) = \mathcal{S}_d(n_0') \equiv \mathcal{S}$. We can therefore write $g(x)$ for any $x$ sharing this floor as:
    \begin{align}
        g(x) = \sum_{n\in\mathcal{S}}|\psi(x,k_0;n)|^2 =\sum_{n \in \mathcal{S}} e^{-\frac{2\pi}{\sigma^2}(n-x)^2}.
    \end{align}
    Thus $A(n_0) = g(n_0)^{-1/2}$ and $A(n_0') = g(n_0')^{-1/2}$. We split this finite sum into an infinite sum over $\mathbb{Z}$ and a truncation tail:
    \begin{align}
        g(x) = \sum_{n=-\infty}^\infty e^{-\frac{2\pi}{\sigma^2}(n-x)^2} - \sum_{n \notin \mathcal{S}} e^{-\frac{2\pi}{\sigma^2}(n-x)^2}.
    \end{align}
    
    \noindent\textbf{Step 1: Bounding the truncation tail.}
    
    Let $m = (d-1)/2$. The set $\mathcal{S}$ is centered at $\lfloor n_0 \rfloor$ and spans $[\lfloor n_0 \rfloor - m, \dots, \lfloor n_0 \rfloor + m]$. 
    For any $n \notin \mathcal{S}$, the distance to the integer center is $|n - \lfloor n_0 \rfloor| \geq m + 1$.
    Because $x$ shares the same floor $\lfloor x \rfloor = \lfloor n_0 \rfloor$, its fractional part satisfies $x - \lfloor n_0 \rfloor \in [0, 1)$. 
    By the triangle inequality, the physical distance from the omitted lattice points to the continuous center $x$ is strictly bounded: 
    \begin{align}
        |n - x| \ge |n - \lfloor n_0 \rfloor| - |x - \lfloor n_0 \rfloor| \geq  m + 1 - 1 = m.
    \end{align}
    Thus, all omitted terms strictly satisfy $|n - x| \ge m$. 
    We bound the truncation error by a standard Gaussian tail sum:
    \begin{align}
        \epsilon_{\rm tail}(x) = \sum_{n \notin \mathcal{S}} e^{-\frac{2\pi}{\sigma^2}(n-x)^2} \leq 2 \sum_{n=m}^\infty e^{-\frac{2\pi}{\sigma^2} n^2} \leq 2 e^{-\frac{2\pi}{\sigma^2} m^2}  \sum_{k=0}^\infty e^{-\frac{4\pi}{\sigma^2} mk} \leq \frac{2 e^{-\frac{2\pi}{\sigma^2} m^2}}{1 - e^{-\frac{4\pi m}{\sigma^2}}} \leq O\left( e^{-\frac{\pi d^2}{2\sigma^2}} \right).
    \end{align}
    
    \noindent\textbf{Step 2: Evaluating the infinite sum via Poisson summation.}
    
    We apply the Poisson Summation Formula to the infinite sum:
    \begin{align}
        \sum_{n=-\infty}^\infty e^{-\frac{2\pi}{\sigma^2}(n-x)^2} = \frac{\sigma}{\sqrt{2}} \sum_{s=-\infty}^\infty e^{-\frac{\pi \sigma^2}{2} s^2} e^{\iu 2\pi s x}.
    \end{align}
    We separate the $s=0$ term from the higher-order aliasing terms ($s \neq 0$):
    \begin{align}
        \epsilon_{\mathrm{alias}}(x) = \frac{\sigma}{\sqrt{2}} \sum_{s\neq 0}e^{-\frac{\pi \sigma^2}{2} s^2} e^{\iu 2\pi s x}. 
    \end{align}
    The $s=0$ term yields exactly $\sigma/\sqrt{2}$, which is a constant strictly independent of $x$. The aliasing error is bounded uniformly:
    \begin{align}
        |\epsilon_{\mathrm{alias}}(x)| \leq \sqrt{2} \sigma \sum_{s=1}^\infty e^{-\frac{\pi \sigma^2}{2} s^2} \leq \sqrt{2} \sigma e^{-\frac{\pi \sigma^2}{2}} \sum_{s=0}^\infty e^{-\pi \sigma^2 s}  \leq \sqrt{2} \sigma \frac{e^{-\frac{\pi \sigma^2}{2}}}{1 - e^{-\pi \sigma^2}} \leq \mathcal{O}\big( \sigma e^{-\frac{\pi \sigma^2}{2}} \big).
    \end{align}
    Therefore, 
    \begin{align}
        \sum_{n=-\infty}^\infty e^{-\frac{2\pi}{\sigma^2}(n-x)^2} = \frac{\sigma}{\sqrt{2}} + \epsilon_{\mathrm{alias}}(x). 
    \end{align}

    \noindent\textbf{Step 3: Bounding the normalization difference.} \\
    Combining the bounds, for both $n_0$ and $n_0'$, the sum evaluates to:
    \begin{align}
        g(n_0) = \frac{\sigma}{\sqrt{2}} + \epsilon(n_0), \quad g(n_0') = \frac{\sigma}{\sqrt{2}} + \epsilon(n_0'),
    \end{align}
    where the total error is bounded by
    \begin{align}
        |\epsilon(x)| \leq |\epsilon_{\mathrm{tail}}(x)| + |\epsilon_{\mathrm{alias}}(x)| = \mathcal{O}\big(e^{-\frac{\pi d^2}{2\sigma^2}}\big) + \mathcal{O}\big(\sigma e^{-\frac{\pi \sigma^2}{2}}\big). 
    \end{align}    
    The difference in the sums is bounded via the triangle inequality. The constant term $\sigma/\sqrt{2}$ perfectly cancels:
    \begin{align}
        |g(n_0) - g(n_0')| = |\epsilon(n_0) - \epsilon(m_0)| \leq   |\epsilon(n_0)|+|\epsilon(m_0)|\leq \mathcal{O}\big(e^{-\frac{\pi d^2}{2\sigma^2}}\big) + \mathcal{O}\big(\sigma e^{-\frac{\pi \sigma^2}{2}}\big). 
    \end{align}
    We evaluate the absolute difference between the normalization factors using the algebraic identity:
    \begin{align}
        |A(n_0) - A(n_0')| = \left| \frac{1}{\sqrt{g(n_0)}} - \frac{1}{\sqrt{g(n_0')}} \right| = \frac{|g(n_0)- g(n_0')|}{\sqrt{g(n_0) g(n_0')} (\sqrt{g(n_0)} + \sqrt{g(n_0')})}.
    \end{align}
    Because the continuous integral term $\sigma/\sqrt{2}$ dominates $e(x)$, we get
    \begin{align}
        |A(n_0) - A(n_0')| \leq 2^{\frac{3}{4}} \sigma^{-\frac{3}{2}} |\epsilon(x)|  \leq \mathcal{O}\big(\sigma^{-\frac{1}{2}} e^{-\frac{\pi \sigma^2}{2}} \big)+ \mathcal{O}\big(\sigma^{-\frac{3}{2}}e^{-\frac{\pi d^2}{2\sigma^2}}\big).
    \end{align}
    Given the optimal width scaling $\sigma^2 \approx d$, both exponential arguments scale as $-\pi d/2$. This completes the proof.
\end{proof}

\section{Approximate Weyl relation}
\label{SI-sec:weyl_relation}

The DGW framework offers us the approximate Weyl relation, which mimics the continuous Lie algebra in the sub-exponential bulk. In this section, we aanlyze the continuous translation acting on the DGW state and provide a rigorous proof of the approximate Weyl relation. 

\subsection{Action of Arbitrary Phase-Space Displacement}
To analyze the action of arbitrary phase-sopace displacement, we apply the Trotterization and Poisson summation techniques from~\cite[Theorem 8.1]{Woods_2019} to arbitrary phase-space displacement generator $g(\gamma)= \gamma_p \hat{X} - \gamma_x \hat{P}$. Crucially, we explicitly track the Gaussian tail errors to demonstrate exactly how the approximation breaks down when the state's center approaches the boundary of the bulk. This leads to Lemmas~\ref{SI-lem:general_translation}. 

\begin{lemma}[Action of Arbitrary Phase-Space Displacement]
\label{SI-lem:general_translation}
Let $\sigma\approx \sqrt d$.  Let $\gamma=(\gamma_p,\gamma_x)\in\mathbb{R}^2$ and define the phase-space displacement coordinates $\Gamma=(\Gamma_p,\Gamma_x) = \sqrt{\frac{d}{2\pi}}(\gamma_p,\gamma_x) $. Assume $|\Gamma_p|,|\Gamma_x|\leq d$. Let $g(\gamma)=\gamma_p \hat{X} -\gamma_x \hat{P}$ and $G(\gamma) = e^{\iu g(\gamma)}$ for $\gamma =(\gamma_p,\gamma_x)$ and $\Gamma=(\Gamma_p,\Gamma_x)$. The action of phase-space displacement yields
\begin{align}
    G(\gamma) \ket{\Psi(n_0,k_0)} = e^{\iu \phi(\gamma;n_0,k_0)} \ket{\Psi(n_0+\Gamma_x,k_0+\Gamma_p)}  + \ket{\epsilon},
    \label{SI-eqn:general_translation_action}
\end{align}
where the phase shift is 
\begin{align}
    \phi(\gamma;n_0,k_0) = \sqrt{\frac{2\pi}{d}} \gamma_p n_0 +\frac{1}{2} \gamma_p \gamma_x
\end{align}
and the error vector is bounded by
\begin{align}
    \|\ket{\epsilon}\|_2 \leq \sqrt{d} (\epsilon_{\mathrm{pos},x} + \epsilon_{\mathrm{pos},p} + \epsilon_{\mathrm{alias}} + \epsilon_{\mathrm{mom}})  + (\lceil |\Gamma_x|\rceil+1) (\epsilon_{\mathrm{window}}+\epsilon_{\mathrm{norm}}).
    \label{SI-eqn:general_translation_error}
\end{align}
where the state-dependence error is
\begin{align}
    & \epsilon_{\mathrm{pos},x} = |\Gamma_p| \mathcal{O}\big(\poly(d)e^{-\frac{\pi}{\sigma^2}D_{X,*}^2 } \big), && \epsilon_{\mathrm{mom}} \leq |\Gamma_x| \mathcal{O}\big(\poly(d)e^{-\frac{\pi\sigma^2}{d^2}D_{P,*}^2 } \big), 
\end{align}
where $D_{X,*}$ and $D_{P,*}$ are the infimum distance from the center to the boundary in the $X$ axis and in the $P$ axis, respectively, and the structural errors are
\begin{align}
    & \epsilon_{\mathrm{pos},p} \leq  |\Gamma_x|  \mathcal{O}\big( e^{-\frac{\pi d^2}{4\sigma^2}} \big),  &&  \epsilon_{\mathrm{alias}}  \leq  |\Gamma_x|  \mathcal{O}\big(  \poly(d)e^{-\frac{\pi d^2}{4\sigma^2}} \big), \\
    & \epsilon_{\mathrm{window}}   \leq   \mathcal{O}\big(  e^{-\frac{\pi d^2}{4\sigma^2}}  \big),  &&  \epsilon_{\mathrm{norm}}  \leq  \mathcal{O}\big( e^{-\frac{\pi\sigma^2}{2}} \big) +  \mathcal{O}\big(  e^{-\frac{\pi d^2}{2\sigma^2}} \big).
\end{align}
\end{lemma}

\begin{remark}
    In particular, if $\sigma\approx \sqrt{d}$, and if the full general trajectory remains in the bulk,
    \begin{align}
        |n_0|, |n_0+\Gamma_x|\leq c\frac{d}{2}, \qquad  |k_0|, |k_0+\Gamma_p|\leq c\frac{d}{2}, 
    \end{align}
    for some constant $c\in(0,1)$, then
    \begin{align}
        \|\ket{\epsilon}\|_2 \leq \mathcal{O}\big( \poly(d)e^{-\frac{\pi}{4}(1-c)^2 d} \big). 
        \label{SI-eqn:general_bulk_error}
    \end{align}
    However, if the accumulated parameter shifts ($\Gamma_x,\Gamma_p$) force the wavepacket to hit the boundary ($c \to 1$, $D_X, D_P \to 0$), the exponential suppression of the error term is destroyed, and the approximate phase-space translation fails. 
\end{remark}

\begin{proof}
We track the continuous translation path. 
Without loss of generality, we analyze the case $\Gamma_x,\Gamma_p\geq0$; the other cases follow similarly. For $\Gamma_x>0$, let $0=s_0<s_1<\cdots<s_L=1$ be the ordered set consisting of $0$, $1$, and all $s\in(0,1)$ satisfying $n_0+s\Gamma_x\in\mathbb{Z}$, and set $t_i=s_{i+1}-s_i$; if $\Gamma_x=0$, take $L=1$ and $t_0=1$. The errors in all blocks are additive:
\begin{align}
    G(\gamma)\ket{\Psi(n_0,k_0)} & = \prod_{i=0}^{L-1}G(t_i\gamma)\ket{\Psi(n_0,k_0)}=\prod_{i=1}^{L-1}G(t_i\gamma) \big(e^{\iu\phi(t_0\gamma;n_0,k_0)}\ket{\Psi(n_1,k_1)}+\ket{\epsilon^{(0)}}\big)=\cdots\\
    & =e^{\iu\phi(\gamma;n_0,k_0)}\ket{\Psi(n_0+\Gamma_x,k_0+\Gamma_p)}+\sum_{j=0}^{L-1}e^{\iu\sum_{i=0}^{j-1}\phi(t_i\gamma;n_i,k_i)}\prod_{i=j+1}^{L-1}G(t_i\gamma)\ket{\epsilon^{(j)}},
\end{align}
where $t_i=s_{i+1}-s_i$, $n_i=n_0+s_i\Gamma_x=n_0+\sum_{j=0}^{i-1}t_j\Gamma_x$, and $k_i=k_0+s_i\Gamma_p=k_0+\sum_{j=0}^{i-1}t_j\Gamma_p$, with $\sum_{i=0}^{L-1}t_i=1$ and $L\leq\lceil\Gamma_x\rceil+1$; an empty sum is zero and an empty product is the identity. We will show that for each step
\begin{align}
    \phi(t_i\gamma;n_i,k_i)=\frac{2\pi}{d}t_i\Gamma_p n_i+\frac{\pi}{d}t_i^2\Gamma_p\Gamma_x .
\end{align}
and therefore for all steps
\begin{align}
    \phi(\gamma;n_0,k_0)&=\sum_{i=0}^{L-1}\phi(t_i\gamma;n_i,k_i)=\sum_{i=0}^{L-1}\frac{2\pi}{d}t_i\Gamma_p n_i+\sum_{i=0}^{L-1}\frac{\pi}{d}t_i^2\Gamma_p\Gamma_x\\
    &=\sum_{i=0}^{L-1}\frac{2\pi}{d}t_i\Gamma_p\left(n_0+\sum_{j=0}^{i-1}t_j\Gamma_x\right)+\sum_{i=0}^{L-1}\frac{\pi}{d}t_i^2\Gamma_p\Gamma_x\\
    &=\frac{2\pi}{d}\Gamma_p n_0+\sum_{i=0}^{L-1}\sum_{j=0}^{i-1}\frac{2\pi}{d}t_it_j\Gamma_p\Gamma_x+\sum_{i=0}^{L-1}\frac{\pi}{d}t_i^2\Gamma_p\Gamma_x\\
    &=\frac{2\pi}{d}\Gamma_p n_0+\frac{\pi}{d}\Gamma_p\Gamma_x\left(\sum_{i=0}^{L-1}t_i\right)^2\\
    &=\frac{2\pi}{d}\Gamma_p n_0+\frac{\pi}{d}\Gamma_p\Gamma_x=\sqrt{\frac{2\pi}{d}}\gamma_p n_0+\frac{1}{2}\gamma_p\gamma_x .
\end{align}
It is worth noting that, since every $G(t_i\gamma)$ is unitary and every phase has unit modulus,
\begin{align}
    \left\|\sum_{j=0}^{L-1}e^{\iu\sum_{i=0}^{j-1}\phi(t_i\gamma;n_i,k_i)} \prod_{i=j+1}^{L-1} G(t_i\gamma)\ket{\epsilon^{(j)}}\right\|_2\leq\sum_{j=0}^{L-1}\|\ket{\epsilon^{(j)}}\|_2 .
\end{align}
We analyze the first block, corresponding to $t\in[0,t_0]$; all other blocks follow similarly after replacing the initial center $(n_0,k_0)$ by the corresponding block center $(n_i,k_i)$. 
\begin{align}
    \ket{\Psi(n_0,k_0)} = A_{\mathrm{norm}}^0 \sum_{n\in\mathcal{S}_d(n_0)} \psi(n_0,k_0;n)\ket{x_n}, 
\end{align}
where 
\begin{align}
    \psi(n_0, k_0; n) = e^{-\frac{\pi}{\sigma^2}(n-n_0)^2} e^{\iu \frac{2\pi}{d} k_0(n-n_0)},
\end{align}
and 
\begin{align}
    A_{\mathrm{norm}}^0 = \Big(\sum_{n\in\mathcal{S}_d(n_0)} |\psi(n_0, k_0; n)|^2\Big)^{-\frac{1}{2}}. 
\end{align}
We also set the unnormalized state
\begin{align}
    \ket{\Psi_0(n_t,k_t)} = A_{\mathrm{norm}}^0 \sum_{n\in\mathcal{S}_d(n_t)} \psi(n_t,k_t;n)\ket{x_n}.  
\end{align}
We define the continuous Fourier transform as 
\begin{align}
    \mathcal{F}[f](k) = \frac{1}{\sqrt{d}} \int_{-\infty}^{\infty} f(x) e^{-\iu\frac{2\pi}{d}kx} \d x,
\end{align}
and its corresponding inverse continuous Fourier transform is defined as:
\begin{align}
    \mathcal{F}^{-1}[g](x) = \frac{1}{\sqrt{d}} \int_{-\infty}^{\infty} g(k) e^{\iu\frac{2\pi}{d}kx} \d k.
\end{align}
We denote the continuous Fourier transform of $\psi$ as
\begin{align}
    \tilde{\psi}(n_0,k_0;p) = \mathcal{F}[\psi](n_0,k_0;p) = \frac{\sigma}{\sqrt{d}} e^{-\iu \frac{2\pi}{d} p n_0} e^{-\frac{\pi\sigma^2}{d^2}(p-k_0)^2}. 
\end{align}

\noindent \textbf{Block One: Continuous Errors}

We will denote 
\begin{align}
    n_t = n_0 + t\Gamma_x, \qquad k_t =  k_0 + t \Gamma_p. 
\end{align}
In the $\ket{x_l}$ basis, the exact dynamically evolved components are:
\begin{align}
    c_l(t)  = \bra{x_l} G(t \gamma) \ket{\Psi(n_0, k_0)} .
\end{align}
We compare this against the ideal analytic shift components
\begin{align}
    c'_l(t) = A_{\mathrm{norm}}^0 e^{\iu \phi(t \gamma; n_0,k_0)}\psi(n_t, k_t; l). 
\end{align}
At $t = 0$, they are equal. We will prove 
\begin{align}
    \partial_y c_l(y)\big|_{y=t} \approx \partial_y c_l'(y)\big|_{y=t }, 
\end{align}
which subsequently yields via integration
\begin{align}
    c_l(t) \approx c_l'(t). 
\end{align}
We define
\begin{align}
    \ket{\Delta(t)} = G(t \gamma) \ket{\Psi(n_0, k_0)} - e^{\iu \phi(t\gamma;n_0,k_0)} \ket{\Psi_0(n_t,k_t)} ,
\end{align}
and 
\begin{align}
    \Delta_l(t) = \braket{x_l}{\Delta(t)} = c_l(t) - c_l'(t). 
\end{align}
The derivative of the exact state is:
\begin{align}
    \partial_y c_l(y)\big|_{y=t} & = \bra{x_l} \partial_t  G(t \gamma) \ket{\Psi(n_0, k_0)} = \iu  \bra{x_l} g(\gamma)  G(t \gamma) \ket{\Psi(n_0, k_0)} \\
    & =\iu \bra{x_l}  g(\gamma)  \ket{\Delta(t)} + \iu e^{\iu \phi(t\gamma;n_0,k_0)} \bra{x_l} g(\gamma)  \ket{\Psi_0(n_t,k_t)} ,
\end{align}
The derivative of the ideal state is 
\begin{align}
    \partial_y c_l'(y) \big|_{y=t} & = A_{\mathrm{norm}}^0 e^{\iu\phi(t\gamma;n_0,k_0)}  \Big(\frac{2\pi}{\sigma^2} \Gamma_x(l-n_t)+\iu \frac{2\pi}{d} (\Gamma_p l- \Gamma_x k_t)\Big) \psi(n_t,k_t;l). 
\end{align}
The difference between these two derivative is 
\begin{align}
    \bra{x_l} \partial \ket{\Delta(t)} = \partial_y c_l(y)\big|_{y=t} -\partial_y c_l'(y)\big|_{y=t}. 
\end{align}
To show that they are close, we need to evaluate each term. Specifically, 
\begin{align}
    \bra{x_l} g(\gamma)  \ket{\Psi_0(n_t,k_t)}  & =  \bra{x_l} (\gamma_p \hat{X} - \gamma_x \hat{P})  \ket{\Psi(n_t,k_t)}  . 
\end{align}
We will deal with the following two terms separately: 
\begin{align}
    \bra{x_l} \gamma_p \hat{X} \ket{\Psi(n_t,k_t)}  = A_{\mathrm{norm}}^0 \frac{2\pi}{d} \Gamma_p [l]_d  \psi(n_t,k_t;l) ,
\end{align}
and 
\begin{align}
    \bra{x_l} \gamma_x \hat{P} \ket{\Psi_0(n_t,k_t)} & = A_{\mathrm{norm}}^0 \frac{2\pi}{d}  \sum_{k=-m}^{m} \sum_{n\in\mathcal{S}_d(n_t)} \Gamma_x k \psi(n_t,k_t;n) \braket{x_l}{p_k}  \braket{p_k}{x_{n}}  \\
    & =  A_{\mathrm{norm}}^0 \frac{2\pi}{d^2}   \sum_{k=-m}^{m}   \Gamma_x k e^{\iu \frac{2\pi l k }{d}} \sum_{n\in\mathcal{S}_d(n_t)} \psi(n_t,k_t;n)e^{-\iu \frac{2\pi n k }{d}} . 
\end{align}
Note that as we are working in the same block before crossing the window, we have $\mathcal{S}_d(n_0)=\mathcal{S}_d(n_t)$ at this stage. 

\noindent\textbf{Step 1: Extend position operator.}

For $\bra{x_l} \gamma_p \hat{X} \ket{\Psi(n_t,k_t)}$, we extend the remainder $[l]_d$ in finite regime to the number $l$ in infinite regime, creating an error $\epsilon_{\mathrm{pos},x}$
\begin{align}
    \bra{x_l} \gamma_p \hat{X} \ket{\Psi_0(n_t,k_t)} = A_{\mathrm{norm}}^0 \frac{2\pi}{d} \Gamma_p l  \psi(n_t,k_t;l)  + \epsilon_{\mathrm{pos},x},
\end{align}
where 
\begin{align}
    \epsilon_{\mathrm{pos},x} = A_{\mathrm{norm}}^0 \frac{2\pi}{d} \Gamma_p ([l]_d - l)  \psi(n_t,k_t;l). 
\end{align}
Notice that $[l]_d=l$ for $l\in\mathcal{S}_d(0)$, 
\begin{align}
    |\epsilon_{\mathrm{pos},x} | \leq \Gamma_p \mathcal{O}\big( \poly(d) e^{-\frac{\pi}{\sigma^2}D_X(n_t)^2 }\big). 
\end{align}

\noindent\textbf{Step 2: Extend position sum.}

For $\bra{x_l} \gamma_x \hat{P} \ket{\Psi(n_t,k_t)}$, we also extend the sum over $n$ from $\mathcal{S}_d(n_t)$ to $\mathbb{Z}$, creating an error $\epsilon_{\mathrm{pos},p}$
\begin{align}
    \bra{x_l} \gamma_x \hat{P} \ket{\Psi(n_t,k_t)} = A_{\mathrm{norm}}^0 \frac{2\pi}{d^2}   \sum_{k=-m}^m \Gamma_x k e^{\iu \frac{2\pi k l}{d}} \sum_{n \in \mathbb{Z}} \psi(n_t, k_t; n) e^{-\iu \frac{2\pi k n}{d}} + \epsilon_{\mathrm{pos},p},
\end{align}
where 
\begin{align}
    \epsilon_{\mathrm{pos},p} = A_{\mathrm{norm}}^0 \frac{2\pi}{d^2}   \sum_{k=-m}^m \Gamma_x k e^{\iu \frac{2\pi k l}{d}} \sum_{n \notin \mathcal{S}_d(n_t) } \psi(n_t, k_t; n) e^{-\iu \frac{2\pi k n}{d}}. 
\end{align}
The set $\mathcal{S}_d(n_t)$ is strictly centered on $n_t$, meaning the omitted terms are $|n - n_t| > d/2$. 
\begin{align}
    |\epsilon_{\mathrm{pos},p}| \leq A_{\mathrm{norm}}^0\frac{2\pi}{d^2}  \sum_{k=-m}^m \Gamma_x |k |\sum_{n \notin \mathcal{S}_d(n_t)} |\psi(n_t,k_t;n)| \leq \Gamma_x \mathcal{O}\big(e^{-\frac{\pi d^2}{4\sigma^2}}\big).
\end{align}

\noindent\textbf{Step 3: Poisson summation.}

Now we apply Poisson summation to transform the infinite position sum to the infinite momentum sum 
\begin{align}
    \sum_{n \in \mathbb{Z}} \psi(n_t,k_t;n) e^{-\iu \frac{2\pi}{d} k n} = \sqrt{d} \sum_{s \in \mathbb{Z}} \tilde{\psi}(n_t,k_t;sd + k). 
\end{align}
Let $p = sd + k$. When summing over $k\in[-m,m]$ and $s\in\mathbb{Z}$, we equivalently sum over $p\in\mathbb{Z}$ by setting $k = p - s(p) d = p - \mathrm{round}(p/d)  d$. Substituting this:
\begin{align}
    \bra{x_l} \gamma_x \hat{P} \ket{\Psi_0(n_t,k_t)} = A_{\mathrm{norm}}^0 \frac{2\pi}{d^{\frac{3}{2}}} \sum_{p \in \mathbb{Z}} \Gamma_x (p - s(p) d) \tilde{\psi}(n_t,k_t;p) e^{\iu \frac{2\pi}{d} p l} e^{- \iu 2\pi s(p)l}  + \epsilon_{\mathrm{pos},p}.
\end{align}
Noting that $e^{\iu 2\pi s(p)l}=1$ as $s(p)\in \mathbb{Z}$, we split this into the ideal continuous derivative (involving $p$) and the boundary cutoff error $\epsilon_{\mathrm{mom}}$ (involving $s(p)d$):
\begin{align}
    \bra{x_l} \gamma_x \hat{P} \ket{\Psi_0(n_t,k_t)} = A_{\mathrm{norm}}^0 \frac{2\pi}{d^{\frac{3}{2}}} \sum_{p \in \mathbb{Z}} \Gamma_x p \tilde{\psi}(n_t,k_t;p) e^{\iu \frac{2\pi}{d} p l} + \epsilon_{\mathrm{mom}} +  \epsilon_{\mathrm{pos},p}.
\end{align}
where, using $s(p) = 0$ for $p\in\mathcal{S}_d(0)$, we replace the sum for $s(p)$ from over all integer momentum values to over integer momentum values outside the Brillouin zone
\begin{align}
    \epsilon_{\mathrm{mom}} = - A_{\mathrm{norm}}^0  \frac{2\pi}{d^{\frac{3}{2}}} \sum_{p\notin \mathcal{S}_d(0)} \Gamma_x s(p) d  \tilde{\psi}(n_t,k_t;p) e^{\iu \frac{2\pi}{d} p l} .
\end{align}
We bound $|s(p)| \leq \frac{|p|}{d} + \frac{1}{2}$. The distance from the center $k_t$ to the summation boundary $\pm m$ is precisely $D_P(k_t)$. Bounding the sum with a Gaussian integral yields the bound for $\epsilon_{\mathrm{mom}}$
\begin{align}
    |\epsilon_{\mathrm{mom}}| &  \leq  A_{\mathrm{norm}}^0  \frac{2\pi \sigma}{d^2} \sum_{p\notin\mathcal{S}_d(0)} \Gamma_x |s(p)|d  e^{-\frac{\pi \sigma^2}{d^2}(p-|k_t|)^2}  \leq A_{\mathrm{norm}}^0   \frac{4\pi \sigma}{d^2} \sum_{x = 1}^\infty \Gamma_x  (x+d)  e^{-\frac{\pi \sigma^2}{d^2} (x+D_P(k_t))^2}  \\
    & \leq  A_{\mathrm{norm}}^0  \frac{4\pi \sigma}{d^2}  e^{-\frac{\pi \sigma^2}{d^2} D_P(k_t)^2} \sum_{x = 1}^\infty \Gamma_x  (x + d) e^{-\frac{2\pi \sigma^2}{d^2} D_P(k_t) x}  \leq \Gamma_x \mathcal{O}\big(\poly(d)  e^{-\frac{\pi \sigma^2}{d^2}D_P(k_t)^2}\big).
\end{align}

\noindent\textbf{Step 4: Inverse Poisson summation.}

The Fourier transform of the derivative of a function is exactly the Fourier transform of the function multiplied by the variable. Therefore, 
\begin{align}
    \frac{\iu 2\pi}{d} p \tilde{\psi}(n_t,k_t;p) =  \mathcal{F}[\partial_y \psi](n_t,k_t;\tilde{y})\big|_{\tilde{y}=p}. 
\end{align}
Applying Poisson summation inversely yields
\begin{align}
    \sum_{p \in \mathbb{Z}} \mathcal{F}[\partial_y \psi](n_t,k_t;\tilde{y})\big|_{\tilde{y}=p} e^{\iu \frac{2\pi}{d} p l} = \sqrt{d} \sum_{ r \in \mathbb{Z}} [\partial_y \psi](n_t,k_t;y)\big|_{y=rd+l} . 
\end{align}
Substituting into $\bra{x_l} \gamma_x \hat{P} \ket{\Psi(n_t,k_t)}$, we obtain
\begin{align}
    \bra{x_l} \gamma_x \hat{P}\ket{\Psi(n_t,k_t)}= - \iu A_{\mathrm{norm}}^0  \Gamma_x\sum_{r\in\mathbb{Z}}  [\partial_y \psi](n_t,k_t;y)\big|_{y=rd+l} + \epsilon_{\mathrm{mom}} +  \epsilon_{\mathrm{pos},p}. 
\end{align}
$r=0$ gives the exact target derivative 
\begin{align}
    - \iu A_{\mathrm{norm}}^0  \Gamma_x [\partial_y \psi](n_t,k_t;y)\big|_{y=l} = A_{\mathrm{norm}}^0 \Big(\iu \frac{2\pi}{\sigma^2}\Gamma_x (l-n_t) + \frac{2\pi}{d} \Gamma_x  k_t \Big) \psi(n_t,k_t;l) .
\end{align}
For $r \neq 0$, the distance $|l +  rd - n_t| \geq (d-1)/2$ because $l \in \mathcal{S}_d(n_t)$. This gives the position tail error
\begin{align}
    \epsilon_{\mathrm{alias}} =  A_{\mathrm{norm}}^0 \Gamma_x \sum_{ r\neq 0} [\partial_y \psi](n_t,k_t;y)\big|_{y=rd+l},
\end{align}
such that 
\begin{align}
    |\epsilon_{\mathrm{alias}} |\leq  \Gamma_x \mathcal{O}\big(\poly(d)e^{-\frac{\pi d^2}{ 4\sigma^2}}\big).
\end{align}
Therefore, we obtain
\begin{align}
    \bra{x_l} \gamma_x \hat{P} \ket{\Psi_0(n_t,k_t)} = A_{\mathrm{norm}}^0 \Big(\iu \frac{2\pi}{\sigma^2} \Gamma_x (l-n_t) + \frac{2\pi}{d} \Gamma_x  k_t \Big) \psi(n_t,k_t;l)  + \epsilon_{\mathrm{alias}}  + \epsilon_{\mathrm{mom}} +  \epsilon_{\mathrm{pos},p}.
\end{align}

\noindent\textbf{Step 5: Integration.}

Summarizing the previous steps, we obtain 
\begin{align}
    \partial_y c_l'(y) \big|_{y=t} & = A_{\mathrm{norm}}^0 e^{\iu\phi(t\gamma;n_0,k_0)}  \Big(\frac{2\pi}{\sigma^2} \Gamma_x(l-n_t)+\iu \frac{2\pi}{d} (\Gamma_p l- \Gamma_x k_t)\Big) \psi(n_t,k_t;l),
\end{align}
and 
\begin{align}
    \partial_y c_l(y)\big|_{y=t}  & = A_{\mathrm{norm}}^0 e^{\iu \phi(t\gamma;n_0,k_0)} \Big(\frac{2\pi}{\sigma^2} \Gamma_x(l-n_t)+\iu \frac{2\pi}{d} (\Gamma_p l - \Gamma_x k_t)\Big) \psi(n_t,k_t;l)  \\
    & \quad + \iu \bra{x_l}  g(\gamma) \ket{\Delta(t)}  + \iu e^{\iu \phi(t\gamma;n_0,k_0)} (\epsilon_{\mathrm{pos},x} - \epsilon_{\mathrm{pos},p} - \epsilon_{\mathrm{mom}} - \epsilon_{\mathrm{alias}}) .
\end{align}
We see that 
\begin{align}
    \bra{x_l} \partial_t  \ket{\Delta(t)} =\iu \bra{x_l}  g(\gamma)\ket{\Delta(t)}  + r_l(t)
\end{align}
where 
\begin{align}
    r_l(t) = \iu e^{\iu \phi(t\gamma;n_0,k_0)} (\epsilon_{\mathrm{pos},x} - \epsilon_{\mathrm{pos},p} - \epsilon_{\mathrm{mom}}  -\epsilon_{\mathrm{alias}})
\end{align}
and 
\begin{align}
     |r_l(t)| = e(t) = \mathcal{O}\Big(\poly(d)\big(\Gamma_p e^{-\frac{\pi}{\sigma^2}D_X(n_t)^2} +  \Gamma_x e^{-\frac{\pi \sigma^2}{d^2}D_P(k_t)^2} + \Gamma_x e^{-\frac{\pi d^2}{4\sigma^2}}  \big)  \Big).
\end{align}
Define the residual vector
\begin{align}
    \ket{R(t)}=\sum_{l=-m}^{m} r_l(t)\ket{x_l}.
\end{align}
where 
\begin{align}
    \|\ket{R(t)}\|_2 \leq \sqrt{d} \max_l |r_l(t)| \leq \sqrt{d} e(t).
\end{align}
Therefore
\begin{align}
    \partial_t \ket{\Delta(t)} = \iu g(\gamma)\ket{\Delta(t)} + \ket{R(t)}
\end{align}
Because $g(\gamma)$ is Hermitian, $G(t\gamma)$ is unitary. Solving this homogeneous linear differential equation yields 
\begin{align}
    \ket{\Delta(t)} = G(t\gamma)\ket{\Delta(0)} + \int_0^t G((t-s)\gamma) \ket{R(s)}\d s.
\end{align}
Taking norms and using unitarity,
\begin{align}
    \|\ket{\Delta(t)}\|_2 \leq \|\ket{\Delta(0)}\|_2 + \int_0^t \|\ket{R(s)}\|_2\d s  \leq  \|\ket{\Delta(0)}\|_2 + \sqrt{d} t \sup_{s\in(0,t)} e(s) .
\end{align}
It is obvious that the initial approximation is exact, i.e. $\|\ket{\Delta(0)}\|_2=0$. Let 
\begin{align}
    D_{X,*}=\inf_{s\in (0,t)} D_X(n_s), \qquad D_{P,*} = \inf_{s\in(0,t)} D_P(k_s). 
\end{align}
Therefore 
\begin{align}
    \|\ket{\Delta(t)}\|_2  \leq   t  \mathcal{O}\Big(\poly(d)\big(\Gamma_p e^{-\frac{\pi}{\sigma^2}D_{X,*}^2} +  \Gamma_x e^{-\frac{\pi \sigma^2}{d^2}D_{P,*}^2} + \Gamma_x e^{-\frac{\pi d^2}{4\sigma^2}}  \big)  \Big) .
\end{align}
Note here we assume that $t$ are sufficiently small such that the center remains in the same block. 

\noindent \textbf{Block Two: Jump Errors}

When the center of the state crosses an integer at $t=T$, the window changes and the sum shifts to a different set of points, which leads to the window jump error. As the center of the state changes, the normalization factor also changes, which leads to the normalization error.

\noindent\textbf{Step 6: Window jumps.}

Whenever the increasing center $n_T$ crosses the next integer, the physical support window jumps from $\mathcal{S}_d(n_0)$ to $\mathcal{S}_d(n_0+1)$, thereby replacing one edge component and re-normalizing the full state. That is, the unnormalized state represented by $c_l(T^-)$ before the jump is
\begin{align}
    \ket{\Psi^-}=A_{\mathrm{norm}}^0\sum_{l\in\mathcal{S}_d(n_0)}\psi(n_T,k_T;l)\ket{x_l},
\end{align}
while the unnormalized state represented by $c_l(T^+)$ after the jump is
\begin{align}
    \ket{\Psi^+}=A_{\mathrm{norm}}^0\sum_{l\in\mathcal{S}_d(n_0+1)}\psi(n_T,k_T;l)\ket{x_l},
\end{align}
where $A_{\mathrm{norm}}^0$ is the normalization factor not only for the DGW state $\ket{\Psi(n_0,k_0)}$ but also for the integer-translated DGW state $\ket{\Psi(n_0+1,k_0)}$:
\begin{align}
    A_{\mathrm{norm}}^0=\left(\sum_{n\in\mathcal{S}_d(n_0)}|\psi(n_0,k_0;n)|^2\right)^{-\frac12}=\left(\sum_{n\in\mathcal{S}_d(n_0+1)}|\psi(n_0+1,k_0;n)|^2\right)^{-\frac12}.
\end{align}
This induces a truncation jump error. The support loses its lower-edge component and gains a new upper-edge component, so the step error is bounded by the sum of these two edge terms:
\begin{align}
    \epsilon_{\mathrm{window}}&\leq\|\ket{\Psi^-}-\ket{\Psi^+}\|_2\leq\|\ket{\Psi^-}-\ket{\Psi^+}\|_1\\
    &\leq A_{\mathrm{norm}}^0|\psi(n_T,k_T;n_T-m-1)|+A_{\mathrm{norm}}^0|\psi(n_T,k_T;n_T+m)|\leq\mathcal{O}\!\left(e^{-\frac{\pi d^2}{4\sigma^2}}\right).
\end{align}

\noindent \textbf{Step 7: Normalization. }

The normalized target analytic state after the jump is
\begin{align}
    \ket{\Psi(n_T,k_T)}=A_{\mathrm{norm}}^T\sum_{l\in\mathcal{S}_d(n_0+1)}\psi(n_T,k_T;l)\ket{x_l},
\end{align}
where $A_{\mathrm{norm}}^T$ is the normalization factor for the DGW state $\ket{\Psi(n_T,k_T)}$:
\begin{align}
    A_{\mathrm{norm}}^T=\left(\sum_{n\in\mathcal{S}_d(n_0+1)}|\psi(n_T,k_T;n)|^2\right)^{-\frac12}.
\end{align}
The change of the normalization factor induces a normalization error. The centers $n_0+1$ and $n_T$ have the same support set $\mathcal{S}_d(n_0+1)$, so their normalization factors are exponentially close:
\begin{align}
    \epsilon_{\mathrm{norm}}\leq\|\ket{\Psi(n_T,k_T)}-\ket{\Psi^+}\|_2\leq|A_{\mathrm{norm}}^T-A_{\mathrm{norm}}^0|\leq\mathcal{O}\!\left(e^{-\frac{\pi\sigma^2}{2}}\right)+\mathcal{O}\!\left(e^{-\frac{\pi d^2}{2\sigma^2}}\right).
\end{align}

All these continuous errors integrates over these intervals and all these additive jump errors accumulate over the $L\leq\lceil|\Gamma_x|\rceil+1$ blocks. Therefore, the final error is at most 
\begin{align}
    \|\ket{\epsilon}\|_2 \leq  \sqrt{d} (\epsilon_{\mathrm{alias}} + \epsilon_{\mathrm{mom}} +\epsilon_{\mathrm{pos},x} + \epsilon_{\mathrm{pos},p})  + ( \lceil|\Gamma_x| \rceil +1)(\epsilon_{\mathrm{window}} + \epsilon_{\mathrm{norm}}),
\end{align}
which completes the proof. 
\end{proof}

\subsection{Approximate Weyl relation}
Based on Lemma~\ref{SI-lem:general_translation}, we are ready to present the approximate Weyl relation which is the core theorem for our quantum metrology with indefinite causal order. 
\begin{theorem}[Approximate Weyl Relation for Arbitrary Phase-Space Displacement]
\label{SI-thm:general_weyl_relation}
Let $\sigma\approx \sqrt{d}$.  Let $\alpha=(\alpha_p,\alpha_x)$ and $\beta=(\beta_p,\beta_x)$ and define the scaled phase-space displacements $A = (A_p,A_x) = \sqrt{\frac{d}{2\pi}}(\alpha_p,\alpha_x)$ and $B = (B_p, B_x) = \sqrt{\frac{d}{2\pi}}(\beta_p,\beta_x)$. Assume that $|A_p|, |A_x|,|B_p|,|B_x|\leq d$. Let $g(\gamma) = \gamma_p\hat{X} - \gamma_x\hat{P}$ and $G(\gamma)=e^{\iu g(\gamma)}$ for a general pair $\gamma =(\gamma_p,\gamma_x)$ and a displacement $\Gamma = (\Gamma_p, \Gamma_x)$. Then 
\begin{align}
    G(\beta) G(\alpha) \ket{\Psi(n_0,k_0)} = e^{\iu\omega(\alpha,\beta)} G(\alpha) G(\beta) \ket{\Psi(n_0,k_0)} + \ket{\delta_{\alpha,\beta}},
    \label{SI-eqn:general_weyl_relation}
\end{align}
where the phase $\omega(\alpha,\beta)$ is exactly the area enclosed by the two phase-space displacement paths
\begin{align}
    \omega(\alpha,\beta) =  \alpha_x\beta_p - \alpha_p\beta_x = \frac{2\pi}{d}(A_xB_p - A_pB_x).
    \label{SI-eqn:general_phase}
\end{align}
The total error vector is bounded by
\begin{align}
    \|\ket{\delta_{\alpha,\beta}}\|_2  \leq   \mathcal{O}\Big(\poly(d) \big(|\Gamma_p^*| e^{-\frac{\pi}{\sigma^2}D_{X,*}^2 }  + |\Gamma_x^*|(e^{-\frac{\pi\sigma^2}{d^2}D_{P,*}^2 }  +  e^{-\frac{\pi\sigma^2}{2}}  +e^{-\frac{\pi d^2}{4\sigma^2}} \big) \Big)
\end{align}
where $\Gamma_x^*= |A_x|+|B_x| $ and $\Gamma_p^*= |A_p|+|B_p| $ are the cumulative displacement in the trajectory, $D_{X,*}$ and $D_{P,*}$ are the minimum of the distance to the boundary in the $X$ parameter and in the $P$ parameter in the trajectory from $(n_0,k_0)$ to $(n_0+A_x+B_x, k_0+A_p+B_p)$. 
\end{theorem}

\begin{remark}[Mechanism of Boundary Breakdown]
    The error terms dictate precisely when the approximation breaks down. In particular, suppose that $\sigma\approx\sqrt{d}$, and if the full general trajectory remains in the bulk, 
    \begin{align}
        & |n_0|,|n_0+A_x|, |n_0+B_x| , |n_0+A_x+B_x|\leq c \frac{d}{2}, \\
        & |k_0|, |k_0+A_p|, |k_0+ B_p|, |k_0+A_p+B_p|\leq c \frac{d}{2}. 
    \end{align}
    for some constant $c\in(0,1)$. Then
    \begin{align}
        \|\ket{\delta_{\alpha,\beta}}\|_2 \leq \mathcal{O}\big(\poly(d)e^{-\frac{\pi}{4}(1-c)^2d}\big),
        \label{SI-eqn:general_weyl_error}
    \end{align}
    However, if the accumulated parameter shifts at any intermediate time ($A_x,B_x,A_x+B_x,A_p,B_p,A_p+B_p$) force the wavepacket to hit the boundary ($c \to 1$, $D_X,D_P \to 0$), the exponential suppression of the error term is destroyed, and the approximate Weyl relation fails. 
\end{remark}

\begin{proof}
We use Lemma~\ref{SI-lem:general_translation} twice, once for each ordered path. When we first displace $\alpha$ then $\beta$, from Eq.\eqref{SI-eqn:general_translation_action}, we obtain
\begin{align}
    G(\beta)G(\alpha) \ket{\Psi(n_0,k_0)} & = G(\beta) (e^{\iu \phi(\alpha;n_0,k_0)} \ket{\Psi(n_0+A_x,k_0+A_p)} +\ket{\epsilon(\alpha,n_0,k_0)} ) \\
    & = e^{\iu( \phi(\alpha;n_0,k_0) + \phi(\beta;n_0+A_x,k_0+A_p) )}   \ket{\Psi(n_0+A_x+B_x,k_0+A_p+B_p)} \\
    & \quad +G(\beta) \ket{\epsilon(\alpha,n_0,k_0)}+e^{\iu \phi(\alpha;n_0,k_0)} \ket{\epsilon(\beta,n_0+A_x,k_0+A_p)}
\end{align}
and 
\begin{align}
    G(\alpha) G(\beta) \ket{\Psi(n_0,k_0)} & = G(\alpha) (e^{\iu \phi(\beta;n_0,k_0)} \ket{\Psi(n_0+B_x,k_0+B_p)} +\ket{\epsilon(\beta,n_0,k_0)} ) \\
    & = e^{\iu( \phi(\beta;n_0,k_0) + \phi(\alpha;n_0+B_x,k_0+B_p) )}   \ket{\Psi(n_0+A_x+B_x,k_0+A_p+B_p)} \\
    & \quad + G(\alpha) \ket{\epsilon(\beta,n_0,k_0)}+e^{\iu \phi(\beta;n_0,k_0)} \ket{\epsilon(\alpha,n_0+B_x,k_0+B_p)}
\end{align}
where for a general pair $\gamma$ starting from $(n,k)$, we define
\begin{align}
    \phi(\gamma;n,k) = \sqrt{\frac{2\pi}{d}} \gamma_p n +\frac{1}{2} \gamma_p \gamma_x 
\end{align}
and 
\begin{align}
    \|\ket{\epsilon(\gamma;n,k)}\|_2 \leq \mathcal{O}\Big(\poly(d) \big(|\Gamma_p| e^{-\frac{\pi}{\sigma^2}D_{X,*}^2 }  + |\Gamma_x|(e^{-\frac{\pi\sigma^2}{d^2}D_{P,*}^2 }  +  e^{-\frac{\pi\sigma^2}{2}}  +e^{-\frac{\pi d^2}{4\sigma^2}} \big) \Big)
\end{align}
where $D_{X,*}$ and $D_{P,*}$ are the infimum of the distance to the boundary in the $X$ parameter and in the $P$ parameter in the trajectory from $(n,k)$ to $(n+\Gamma_x,k+\Gamma_p)$, respectively. Therefore, 
\begin{align}
      G(\beta) G(\alpha)\ket{\Psi(n_0,k_0)} & = e^{\iu \omega(\alpha,\beta) }  G(\alpha) G(\beta) \ket{\Psi(n_0,k_0)} + \ket{\delta_{\alpha,\beta}}
\end{align}
where 
\begin{align}
    \omega(\alpha,\beta) & = \phi(\alpha;n_0,k_0) + \phi(\beta;n_0 +A_x,k_0+A_p) - \phi(\beta;n_0,k_0) - \phi(\alpha;n_0+B_x,k_0+B_p) =  \alpha_x \beta_p  - \alpha_p\beta_x, 
\end{align}
and 
\begin{align}
    \ket{\delta_{\alpha,\beta}} &  = G(\beta) \ket{\epsilon(\alpha,n_0,k_0)}+e^{\iu \phi(\alpha;n_0,k_0)} \ket{\epsilon(\beta,n_0+A_x,k_0+A_p)} \\
    & \quad - e^{\iu\omega(\alpha,\beta)} (G(\alpha)\ket{\epsilon(\beta,n_0,k_0)}+e^{\iu \phi(\beta;n_0,k_0)} \ket{\epsilon(\alpha,n_0+B_x,k_0+B_p)}) 
\end{align}
such that 
\begin{align}
    \|\ket{\delta_{\alpha,\beta}}\|_2 & \leq \|\ket{\epsilon(\beta,n_0,k_0)}\|_2 +  \|\ket{\epsilon(\alpha,n_0+B_x,k_0+B_p)}\|_2 \\
    & \quad + \|\ket{\epsilon(\alpha,n_0,k_0)}\|_2 + \|\ket{\epsilon(\beta,n_0+A_x,k_0+A_p)}\|_2 \\
    & \leq \mathcal{O}\Big(\poly(d) \big(|\Gamma_p^*| e^{-\frac{\pi}{\sigma^2}D_{X,*}^2 }  + |\Gamma_x^*|(e^{-\frac{\pi\sigma^2}{d^2}D_{P,*}^2 }  +  e^{-\frac{\pi\sigma^2}{2}}  +e^{-\frac{\pi d^2}{4\sigma^2}} \big) \Big)
\end{align}
where $\Gamma_x^*= |A_x|+|B_x| $ and $\Gamma_p^*= |A_p|+|B_p| $ are the cumulative displacement in the trajectory, $D_{X,*}$ and $D_{P,*}$ are the infimum of the distance to the boundary in the $X$ parameter and in the $P$ parameter in the trajectory from $(n_0,k_0)$ to $(n_0+A_x+B_x, k_0+A_p+B_p)$. This completes the proof.
\end{proof}

\section{Approximate Energy Estimates}
\subsection{Approximate expectations}
\label{SI-sec:expectation}
We also evaluate the expected position, momentum and energy of the DGW state, keeping track of the exact discrete aliasing, truncation, and boundary wrapping errors.

\begin{theorem}[Approximate Expectations]
    \label{SI-thm:expectations}
    Let $\sigma\approx \sqrt{d}$ and $n_0, k_0 \in [-m,m]$. The expectations of the canonical operators evaluate exactly to their classical continuous values up to bounded error terms:
    \begin{align}
        \langle \hat{X} \rangle_{\Psi} &= \sqrt{\frac{2\pi}{d}} n_0 + \epsilon_{X}, \quad \langle \hat{P} \rangle_{\Psi} = \sqrt{\frac{2\pi}{d}} k_0 + \epsilon_{P},
    \end{align}
    and the expected total energy for the Hamiltonian $\hat{H} = \hat{X}^2 + \hat{P}^2$ is:
    \begin{align}
        \langle \hat{H} \rangle_{\Psi} =  \frac{2\pi}{d}\big(n_0^2 + k_0^2 ) + \frac{\sigma^2}{2d}  +  \frac{d}{2\sigma^2}   + \epsilon_{H}.
    \end{align}
    The expectation errors are explicitly bounded by the Gaussian tails beyond the boundary distances $D_X(n_0) = m - |n_0|$ and $D_P(k_0) = m - |k_0|$:
    \begin{align}
        |\epsilon_{X}| & \leq  \mathcal{O}\big(\poly(d)e^{-\frac{\pi d^2 }{2\sigma^2}}\big) + \mathcal{O}\big(\poly(d)e^{-\frac{\pi \sigma^2}{2}} \big)  + \mathcal{O}\big(\poly(d) e^{-\frac{2\pi}{\sigma^2}D_X(n_0)^2}\big) ,\\
        |\epsilon_{P}| &  \leq \mathcal{O}\big(\poly(d)e^{-\frac{\pi d^2 }{4\sigma^2}}\big) + \mathcal{O}\big(\poly(d)e^{-\frac{\pi \sigma^2}{2}} \big) + \mathcal{O}\big(\poly(d) e^{-\frac{2\pi\sigma^2}{d^2}D_P(k_0)^2}\big).
    \end{align}
    The total energy error $\epsilon_{H}$ is dominated by the second moments of the truncation tails:
    \begin{align}
        \epsilon_{H}\leq \mathcal{O}\big(\poly(d)e^{-\frac{\pi d^2 }{4\sigma^2}}\big) + \mathcal{O}\big(\poly(d)e^{-\frac{\pi \sigma^2}{2}} \big) + \mathcal{O}\big(\poly(d) e^{-\frac{2\pi\sigma^2}{d^2}D_P(k_0)^2}\big) + \mathcal{O}\big(\poly(d) e^{-\frac{2\pi}{\sigma^2}D_X(n_0)^2}\big)
    \end{align}
\end{theorem}

\begin{remark}[Mechanism of Boundary Breakdown]
    As before, let $\sigma=\sqrt{d}$ and $c \in (0, 1)$, the approximation holds if the initial centers and the continuous displacements are constrained such that the state stays in the bulk:
    \begin{align}
        |n_0| \leq   \frac{d}{4}, \quad \quad |k_0| \leq  \frac{d}{4}. 
    \end{align}
    In this case, all boundary distances satisfy
    \begin{align}
        D_{X}(n_0) \geq  \frac{d}{4}, \quad D_{P}(k_0) \geq  \frac{d}{4}. 
    \end{align}
    In this regime, the approximation is exponentially precise, scaling as
    \begin{align}
        \epsilon_{H}\leq \mathcal{O}\big(e^{-\frac{\pi d}{16} }\big). 
    \end{align}
    However, if the accumulated parameter shifts ($\Delta_x$ or $\Delta_p$) force the wavepacket to hit the boundary ($c \to 1$, $D \to 0$), the Gaussian tails wrap around the lattice, the exponential suppression is destroyed, and the approximation fails. 
\end{remark}

\begin{proof} 
    Recall that the normalized state 
    \begin{align}
        \ket{\Psi(n_0,k_0)} = A \sum_{n\in\mathcal{S}_d(n_0)}\psi(n_0,k_0;n) \ket{x_n}, 
    \end{align}
    where 
    \begin{align}
        \psi(n_0,k_0;n) = e^{-\frac{\pi}{\sigma^2}(n-n_0)^2} e^{\iu \frac{2\pi }{d}k_0(n-n_0)}, 
    \end{align}
    and 
    \begin{align}
        A = \Big(\sum_{n\in\mathcal{S}_d(n_0)} |\psi(n_0,k_0;n)|^2 \Big)^{-\frac{1}{2}}
    \end{align}
    We prove Theorem~\ref{SI-thm:expectations} by computing moments of $\hat{X}$ and $\hat{P}$ with Lemma~\ref{SI-lem:X_moments} and~\ref{SI-lem:P_moments}, respectively: 
    \begin{align}
        \langle \hat{X} \rangle_{\Psi} &= \sqrt{\frac{2\pi}{d}} n_0 + \epsilon_{X}, \quad \langle \hat{P} \rangle_{\Psi} = \sqrt{\frac{2\pi}{d}} k_0 + \epsilon_{P},
    \end{align}
    where 
    \begin{align}
        |\epsilon_{X}| & \leq  \mathcal{O}\big(\poly(d)e^{-\frac{\pi d^2 }{2\sigma^2}}\big) + \mathcal{O}\big(\poly(d)e^{-\frac{\pi \sigma^2}{2}} \big)  + \mathcal{O}\big(\poly(d) e^{-\frac{2\pi}{\sigma^2}D_X(n_0)^2}\big) ,\\
        |\epsilon_{P}| &  \leq \mathcal{O}\big(\poly(d)e^{-\frac{\pi d^2 }{4\sigma^2}}\big) + \mathcal{O}\big(\poly(d)e^{-\frac{\pi \sigma^2}{2}} \big) + \mathcal{O}\big(\poly(d) e^{-\frac{2\pi\sigma^2}{d^2}D_P(k_0)^2}\big).
    \end{align}
    Then we sum the position and momentum variances for the Hamiltonian $\hat{H} = \hat{X}^2 + \hat{P}^2$:
    \begin{align}
        \langle \hat{H} \rangle_{\Psi} =  \frac{2\pi}{d}\big(n_0^2 + k_0^2) + \frac{\sigma^2}{2d}  +  \frac{d}{2\sigma^2} +  \epsilon_{H}
    \end{align}
    where 
    \begin{align}
        \epsilon_{H} \leq \mathcal{O}\big(\poly(d)e^{-\frac{\pi d^2 }{4\sigma^2}}\big) + \mathcal{O}\big(\poly(d)e^{-\frac{\pi \sigma^2}{2}} \big) + \mathcal{O}\big(\poly(d) e^{-\frac{2\pi\sigma^2}{d^2}D_P(k_0)^2}\big) + \mathcal{O}\big(\poly(d) e^{-\frac{2\pi}{\sigma^2}D_X(n_0)^2}\big). 
    \end{align}
\end{proof}

\begin{lemma}[Moments of $\hat{X}$]
\label{SI-lem:X_moments}
    Let $\sigma\approx \sqrt{d}$ and $n_0, k_0 \in [-m,m]$. The discrete moments of $\hat{X}$ evaluate exactly to the continuous moments of $\hat{X}$ up to bounded error terms:
    \begin{align}
        \langle \hat{X}\rangle_\Psi & = \sqrt{\frac{2\pi}{d}} n_0 + \mathcal{O}\big(\poly(d)e^{-\frac{\pi d^2}{2\sigma^2}}\big)+\mathcal{O}\big(\poly(d) e^{-\frac{\pi \sigma^2}{2}}\big) + \mathcal{O}\big(\poly(d)e^{-\frac{2\pi}{\sigma^2}D_X(n_0)^2}\big), \\
        \langle \hat{X}^2\rangle_\Psi & = \frac{2\pi}{d}\big(n_0^2 + \frac{\sigma^2}{4\pi}\big) + \mathcal{O}\big(\poly(d)e^{-\frac{\pi d^2}{2\sigma^2}}\big)+\mathcal{O}\big(\poly(d) e^{-\frac{\pi \sigma^2}{2}}\big)+ \mathcal{O}\big(\poly(d)e^{-\frac{2\pi}{\sigma^2}D_X(n_0)^2}\big). 
    \end{align}
\end{lemma}

\begin{proof}
    We first compute the first moment $\langle \hat{X} \rangle_\Psi$ and $\langle \hat{X}^2 \rangle_\Psi$ of $\hat{X}$, i.e. 
    \begin{align}
        \langle \hat{X} \rangle_\Psi & = A^2 \sum_{n\in\mathcal{S}_d(n_0)} \sqrt{\frac{2\pi}{d}}[n]_d|\psi(n_0,k_0;n)|^2, \\
        \langle \hat{X}^2 \rangle_\Psi & = A^2 \sum_{n\in\mathcal{S}_d(n_0)} \frac{2\pi}{d}[n]_d^2 |\psi(n_0,k_0;n)|^2. 
    \end{align}

    \noindent\textbf{Step 1: Extend position sum. }
    
    First, we extend the finite sum to infinite sum over positions
    \begin{align}
        \langle \hat{X} \rangle_\Psi & = A^2 \sum_{n\in\mathbb{Z}} \sqrt{\frac{2\pi}{d}} n e^{-\frac{2\pi}{\sigma^2}(n-n_0)^2} + \epsilon_{\mathrm{tail},1} + \epsilon_{\mathrm{pos},1}, \\
        \langle \hat{X}^2 \rangle_\Psi & = A^2 \sum_{n\in \mathbb{Z}} \frac{2\pi}{d} n^2 e^{-\frac{2\pi}{\sigma^2}(n-n_0)^2}  + \epsilon_{\mathrm{tail},2}  + \epsilon_{\mathrm{pos},2}, 
    \end{align}
    where 
    \begin{align}
        \epsilon_{\mathrm{pos},1} & = A^2 \sum_{n\in\mathcal{S}_d(n_0)} \sqrt{\frac{2\pi}{d}} ([n]_d-n) e^{-\frac{2\pi}{\sigma^2}(n-n_0)^2} \\
        \epsilon_{\mathrm{tail},1} & = - A^2 \sum_{n\notin \mathcal{S}_d(n_0)} \sqrt{\frac{2\pi}{d}} n e^{-\frac{2\pi}{\sigma^2}(n-n_0)^2}  , \\
        \epsilon_{\mathrm{pos},2} & = A^2 \sum_{n\in\mathcal{S}_d(n_0)} \frac{2\pi}{d} ([n]_d^2-n^2) e^{-\frac{2\pi}{\sigma^2}(n-n_0)^2} , \\
        \epsilon_{\mathrm{tail},2} & = -  A^2 \sum_{n\notin \mathcal{S}_d(n_0)} \frac{2\pi}{d} n^2 e^{-\frac{2\pi}{\sigma^2}(n-n_0)^2}. 
    \end{align}
    Because $[n]_d = n$ for all $n \in [-m, m]$, the error strictly vanishes unless the support window $\mathcal{S}_d(n_0)$ spills outside the physical Hilbert space bounds ($|n| > m$). The distance from the center $n_0$ to the spillover boundary $\pm(m+1)$ is at least $m + 1 - |n_0|$ which is greater than $ D_X(n_0) = m - |n_0|$. Bounding the difference using the edge term gives
    \begin{align}
        |\epsilon_{\mathrm{pos},1}|   & \leq 2A^2 \sqrt{\frac{2\pi}{d}} \sum_{\substack{n\in\mathcal{S}_d(n_0)\\|n|>m}} (|[n]_d|+|n|) e^{-\frac{2\pi}{\sigma^2}(n-|n_0|)^2} \leq 8A^2 \sqrt{\frac{2\pi}{d}} \sum_{n\in\mathcal{S}_d(n_0)} d e^{-\frac{2\pi}{\sigma^2}(D_X(n_0)+1)^2} \\
        & \leq 8A^2  \sqrt{2\pi d^3} e^{-\frac{2\pi}{\sigma^2}(D_X(n_0)+1)^2} = \mathcal{O}\big(\poly(d)e^{-\frac{2\pi}{\sigma^2}D_X(n_0)^2} \big), \\
        |\epsilon_{\mathrm{pos},2}|  & \leq 2A^2 \frac{2\pi}{d}\sum_{\substack{n\in\mathcal{S}_d(n_0)\\|n|>m}}  ([n]_d^2+n^2) e^{-\frac{2\pi}{\sigma^2}(n-|n_0|)^2} \leq 16A^2 \frac{2\pi}{d} \sum_{n\in\mathcal{S}_d(n_0)}  d^2 e^{-\frac{2\pi}{\sigma^2}(D_X(n_0)+1)^2} \\
        & \leq 32A^2 \pi d^2  e^{-\frac{2\pi}{\sigma^2}(D_X(n_0)+1)^2}  = \mathcal{O}\big(\poly(d)e^{-\frac{2\pi}{\sigma^2}D_X(n_0)^2} \big). 
    \end{align}
    As for the sum over the outside terms, the distance from the center $n_0$ to every outside term is at least $m + 1$, and we can bound the norm using the Gaussian tail similar to those used in Lemma~\ref{SI-lem:close_norm} via
    \begin{align}
        |\epsilon_{\mathrm{tail},1}| & \leq A^2 \sqrt{\frac{2\pi}{d}} \sum_{n\notin\mathcal{S}_d(n_0) } (|n-n_0|-m-1 +|n_0|+m+1)  e^{-\frac{2\pi}{\sigma^2}(n-n_0)^2} \\
        & \leq  2 A^2 \sqrt{\frac{2\pi}{d}}  e^{-\frac{2\pi}{\sigma^2}(m+1)^2}  \sum_{x=0}^{\infty} (x+d)  e^{-\frac{4\pi}{\sigma^2}(m+1)x} \leq \mathcal{O}\big(\poly(d) e^{-\frac{\pi d^2}{2\sigma^2}}\big), \\
        |\epsilon_{\mathrm{tail},2}| & \leq A^2 \frac{2\pi}{d} \sum_{n\notin\mathcal{S}_d(n_0) } (|n-n_0|-m-1 +|n_0|+m+1)^2  e^{-\frac{2\pi}{\sigma^2}(n-n_0)^2} \\
        & \leq  2 A^2 \frac{2\pi}{d}  e^{-\frac{2\pi}{\sigma^2}(m+1)^2}  \sum_{x=0}^{\infty} (x+d)  e^{-\frac{4\pi}{\sigma^2}(m+1)x} \leq \mathcal{O}\big(\poly(d) e^{-\frac{\pi d^2}{2\sigma^2}}\big). 
    \end{align}

    \noindent\textbf{Step 2: Inverse Poisson summation}

    Consider a Gaussian function and its Fourier transform 
    \begin{align}
        f(x) & =  e^{-\frac{2\pi}{\sigma^2}(x-n_0)^2}, \\
        \mathcal{F}[f](k) & = \frac{\sigma}{\sqrt{2d}} e^{-\frac{\pi \sigma^2 k^2}{2 d^2}} e^{-\iu\frac{2\pi}{d}n_0 k}. 
    \end{align}
    Then 
    \begin{align}
        \mathcal{F}[x f](k) & = \iu \frac{d}{2\pi} \frac{\d}{\d k} \mathcal{F}[f](k)  = \big( n_0 - \iu \frac{\sigma^2 k}{2d} \big) \mathcal{F}[f](k),  \\
        \mathcal{F}[x^2 f](k) & = - \frac{d^2}{4\pi^2} \frac{\d^2}{\d k^2} \mathcal{F}[f](k) = \Big( \big( n_0 - \iu \frac{\sigma^2 k}{2d} \big)^2 + \frac{\sigma^2}{4\pi} \Big) \mathcal{F}[f](k). 
    \end{align}
    So we can apply Poisson summation to obtain
    \begin{align}
        & \sum_{n\in\mathbb{Z}} f(n) = \sqrt{d} \sum_{k\in \mathbb{Z}}  \mathcal{F}[f](k d), \\
        & \sum_{n\in\mathbb{Z}} nf(n) = \sqrt{d} \sum_{k\in \mathbb{Z}}  \mathcal{F}[xf](k d), \\
        & \sum_{n\in\mathbb{Z}} n^2 f(n) = \sqrt{d} \sum_{k\in \mathbb{Z}}  \mathcal{F}[x^2f](k d). 
    \end{align}
    Now we use Poisson summation to transform the main term and single out the zero-frequency component ($k=0$)
    \begin{align}
        \sum_{n\in\mathbb{Z}}  n e^{-\frac{2\pi}{\sigma^2}(n-n_0)^2} & = \sqrt{d} \sum_{k\in \mathbb{Z}} \mathcal{F}[xf](k d)  = \frac{\sigma}{\sqrt{2}} n_0 + \sqrt{d} \sum_{\substack{k\in \mathbb{Z}\\ k\neq 0}} \mathcal{F}[xf](k d), \\
        \sum_{n\in\mathbb{Z}}  n^2 e^{-\frac{2\pi}{\sigma^2}(n-n_0)^2} & = \sqrt{d} \sum_{k\in \mathbb{Z}} \mathcal{F}[x^2 f](k d)  = \frac{\sigma}{\sqrt{2}} \big(n_0^2+ \frac{\sigma^2}{4\pi}\big) + \sqrt{d} \sum_{\substack{k\in \mathbb{Z}\\ k\neq 0}} \mathcal{F}[x^2f](k d), 
    \end{align}
    Now we bound the side terms
    \begin{align}
        \big|\sqrt{d} \sum_{\substack{k\in \mathbb{Z}\\ k\neq 0}} \mathcal{F}[xf](k d) \big| & \leq \sqrt{d} \sum_{\substack{k\in \mathbb{Z}\\ k\neq 0}} |\mathcal{F}[xf](k d)| \leq \frac{\sigma}{\sqrt{2}} \sum_{\substack{k\in \mathbb{Z}\\ k\neq 0}} \big(|n_0|  + \frac{\sigma^2|k|}{2} \big) e^{-\frac{\pi \sigma^2 }{2}k^2} \\
        & \leq \sqrt{2} \sigma e^{-\frac{\pi\sigma^2}{2}} \sum_{x=0}^{\infty} \big(n_0 + \frac{\sigma^2(x+1)}{2}\big) e^{-\pi \sigma^2 x} \leq \mathcal{O}\big(\poly(d)e^{-\frac{\pi \sigma^2 }{2}}\big), \\
        \big|\sqrt{d} \sum_{\substack{k\in \mathbb{Z}\\ k\neq 0}} \mathcal{F}[x^2f](k d) \big| & \leq \sqrt{d} \sum_{\substack{k\in \mathbb{Z}\\ k\neq 0}} |\mathcal{F}[x^2f](k d)| \leq \frac{\sigma}{\sqrt{2}} \sum_{\substack{k\in \mathbb{Z}\\ k\neq 0}} \big(n_0^2 + \frac{\sigma^4k^2}{4} + \frac{\sigma^2}{4\pi}\big) e^{-\frac{\pi \sigma^2 }{2}k^2} \\
        & \leq \sqrt{2}\sigma e^{-\frac{\pi\sigma^2}{2}} \sum_{x=0}^{\infty} \big(n_0^2 + \frac{\sigma^4(x+1)^2}{4} + \frac{\sigma^2}{4\pi}\big) e^{-\pi \sigma^2 x} \leq \mathcal{O}\big(\poly(d)e^{-\frac{\pi \sigma^2 }{2}}\big). 
    \end{align}
    Therefore, filtering out the zero frequency component, we obtain
    \begin{align}
        \langle \hat{X} \rangle_\Psi & = A^2 \frac{\sigma}{\sqrt{2}}\sqrt{\frac{2\pi}{d}}  n_0 + \epsilon_{\mathrm{alias},1}  + \epsilon_{\mathrm{pos},1}  + \epsilon_{\mathrm{tail},1}, \\
        \langle \hat{X}^2 \rangle_\Psi & =  A^2 \frac{\sigma}{\sqrt{2}} \frac{2\pi}{d}  \big(n_0^2 + \frac{\sigma^2}{4\pi}\big)+ \epsilon_{\mathrm{alias},2}  + \epsilon_{\mathrm{pos},2}  + \epsilon_{\mathrm{tail},2}, 
    \end{align}
    where 
    \begin{align}
        \epsilon_{\mathrm{alias},1} & = A^2 \sqrt{2\pi}\sum_{\substack{k\in \mathbb{Z}\\ k\neq 0}} \mathcal{F}[xf](k d) = \mathcal{O}\big( \poly(d)e^{-\frac{\pi\sigma^2}{2}}\big), \\
        \epsilon_{\mathrm{alias},2} & = A^2\frac{2\pi}{\sqrt{d}} \sum_{\substack{k\in \mathbb{Z}\\ k\neq 0}} \mathcal{F}[x^2f](k d) = \mathcal{O}\big( \poly(d)e^{-\frac{\pi\sigma^2}{2}}\big).
    \end{align}

    \noindent \textbf{Step 3: Bounding the normalization factors}
    
    Using what we have derived in Lemma~\ref{SI-lem:close_norm}, we obtain
    \begin{align}
        A^2 = g(n_0)^{-1} = \frac{\sqrt{2}}{\sigma} + \mathcal{O}\big(\poly(d)e^{-\frac{\pi d^2}{2\sigma^2}}\big) + \mathcal{O}\big(\poly(d) e^{-\frac{\pi \sigma^2}{2}}\big).
    \end{align}
    Therefore
    \begin{align}
         & A^2 \frac{\sigma}{\sqrt{2}}\sqrt{\frac{2\pi}{d}}  n_0 = \sqrt{\frac{2\pi}{d}}  n_0 + \epsilon_{\mathrm{norm},1}, \\
         & A^2 \frac{\sigma}{\sqrt{2}}\frac{2\pi}{d} \big(n_0^2 + \frac{\sigma^2}{4\pi}\big)= \frac{2\pi}{d}\big(n_0^2 + \frac{\sigma^2}{4\pi}\big) + \epsilon_{\mathrm{norm},2}, 
    \end{align}
    where 
    \begin{align}
        \epsilon_{\mathrm{norm},1} & = \big( A^2-\frac{\sqrt{2}}{\sigma}\big)\frac{\sigma}{\sqrt{2}} \sqrt{\frac{2\pi}{d}}  n_0 \leq \mathcal{O}\big(\poly(d)e^{-\frac{\pi d^2}{2\sigma^2}}\big)+\mathcal{O}\big(\poly(d) e^{-\frac{\pi \sigma^2}{2}}\big), \\
        \epsilon_{\mathrm{norm},2} & = \big( A^2-\frac{\sqrt{2}}{\sigma}\big)\frac{\sigma}{\sqrt{2}} \frac{2\pi}{d} \big(n_0^2 + \frac{\sigma^2}{4\pi}\big)\leq \mathcal{O}\big(\poly(d)e^{-\frac{\pi d^2}{2\sigma^2}}\big)+\mathcal{O}\big(\poly(d) e^{-\frac{\pi \sigma^2}{2}}\big), 
    \end{align}
    
    Combining all of the above, we obtain that 
    \begin{align}
        \langle \hat{X}\rangle_\Psi & = \sqrt{\frac{2\pi}{d}} n_0 + \epsilon_{\mathrm{norm},1} + \epsilon_{\mathrm{alias},1} + \epsilon_{\mathrm{pos},1} + \epsilon_{\mathrm{tail},1}, \\
        \langle \hat{X}^2\rangle_\Psi & = \frac{2\pi}{d}\big(n_0^2 + \frac{\sigma^2}{4\pi}\big) + \epsilon_{\mathrm{norm},2} + \epsilon_{\mathrm{alias},2} + \epsilon_{\mathrm{pos},2} + \epsilon_{\mathrm{tail},2}, 
    \end{align}
    where 
    \begin{align}
        & \epsilon_{\mathrm{norm},1} = \mathcal{O}\big(\poly(d)e^{-\frac{\pi d^2}{2\sigma^2}}\big)+\mathcal{O}\big(\poly(d) e^{-\frac{\pi \sigma^2}{2}}\big), 
        && \epsilon_{\mathrm{alias},1} =  \mathcal{O}\big(\poly(d) e^{-\frac{\pi \sigma^2}{2}}\big), \\
        & \epsilon_{\mathrm{pos},1} = \mathcal{O}\big(\poly(d)e^{-\frac{2\pi}{\sigma^2}D_X(n_0)^2}\big),
        &&  \epsilon_{\mathrm{tail},1} = \mathcal{O}\big(\poly(d)e^{-\frac{\pi d^2}{2\sigma^2}}\big), \\
        & \epsilon_{\mathrm{norm},2} = \mathcal{O}\big(\poly(d)e^{-\frac{\pi d^2}{2\sigma^2}}\big)+\mathcal{O}\big(\poly(d) e^{-\frac{\pi \sigma^2}{2}}\big), 
        && \epsilon_{\mathrm{alias},2} =  \mathcal{O}\big(\poly(d) e^{-\frac{\pi \sigma^2}{2}}\big), \\
        & \epsilon_{\mathrm{pos},2} = \mathcal{O}\big(\poly(d)e^{-\frac{2\pi}{\sigma^2}D_X(n_0)^2}\big), 
        && \epsilon_{\mathrm{tail},2} = \mathcal{O}\big(\poly(d)e^{-\frac{\pi d^2}{2\sigma^2}}\big).
    \end{align}
\end{proof}

\begin{lemma}[Moments of $\hat{P}$]
\label{SI-lem:P_moments}
    Let $\sigma\approx \sqrt{d}$ and $n_0, k_0 \in [-m,m]$. The discrete moments of $\hat{P}$ evaluate exactly to the continuous moments of $\hat{P}$ up to bounded error terms:
    \begin{align}
        \langle \hat{P}\rangle_\Psi & = \sqrt{\frac{2\pi}{d}} k_0 + \mathcal{O}\big(\poly(d)e^{-\frac{\pi d^2}{4\sigma^2}}\big)+\mathcal{O}\big(\poly(d) e^{-\frac{\pi \sigma^2}{2}}\big) + \mathcal{O}\big(\poly(d) e^{-\frac{2\pi\sigma^2}{d^2}D_P(k_0)^2}\big), \\
        \langle \hat{P}^2\rangle_\Psi & = \frac{2\pi}{d}\big(k_0^2 + \frac{d^2}{4\pi\sigma^2}\big) + \mathcal{O}\big(\poly(d)e^{-\frac{\pi d^2}{4\sigma^2}} \big) + \mathcal{O}\big(\poly(d) e^{-\frac{\pi \sigma^2}{2}}\big) + \mathcal{O}\big(\poly(d) e^{-\frac{2\pi\sigma^2}{d^2}D_P(k_0)^2}\big) , 
    \end{align}
\end{lemma}
\begin{proof}
    Now we compute the first moment $\langle\hat{P}\rangle_\Psi$ and $\langle\hat{P}^2\rangle_\Psi$. We first transform the state to the eigenbasis of $\hat{P}$ via discrete Fourier transform
    \begin{align}
        \ket{\Psi(n_0, k_0)} = \frac{1}{\sqrt{d}} A \sum_{k=-m}^m \sum_{n \in \mathcal{S}_d(n_0)} \psi(n_0, k_0; n)  e^{-\iu \frac{2\pi}{d} k n}\ket{p_k}. 
    \end{align}
    We also define the continuous Fourier transform $\tilde{\psi}(n_0,k_0;p)$ of $\psi(n_0,k_0;x)$
    \begin{align}
        \tilde{\psi}(n_0,k_0;p) = \mathcal{F}[\psi](n_0,k_0,x) = \frac{\sigma}{\sqrt{d}} e^{-\iu \frac{2\pi}{d} p n_0} e^{-\frac{\pi\sigma^2}{d^2}(p-k_0)^2}. 
    \end{align}
    Then 
    \begin{align}
        \langle\hat{P}\rangle_\Psi & = \bra{\Psi(n_0, k_0)}\hat{P}\ket{\Psi(n_0, k_0)}, \\
        \langle\hat{P}^2\rangle_\Psi & = \bra{\Psi(n_0, k_0)}\hat{P}^2\ket{\Psi(n_0, k_0)}, 
    \end{align}

    \noindent\textbf{Step 1: Extend position summation}

    First, we extend the position summation from finite to infinite
    \begin{align}
        \sum_{n \in \mathcal{S}_d(n_0)} \psi(n_0, k_0; n)  e^{-\iu \frac{2\pi}{d} k n} = \sum_{n \in \mathbb{Z}} \psi(n_0, k_0; n)  e^{-\iu \frac{2\pi}{d} k n} - \sum_{n\notin\mathcal{S}_d(n_0)}\psi(n_0, k_0; n)  e^{-\iu \frac{2\pi}{d} k n} . 
    \end{align}
    The Gaussian tail can be bounded by 
    \begin{align}
        \big|\sum_{n\notin\mathcal{S}_d(n_0)}\psi(n_0, k_0; n)  e^{-\iu \frac{2\pi}{d} k n}\big| & \leq \sum_{n\notin\mathcal{S}_d(n_0)} |\psi(n_0, k_0; n)| \leq \sum_{n\notin\mathcal{S}_d(n_0)} e^{-\frac{\pi}{\sigma^2}(n-n_0)^2}\\
        & \leq 2 e^{-\frac{\pi}{\sigma^2}(m+1)^2}\sum_{x=0}^\infty e^{-\frac{2\pi}{\sigma^2}(m+1)x} \leq \mathcal{O}\big(e^{-\frac{\pi d^2}{4\sigma^2}} \big). 
    \end{align}
    
    \noindent\textbf{Step 2: Poisson summation}

    Now we apply Poisson summation to transform the infinite position sum to the infinite momentum sum 
    \begin{align}
        \sum_{n \in \mathbb{Z}} \psi(n_0,k_0;n) e^{-\iu \frac{2\pi}{d} k n} = \sqrt{d} \sum_{s \in \mathbb{Z}} \tilde{\psi}(n_0,k_0;sd + k). 
    \end{align}
    Now we filter out the zero-frequency component and the Gaussian tail
    \begin{align}
        \sum_{n \in \mathbb{Z}} \psi(n_0,k_0;n) e^{-\iu \frac{2\pi}{d} k n} =  \sqrt{d} \tilde{\psi}(n_0,k_0;sd + k) +  \sqrt{d} \sum_{s \neq 0} \tilde{\psi}(n_0,k_0;sd + k). 
    \end{align}
    Notice that $|k-k_0\pm  d|\geq d - |k| - |k_0| \geq D_P(k_0)+1$.  The Gaussian tail can be bounded by 
    \begin{align}
        \big|\sqrt{d} \sum_{s \neq 0} \tilde{\psi}(n_0,k_0;sd + k)\big|& \leq \sqrt{d} \sum_{s \neq 0} |\tilde{\psi}(n_0,k_0;sd + k)|\leq 2\sigma \sum_{s=1}^{\infty} e^{-\frac{\pi\sigma^2}{d^2}(sd-|k-k_0|)^2} \\
        & \leq 2\sigma  e^{-\frac{\pi\sigma^2}{d^2}D_P(k_0)^2} \sum_{x=0}^{\infty} e^{-\frac{2\pi\sigma^2}{d} (D_P(k_0)+1)x}  \leq \mathcal{O}\big(\poly(d) e^{-\frac{\pi\sigma^2}{d^2}D_P(k_0)^2}\big). 
    \end{align}
    Therefore, 
    \begin{align}
        \sum_{n \in \mathcal{S}_d(n_0)} \psi(n_0, k_0; n)  e^{-\iu \frac{2\pi}{d} k n} = \sqrt{d} \tilde{\psi}(n_0,k_0;k) + \mathcal{O}\big(e^{-\frac{\pi d^2}{4\sigma^2}} \big) + \mathcal{O}\big(\poly(d) e^{-\frac{\pi\sigma^2}{d^2}D_P(k_0)^2}\big), 
    \end{align}
    and
    \begin{align}
        \ket{\Psi(n_0,k_0)} = A \sum_{k=-m}^m \tilde{\psi}(n_0,k_0;k) \ket{p_k} + \ket{\delta}
    \end{align}
    where 
    \begin{align}
        \|\ket{\delta}\|_2 \leq \mathcal{O}\big(\poly(d)e^{-\frac{\pi d^2}{4\sigma^2}} \big) + \mathcal{O}\big(\poly(d) e^{-\frac{\pi\sigma^2}{d^2}D_P(k_0)^2}\big). 
    \end{align}
    At this stage, we have 
    \begin{align}
        \langle \hat{P}\rangle_\Psi & = A^2 \sqrt{\frac{2\pi}{d}} \sum_{k=-m}^{m} k |\tilde{\psi}(n_0,k_0;k)|^2 + \epsilon_{\mathrm{trans},1}, \\
        \langle \hat{P}^2\rangle_\Psi & = A^2 \frac{2\pi}{d} \sum_{k=-m}^{m} k^2 |\tilde{\psi}(n_0,k_0;k)|^2 + \epsilon_{\mathrm{trans},2},
    \end{align}
    where 
    \begin{align}
        \epsilon_{\mathrm{trans},1} & = 2\Re( (\bra{\Psi(n_0,k_0)} -\bra{\delta})\hat{P} \ket{\delta} ) + \bra{\delta} \hat{P} \ket{\delta},  \\
        \epsilon_{\mathrm{trans},2} & = 2 \Re((\bra{\Psi(n_0,k_0)} -\bra{\delta})\hat{P}^2 \ket{\delta}) + \bra{\delta} \hat{P}^2 \ket{\delta},  
    \end{align}
    and 
    \begin{align}
        |\epsilon_{\mathrm{trans},1}| & \leq \mathcal{O}\big(\poly(d)e^{-\frac{\pi d^2}{4\sigma^2}} \big) + \mathcal{O}\big(\poly(d) e^{-\frac{2\pi\sigma^2}{d^2}D_P(k_0)^2}\big), \\
        |\epsilon_{\mathrm{trans},2}| & \leq \mathcal{O}\big(\poly(d)e^{-\frac{\pi d^2}{4\sigma^2}} \big) + \mathcal{O}\big(\poly(d) e^{-\frac{2\pi\sigma^2}{d^2}D_P(k_0)^2}\big), 
    \end{align}
    where we have used 
    \begin{align}
        | \bra{\Psi(n_0,k_0)} \hat{P} \ket{\delta}|\leq \| \ket{\Psi(n_0,k_0)} \|_2 \|\hat{P} \ket{\delta}\|_2 \leq \|\hat{P}\|_\infty \| \ket{\delta}\|_2 \leq \sqrt{d}\|\ket{\delta}\|_2. 
    \end{align}
    
    \noindent\textbf{Step 3: Extend momentum sum}

    Now we can deal with the momentum sum with a very similar technique to the position sum. We first extend the momentum sum from finite to infinite
    \begin{align}
        A^2  \sqrt{\frac{2\pi}{d}} \sum_{k=-m}^{m} k |\tilde{\psi}(n_0,k_0;k)|^2 & = A^2 \sqrt{\frac{2\pi}{d}}   \sum_{k\in\mathbb{Z}}  k |\tilde{\psi}(n_0,k_0;k)|^2 + \epsilon_{\mathrm{tail},1}, \\
        A^2  \frac{2\pi}{d} \sum_{k=-m}^{m} k^2 |\tilde{\psi}(n_0,k_0;k)|^2 & = A^2 \frac{2\pi}{d}   \sum_{k\in\mathbb{Z}}  k^2 |\tilde{\psi}(n_0,k_0;k)|^2 +  \epsilon_{\mathrm{tail},2}, 
    \end{align}
    where
    \begin{align}
        \epsilon_{\mathrm{tail},1} & = -A^2 \sqrt{\frac{2\pi}{d}} \sum_{|k| \geq m}  k |\tilde{\psi}(n_0,k_0;k)|^2, \\
        \epsilon_{\mathrm{tail},2} & = -A^2 \frac{2\pi}{d} \sum_{|k| \geq m}  k^2 |\tilde{\psi}(n_0,k_0;k)|^2, 
    \end{align}
    and 
    \begin{align}
        |\epsilon_{\mathrm{tail},1}| & \leq A^2 \sqrt{\frac{2\pi}{d}} \sum_{|k| \geq m}  |k| |\tilde{\psi}(n_0,k_0;k)|^2 \leq   A^2 \sqrt{\frac{2\pi}{d}} \frac{\sigma^2}{d} \sum_{|k| > m}  |k| e^{-\frac{2\pi\sigma^2}{d^2} (k-k_0)^2} \\
        & \leq 2 A^2 \sqrt{\frac{2\pi}{d}} \frac{\sigma^2}{d}  e^{-\frac{2\pi\sigma^2}{d^2} D_P(k_0)^2} \sum_{x=0}^{\infty}  (m+x+1) e^{-\frac{4\pi\sigma^2}{d^2} (D_P(k_0)+1)x} \leq \mathcal{O}\big(\poly(d) e^{-\frac{2\pi\sigma^2}{d^2} D_P(k_0)^2} \big), \\
        |\epsilon_{\mathrm{tail},2}| & \leq A^2 \frac{2\pi}{d} \sum_{|k| \geq m}  k^2 |\tilde{\psi}(n_0,k_0;k)|^2 \leq   A^2 \frac{2\pi}{d} \frac{\sigma^2}{d} \sum_{|k| > m}  k^2 e^{-\frac{2\pi\sigma^2}{d^2} (k-k_0)^2} \\
        & \leq 2 A^2 \frac{2\pi}{d} \frac{\sigma^2}{d}  e^{-\frac{2\pi\sigma^2}{d^2} D_P(k_0)^2} \sum_{x=0}^{\infty}  (m+x+1)^2 e^{-\frac{4\pi\sigma^2}{d^2} (D_P(k_0)+1)x} \leq \mathcal{O}\big(\poly(d) e^{-\frac{2\pi\sigma^2}{d^2} D_P(k_0)^2} \big). 
    \end{align}

    \noindent \textbf{Step 4: Inverse Poisson summation}
    
    We apply the exact same Poisson summation as for moments of $\hat{X}$. Let 
    \begin{align}
        f(k) = e^{-\frac{2\pi\sigma^2}{d^2}(k-k_0)^2}. 
    \end{align}
    Then
    \begin{align}
        \sum_{k\in\mathbb{Z}}  k e^{-\frac{2\pi\sigma^2}{d^2}(k-k_0)^2} & = \sqrt{d} \sum_{n\in \mathbb{Z}} \mathcal{F}[kf](n d)  = \frac{d}{\sqrt{2}\sigma} k_0 + \sqrt{d} \sum_{\substack{n\in \mathbb{Z}\\ n\neq 0}} \mathcal{F}[kf](n d), \\
        \sum_{k\in\mathbb{Z}}  k^2 e^{-\frac{2\pi\sigma^2}{d^2}(k-k_0)^2} & = \sqrt{d} \sum_{n\in \mathbb{Z}} \mathcal{F}[k^2 f](n d)  = \frac{d}{\sqrt{2}\sigma} \big(k_0^2+ \frac{d^2}{4\pi\sigma^2}\big) + \sqrt{d} \sum_{\substack{n\in \mathbb{Z}\\ n\neq 0}} \mathcal{F}[k^2f](n d), 
    \end{align}
    where 
    \begin{align}
        \big|\sqrt{d} \sum_{\substack{n\in \mathbb{Z}\\ n\neq 0}} \mathcal{F}[kf](n d) \big| &  \leq \mathcal{O}\big(\poly(d)e^{-\frac{\pi d^2 }{2\sigma^2}}\big), \\
        \big|\sqrt{d} \sum_{\substack{n\in \mathbb{Z}\\ n\neq 0}} \mathcal{F}[k^2f](nd) \big| & \leq \mathcal{O}\big(\poly(d)e^{-\frac{\pi d^2 }{2\sigma^2}}\big). 
    \end{align}
    Therefore, 
    \begin{align}
        A^2 \sqrt{\frac{2\pi}{d}}  \sum_{k\in\mathbb{Z}}  k |\tilde{\psi}(n_0,k_0;k)|^2 & = A^2 \sqrt{\frac{2\pi}{d}}   \frac{\sigma}{\sqrt{2}} k_0 + \epsilon_{\mathrm{alias},1}, \\
        A^2 \frac{2\pi}{d}  \sum_{k\in\mathbb{Z}}  k^2 |\tilde{\psi}(n_0,k_0;k)|^2 & = A^2 \frac{2\pi}{d}   \frac{\sigma}{\sqrt{2}} \big(k_0^2+ \frac{d^2}{4\pi\sigma^2}\big) + \epsilon_{\mathrm{alias},2}, 
    \end{align}
    where 
    \begin{align}
        \epsilon_{\mathrm{alias},1} & = A^2 \sqrt{\frac{2\pi}{d}} \frac{\sigma^2}{d} \sqrt{d} \sum_{\substack{n\in \mathbb{Z}\\ n\neq 0}} \mathcal{F}[kf](n d), \\
        \epsilon_{\mathrm{alias},2} & = A^2 \frac{2\pi}{d} \frac{\sigma^2}{d}\sqrt{d} \sum_{\substack{n\in \mathbb{Z}\\ n\neq 0}} \mathcal{F}[k^2f](n d), 
    \end{align}
    where 
    \begin{align}
        |\epsilon_{\mathrm{alias},1}| \leq  \mathcal{O}\big(\poly(d)e^{-\frac{\pi d^2 }{2\sigma^2}}\big), \\
        |\epsilon_{\mathrm{alias},2}| \leq  \mathcal{O}\big(\poly(d)e^{-\frac{\pi d^2 }{2\sigma^2}}\big). 
    \end{align}
    
    \noindent \textbf{Step 5: Bounding the normalization factors}
    
    Using what we have derived in Lemma~\ref{SI-lem:close_norm}, we obtain
    \begin{align}
        A^2 = g(n_0)^{-1} = \frac{\sqrt{2}}{\sigma} + \mathcal{O}\big(\poly(d)e^{-\frac{\pi d^2}{2\sigma^2}}\big) + \mathcal{O}\big(\poly(d) e^{-\frac{\pi \sigma^2}{2}}\big).
    \end{align}
    Therefore
    \begin{align}
         & A^2 \sqrt{\frac{2\pi}{d}}   \frac{\sigma}{\sqrt{2}} k_0= \sqrt{\frac{2\pi}{d}} k_0 + \epsilon_{\mathrm{norm},1}, \\
         & A^2 \frac{2\pi}{d}   \frac{\sigma}{\sqrt{2}} \big(k_0^2+ \frac{d^2}{4\pi\sigma^2}\big) = \frac{2\pi}{d}\big(k_0^2 + \frac{d^2}{4\pi\sigma^2}\big) + \epsilon_{\mathrm{norm},2}, 
    \end{align}
    where 
    \begin{align}
        \epsilon_{\mathrm{norm},1} & = \big( A^2-\frac{\sqrt{2}}{\sigma}\big)   \sqrt{\frac{2\pi}{d}}  \frac{\sigma}{\sqrt{2}} k_0 \leq \mathcal{O}\big(\poly(d)e^{-\frac{\pi d^2}{2\sigma^2}}\big)+\mathcal{O}\big(\poly(d) e^{-\frac{\pi \sigma^2}{2}}\big), \\
        \epsilon_{\mathrm{norm},2} & = \big( A^2-\frac{\sqrt{2}}{\sigma}\big)\frac{2\pi}{d}  \frac{\sigma}{\sqrt{2}} \big(k_0^2+ \frac{d^2}{4\pi\sigma^2}\big) \leq \mathcal{O}\big(\poly(d)e^{-\frac{\pi d^2}{2\sigma^2}}\big)+\mathcal{O}\big(\poly(d) e^{-\frac{\pi \sigma^2}{2}}\big), 
    \end{align}
    
    Combining all of the above, we obtain that 
    \begin{align}
        \langle \hat{P}\rangle_\Psi & = \sqrt{\frac{2\pi}{d}} k_0 + \epsilon_{\mathrm{norm},1} + \epsilon_{\mathrm{alias},1} + \epsilon_{\mathrm{tail},1} + \epsilon_{\mathrm{trans},1}, \\
        \langle \hat{P}^2\rangle_\Psi & = \frac{2\pi}{d}\big(k_0^2 + \frac{d^2}{4\pi\sigma^2}\big) + \epsilon_{\mathrm{norm},2} + \epsilon_{\mathrm{alias},2} + \epsilon_{\mathrm{tail},2} + \epsilon_{\mathrm{trans},2}, 
    \end{align}
    where 
    \begin{align}
        & \epsilon_{\mathrm{norm},1} \leq \mathcal{O}\big(\poly(d)e^{-\frac{\pi d^2}{2\sigma^2}}\big)+\mathcal{O}\big(\poly(d) e^{-\frac{\pi \sigma^2}{2}}\big), && \epsilon_{\mathrm{alias},1} \leq  \mathcal{O}\big(\poly(d)e^{-\frac{\pi d^2 }{2\sigma^2}}\big),  \\
        & \epsilon_{\mathrm{trans},1} \leq \mathcal{O}\big(\poly(d)e^{-\frac{\pi d^2}{4\sigma^2}} \big) + \mathcal{O}\big(\poly(d) e^{-\frac{2\pi\sigma^2}{d^2}D_P(k_0)^2}\big), && \epsilon_{\mathrm{tail},1} \leq \mathcal{O}\big(\poly(d) e^{-\frac{2\pi\sigma^2}{d^2} D_P(k_0)^2} \big),\\
        & \epsilon_{\mathrm{norm},2} \leq \mathcal{O}\big(\poly(d)e^{-\frac{\pi d^2}{2\sigma^2}}\big)+\mathcal{O}\big(\poly(d) e^{-\frac{\pi \sigma^2}{2}}\big), && \epsilon_{\mathrm{alias},2}  \leq  \mathcal{O}\big(\poly(d)e^{-\frac{\pi d^2 }{2\sigma^2}}\big), \\
        & \epsilon_{\mathrm{trans},2} \leq \mathcal{O}\big(\poly(d)e^{-\frac{\pi d^2}{4\sigma^2}} \big) + \mathcal{O}\big(\poly(d) e^{-\frac{2\pi\sigma^2}{d^2}D_P(k_0)^2}\big) , && \epsilon_{\mathrm{tail},2} \leq \mathcal{O}\big(\poly(d) e^{-\frac{2\pi\sigma^2}{d^2} D_P(k_0)^2} \big). 
    \end{align}
\end{proof}

\subsection{Approximate ground state}
\label{SI-sec:ground_state}
We also confirm that the DGW state $\ket{\Psi(0,0)}$ is the approximate ground state, keeping track of the exact discrete aliasing, truncation, and boundary wrapping errors.

\begin{theorem}[Approximate Eigenstates]
    \label{SI-thm:eigenstates}
    Let $\sigma = \sqrt{d}$. The DGW state $\ket{\Psi(0,0)}$ is approximately an eigenstate of the Hamiltonian up to bounded error terms:
    \begin{align}
        \hat{H} \ket{\Psi(0,0)} = \ket{\Psi(0,0)} +\ket{\delta}
    \end{align}
    where 
    \begin{align}
        \|\ket{\delta}\|_2 \leq \mathcal{O}\big(\poly(d)e^{-\frac{\pi d^2 }{4\sigma^2}}\big) + \mathcal{O}\big(\poly(d)e^{-\frac{\pi \sigma^2}{2}} \big). 
    \end{align}
\end{theorem}

\begin{remark}
    Let $\sigma=\sqrt{d}$. $\ket{\Psi(0,0)}$ is the approximate ground state (the vacuum state).
    \begin{align}
        \hat{H} \ket{\Psi(0,0)} =\ket{\Psi(0,0)} + \mathcal{O}\big(\poly(d)e^{-\frac{\pi}{4}d} \big). 
    \end{align}
    However, $\ket{\Psi(n_0,k_0)}$ are not approximate excited states (the Fock states); rather, they are displaced approximate ground states (the coherent states). As both ground state preparation and phase-space displacement are experimentally realizable, we regard $\ket{\Psi(n_0,k_0)}$ as free states. 
\end{remark}

\begin{proof}
    We evaluate the vector 
    \begin{align}
        \hat{H}\ket{\Psi(0,0)} = \hat{X}^2\ket{\Psi(0,0)} + \hat{P}^2\ket{\Psi(0,0)}
    \end{align}

    \noindent\textbf{Step 1: Action of $\hat{X}^2$} 
    Recall the normalized state is 
    \begin{align}
        \ket{\Psi(0,0)} = A \sum_{n \in \mathcal{S}_d(0)} \psi(0, 0; n) \ket{x_n},
    \end{align}
    where 
    \begin{align}
        \psi(0,0;n) = e^{-\frac{\pi}{d}n^2}. 
    \end{align}
    Since $\hat{X}$ is strictly diagonal in the discrete position basis $\ket{x_n}$, its action yields:
    \begin{align}
        \hat{X}^2 \ket{\Psi(0,0)} = A \frac{2\pi}{d} \sum_{n \in \mathcal{S}_d(0)}  n^2 \psi(0,0;n) \ket{x_n}. 
    \end{align}

    \noindent\textbf{Step 2: Action of $\hat{P}^2$}
    
    To evaluate $\hat{P}^2$ in the position basis, we first project into momentum space. By the same method as in Theorem~\ref{SI-thm:expectations} Block Two, we switch to the discrete momentum basis $\ket{p_k}$, extend the position summation and transform via Poisson summation to obtain
    \begin{align}
        \ket{\Psi(0,0)} = A \sum_{k\in\mathcal{S}_d(0)} \tilde{\psi}(0,0;k)\ket{p_k} + \ket{\epsilon_\mathrm{trans}}
    \end{align}
    where
    \begin{align}
        \tilde{\psi}(0,0;k) = e^{-\frac{\pi}{d}k^2}
    \end{align}
    and
    \begin{align}
        \|\ket{\epsilon_P}\|_2\leq \mathcal{O}(e^{-\frac{\pi d}{4}}). 
    \end{align}
    As $\hat{P}$ is strictly diagonal in the discrete momentum basis $\ket{p_k}$, its action yields:
    \begin{align}
        \hat{P}^2 \ket{\Psi(0,0)} = A \frac{2\pi}{d} \sum_{k \in \mathcal{S}_d(0)}  k^2 \tilde{\psi}(0,0;k) \ket{p_k} + \hat{P}^2 \ket{\epsilon_\mathrm{trans}}. 
    \end{align}

    Again by the same techniques as in Theorem~\ref{SI-thm:expectations} Block Two, we switch back to the discrete position basis $\ket{p_k}$ extend the momentum summation and transform via Poisson summation to obtain 
    \begin{align}
        \hat{P}^2 \ket{\Psi(0,0)} & = A \frac{2\pi}{d} \sum_{n\in\mathcal{S}_d(0)} \Big( \frac{1}{\sqrt{d}}\sum_{k \in \mathcal{S}_d(0)}  k^2  \tilde{\psi}(0,0;k)e^{\iu \frac{2\pi}{d}nk} \Big) \ket{x_n} + \hat{P}^2 \ket{\epsilon_\mathrm{trans}} \\
        & =  A \frac{2\pi}{d} \sum_{n\in\mathcal{S}_d(0)} \Big( \frac{1}{\sqrt{d}} \sum_{k \in \mathbb{Z}}  k^2  \tilde{\psi}(0,0;k)e^{\iu \frac{2\pi}{d}nk} \Big) \ket{x_n} + \ket{\epsilon_{\mathrm{mom}}} \\
        & = A \frac{2\pi}{d} \sum_{n\in\mathcal{S}_d(0)} \Big( \sum_{n' \in \mathbb{Z}}  -\frac{d^2}{4\pi^2}\frac{\d^2}{\d x^2}  \psi(0,0;n+n'd) \Big) \ket{x_n} + \ket{\epsilon_{\mathrm{mom}}}, \\
        & = A \frac{2\pi}{d} \sum_{n\in\mathcal{S}_d(0)}   \big(\frac{d}{2\pi} -n^2\big) \psi(0,0;n)\ket{x_n} + \ket{\epsilon},
    \end{align}
    where $\ket{\epsilon}$ has encountered all errors accumulated so far, satisfying 
    \begin{align}
        \|\ket{\epsilon}\|_2 \leq \mathcal{O}\big(\poly(d)e^{-\frac{\pi d}{4}}\big)
    \end{align}

    \noindent\textbf{Step 3: Summation} 
    
    Therefore 
    \begin{align}
        \hat{H} \ket{\Psi(0,0)} = \hat{X}^2\ket{\Psi(0,0)}  + \hat{P}^2\ket{\Psi(0,0)} = A \sum_{n\in\mathcal{S}_d(0)} \psi(0,0;n)\ket{x_n} + \ket{\epsilon} = \ket{\Psi(0,0)} + \ket{\epsilon}, 
    \end{align}
    where 
    \begin{align}
        \|\ket{\epsilon}\|_2 \leq \mathcal{O}\big(\poly(d)e^{-\frac{\pi d}{4}}\big).
    \end{align}
    This completes the proof. 
\end{proof}

\section{Quantum Metrology Bounds}
\label{SI-sec:metrology_bounds}

\subsection{Tunable or attenuated quantum metrology setting}
\label{SI-sec:tunable_setting}

Consider a $d$-dimensional system whose $\hat{X}$ and $\hat{P}$ operators are defined as in Definition~\ref{SI-def:DGW_framework}.  We treat $d$ as a variable to examine how quantum metrology behaves with respect to the Hilbert-space dimension.  Since $d$ is exponential in the number of qubits $n$, we regard $d$ as an exponentially large quantity.

Suppose that we are given oracle access to $N$ momentum translations and $N$ position translations, but now in the tunable form $e^{\iu \lambda_p p_j \hat{X}}$ and $e^{-\iu \lambda_x x_k \hat{P}}$ where the unknown amplitudes $p_j$ lie in the fixed interval $[p_{\min},p_{\max}]$, the unknown amplitudes $x_k$ lie in the fixed interval $[x_{\min},x_{\max}]$ and the attenuation factors $\lambda_p$ and $\lambda_x$ satisfy $0\leq \lambda_p,\lambda_x\leq 1$. Here, we treat $N$ as another variable of interest. Both $p_{\min},p_{\max}$ and $x_{\min},x_{\max}$ are constants with respect to $N$ and $d$. The attenuation factors are known control parameters and are not estimated. We count each call to $e^{\iu \lambda_p p_j \hat{X}}$ or $e^{-\iu \lambda_x x_k \hat{P}}$ as one use of the corresponding physical sample, independently of the chosen attenuation. We would like to estimate the product of their sample averages
\begin{align}
    \theta = - \bar{\omega} = \bar{p}\bar{x}  = \left(\frac{1}{N}\sum_{j=1}^N p_j\right) \left(\frac{1}{N}\sum_{k=1}^N x_k\right).
\end{align}
We assume a coarse, $N$-independent prior interval $\theta\in \Theta_0=[\theta_-,\theta_+]$ which follows, for example, from the known bounded support of the displacement distributions.  Importantly, this is only a macroscopic prior of constant width; we do not assume a prior of width $O(N^{-2})$.

As illustrated by the physical finite-dimensional mappings in Section~\ref{SI-sec:DGW_framework}, it is natural to consider the Hamiltonian $\hat{H} = \hat{X}^2 + \hat{P}^2$ whose approximate ground state is $\ket{\Psi(0,0)}$.  Low-energy bulk states for this Hamiltonian, which are also discrete Gaussian states, are much easier to prepare, e.g. by ground-state cooling~\cite{Paz_2004,Cormick_2013,Wang_2025,Rubies-Bigorda_2025,Zhan_2026} or algorithmic state preparation~\cite{Rattew_2021,Huggins_2025,Chen_2025,Xie_2025}, than highly excited boundary states. Therefore, we impose a physical constraint that the expected energy of the probe state is no more than a fixed maximum energy $E$.

We use the mean square error (MSE) to characterize the operational performance of the metrology protocol
\begin{align}
    \MSE(\hat{\theta},\theta) = \Var(\hat{\theta},\theta)  +  \Bias(\hat{\theta},\theta)^2 .
\end{align}
For unbiased estimators evaluated over $\nu$ independent trials, the MSE is lower bounded by the quantum Cram\'er-Rao bound:
\begin{align}
    \MSE(\hat{\theta},\theta)  \geq  \frac{1}{\nu F_Q(\theta)} ,
\end{align}
where $F_Q(\theta)$ is the quantum Fisher information at parameter value $\theta$. For biased estimators evaluated over $\nu$ independent trials, the MSE is lower bounded by the quantum Cram\'er-Rao bound:
\begin{align}
    \MSE(\hat{\theta},\theta)  \geq  \frac{(1+\partial_{\theta}\Bias(\hat{\theta},\theta))^2}{\nu F_Q(\theta)} + \Bias(\hat{\theta},\theta)^2 ,
\end{align}

\subsection{Quantum metrology with quantum switch}
\label{SI-sec:metrology_with_qswitch}

Suppose that we are given oracle access to a quantum switch $W$ which coherently superposes two unitaries $U_0$ and $U_1$ in different causal orders.  Controlled by a qubit state $\ket{\phi}$, it acts on the target state $\ket{\psi}$ using one oracle query to each unitary:
\begin{align}
    W:\ket{\phi}\times\ket{\psi}\times U_1\times U_0
    \mapsto
    \left(\proj{0}\otimes U_1U_0+\proj{1}\otimes U_0U_1\right)
    (\ket{\phi}\otimes\ket{\psi}) .
\end{align}
Algorithm~\ref{SI-alg:metrology_qswitch} estimates $\theta$ in two parts. Preliminary shots with attenuated displacements identify the correct fringe of the final Ramsey experiment, and the remaining shots with unattenuated displacements are used for precision estimation.

For an attenuated quantum-switch experiment with controls $(\lambda_x,\lambda_p)$, define the gain
\begin{align}
    M:=\lambda_x\lambda_p N^2. 
    \label{SI-eq:metrological_gain}
\end{align}
The corresponding Ramsey phase is
\begin{align}
    \Phi_M=M\theta=-M\bar{\omega}.
    \label{SI-eq:ramsey_phase_gain}
\end{align}
At the final unattenuated setting, $M=N^2$, the likelihood is periodic in $\theta$ with period $2\pi/N^2$. The preliminary shots acquire the prior required to identify a single fringe, starting from the coarse prior interval $\Theta_0$.

Choose fixed constants $0<M_0\leq1$ and $r>1$ such that $M_0(\theta_+-\theta_-)<\pi$, and use the gain schedule for $l=0,\ldots, L=\left\lceil\log_r\left(\frac{N^2}{M_0}\right)\right\rceil$
\begin{align}
    M_l=\min\{M_0r^l,N^2\}.
    \label{SI-eq:gain_schedule}
\end{align}
That is, successive gains have ratios at most $r$ from $l=1, \ldots, L=\Theta(\log N)$. Each stage selects the fringe compatible with the preceding estimate, following the approach of~\cite{Ji_2008}.

We use $Q$ for the number of measured quadratures: $Q=1$ when only one quadrature is unknown in each displacement and $Q=2$ when both quadratures are unknown. 
Let $\nu$ denote the total number of shots per quadrature, and let $\nu_l$ denote the number of shots per quadrature at stage $l$. We allocate $\nu_{\mathrm u}= \sum_{l=0}^{L-1}\nu_l$ shots per quadrature for phase unwrapping and $\nu-\nu_{\mathrm u}$ per quadrature for high precision measurement. 
The preliminary stages are required to achieve only $\MSE=O(N^{-4})$, but their failure probability must decrease with $\nu$ so that incorrect fringe choices do not dominate the final MSE. We show below that $\nu_{\mathrm u}=O(\log (N)\log(\nu N))$ suffices, leaving essentially the full shots for the final stage whenever $\nu_{\mathrm u}=o(\nu)$ and $\nu - \nu_{\mathrm u}\approx \nu $.

\begin{algorithm}[H]
	\caption{\textsf{Tunable Quantum Metrology with Quantum Switch and Phase Unwrapping}}
	\label{SI-alg:metrology_qswitch}

    Require: A probe state $\ket{\Psi(0,0)}$, a control qubit $\ket{+}$, oracle access to a quantum switch $W$, tunable oracle access to $e^{\iu\lambda_p p_j\hat{X}}$ for $j=1,\ldots,N$ and $e^{-\iu\lambda_x x_k\hat{P}}$ for $k=1,\ldots,N$, a prior interval $\Theta_0=[\theta_-,\theta_+]$, gains $M_l$ for $l=1,\ldots , L$, attenuation pairs $(\lambda_{x,l},\lambda_{p,l})$ for $l=1,\ldots ,L$, repetitions $\nu_l$ for $l=1,\ldots,L$, and phase gates $U_{\phi}=\proj{0}+e^{\iu\phi}\proj{1}$ with $\phi_c=0$ and $\phi_s=-\frac{\pi}{2}$.

    \For{$l=0,\ldots,L$}{
        Set counters $t_{l,c}=0$ and $t_{l,s}=0$.
        \For{$q\in\{c,s\}$}{
            \For{$r=1,\ldots,\nu_l$}{
                Apply the quantum switch $W$ to the tuple $\ket{+}\times \ket{\Psi(0,0)}  \times e^{\iu\lambda_p \sum_j p_j\hat{X}} \times e^{-\iu\lambda_x \sum_k x_k\hat{P}}$

                Apply the phase gate $U_{\phi_q}$ to the control qubit.

                Measure the control qubit in the diagonal basis $\{\ket{+},\ket{-}\}$.

                If the outcome is $\ket{+}$, increment $t_{l,q}:=t_{l,q}+1$.
            }
        }
        
        Form the empirical quadratures $C_l:= \frac{2t_{l,c}-\nu_l}{\nu_l}$ and $S_l := \frac{2t_{l,s}-\nu_l}{\nu_l}$. Pick 
        and the wrapped phase-space estimate $\alpha_l := \arg(C_l+\iu S_l)\in(-\pi,\pi]$. 

        \If{$l=0$}{
            Choose the unique integer $z_0$ such that $\frac{\alpha_0+2\pi z_0}{M_0}\in \Theta_0$, and set $\hat{\theta}_0:=\frac{\alpha_0+2\pi z_0}{M_0}$. 
        }

        \If{$l\geq 1$}{
            Choose the integer $z_l:= \arg\min_{z\in\mathbb{Z}}\left|\frac{\alpha_l+2\pi z}{M_l}-\hat{\theta}_{l-1}\right|,$
            and set $\hat{\theta}_l:= \frac{\alpha_l+2\pi z_l}{M_l} $. 
        }
    }

    Return $\hat{\theta}:=\hat{\theta}_L$.
\end{algorithm}

The following theorem bounds the MSE of Algorithm~\ref{SI-alg:metrology_qswitch}, including the errors from phase estimation, phase unwrapping, and the DGW approximation. 

\begin{theorem}[Global Upper Bound for MSE with Tunable Quantum Switch]
\label{SI-thm:upper_bound_with_ico}
    Consider the setting of Section~\ref{SI-sec:tunable_setting}, with two sets of $N$ unknown displacements, finite odd dimension $d$, and DGW width $\sigma=\sqrt{d}$. Assume the bulk condition
    \begin{align}
        \Delta_{N\bar{x}}:=N|\bar{x}|\sqrt{\frac{d}{2\pi}}  \leq\frac{d-1}{4}, \qquad   \Delta_{N\bar{p}}:=N|\bar{p}|\sqrt{\frac{d}{2\pi}}   \leq\frac{d-1}{4}.
        \label{SI-eqn:condition}
    \end{align}
    Let $\Theta_0=[\theta_{\min},\theta_{\max}]$, $\Delta\theta=\theta_{\max}-\theta_{\min}$ be a known $N$-independent prior interval, and let $M_0<\cdots<M_L=N^2$ be a gain schedule satisfying
    \begin{align}
        M_0 \Delta \theta<\pi,   \qquad  r= \frac{M_{l+1}}{M_l}.
    \end{align}
    Choose the number of repetitions at the unwrapping stages so that for $l=0,\ldots, L -1  = \Omega( 2 \log_r N)$
    \begin{align}
        \nu_l  = \left\lceil \frac{8 r^2}{(1-2r\epsilon_{\mathrm{bd}})^2} \log\Big(\frac{4(L+1)}{\delta_{\mathrm{u}}}\Big)\right\rceil, 
    \end{align}
    for a chosen failure probability $\delta_{\mathrm{u}}\in(0,1)$, summing up to
    \begin{align}
        \nu_{\rm u} = \mathcal{O}(\log (N) \log (1/\delta_{\rm u})). 
    \end{align}
    This means that setting $\delta_{\rm u}= 1/\poly(\nu,N)$, we can use very few rounds for phase unwrapping and the majority for the final rounds such that $\nu_{\rm u}=o(\nu)$ and $\nu_L= \nu - \nu_{\rm u} \approx \nu$ repetitions per quadrature.

    Running Algorithm~\ref{SI-alg:metrology_qswitch} with the bounded-energy probe state $\ket{\Psi(0,0)}$, whose energy is $1+\mathcal{O}\!\left(\poly(d)e^{-\frac{\pi}{4}d}\right)$, 
    produces an estimator $\hat{\theta}$ satisfying
    \begin{align}
        \MSE(\hat{\theta},\theta)  \leq  \frac{1}{\nu N^4} + \frac{\epsilon_\mathrm{bd}^2}{N^4} + \Delta \theta^2 \delta_{\mathrm{u}},
    \end{align}
    where $C>0$ is a universal constant and
    \begin{align}
        \epsilon_{\mathrm{bd}} \leq  \mathcal{O}\big(  \poly(d)e^{-\frac{\pi d}{16}}  \big).
    \end{align}
    In particular, choosing $\delta_{\mathrm{u}}=\frac{1}{\Delta \theta^2 \nu N^4}$ gives
    \begin{align}
        \MSE(\hat{\theta},\theta) \leq  \frac{2}{\nu N^4}  + \mathcal{O}\Big(  \frac{\poly(d) e^{-\frac{\pi d}{16}}}{N^4}  \Big),
    \end{align}
    with only $O(\polylog (\nu N))$ additional low-gain repetitions for phase unwrapping.
\end{theorem}

\begin{proof}
    We first analyze a single attenuated quantum-switch stage with fixed attenuations $(\lambda_x,\lambda_p)$ and gain
    \begin{align}
        M=\lambda_x\lambda_p N^2 .
    \end{align}
    Successive translations combine linearly:
    \begin{align}
        \prod_{j=1}^{N} e^{\iu\lambda_p p_j\hat{X}}  = e^{\iu\lambda_p N\bar{p}\hat{X}}, \qquad  \prod_{k=1}^{N} e^{-\iu\lambda_x x_k\hat{P}} = e^{-\iu\lambda_x N\bar{x}\hat{P}} .
    \end{align}
    Let the initial probe state be $\ket{\Psi(0,0)}$ and the control be initialized in $\ket{+}$.  Applying the quantum switch, followed by a control phase gate $U_\phi$, gives
    \begin{align}
        \ket{\Phi_{M,\phi}}  = \frac{1}{\sqrt{2}}
        \big(  \ket{0}\otimes  e^{-\iu\lambda_xN\bar{x}\hat{P}} e^{\iu\lambda_pN\bar{p}\hat{X}} \ket{\Psi(0,0)} +  e^{\iu\phi}\ket{1}\otimes  e^{\iu\lambda_pN\bar{p}\hat{X}} e^{-\iu\lambda_xN\bar{x}\hat{P}} \ket{\Psi(0,0)}   \big).
    \end{align}
    By Theorem~\ref{SI-thm:general_weyl_relation}, and using the fact that
    \begin{align}
        \lambda_x N|\bar{x}|\sqrt{\frac{d}{2\pi}} \leq   N|\bar{x}|\sqrt{\frac{d}{2\pi}}, \qquad   \lambda_p N|\bar{p}|\sqrt{\frac{d}{2\pi}} \leq  N|\bar{p}|\sqrt{\frac{d}{2\pi}},
    \end{align}
    every attenuated stage remains in the bulk whenever the final unattenuated stage does.  Hence
    \begin{align}
        e^{\iu\lambda_pN\bar{p}\hat{X}} e^{-\iu\lambda_xN\bar{x}\hat{P}}  \ket{\Psi(0,0)}  =e^{\iu M\theta} e^{-\iu\lambda_x N\bar{x}\hat{P}} e^{\iu\lambda_pN\bar{p}\hat{X}}\ket{\Psi(0,0)} + \ket{\delta_M},
    \end{align}
    with the uniform boundary estimate
    \begin{align}
        \|\ket{\delta_M}\|_2  \leq \epsilon_\mathrm{bd}  \leq  \mathcal{O}\big(\poly(d)e^{-\frac{\pi d}{16}} \big).
    \end{align}
    Therefore
    \begin{align}
        \ket{\Phi_{M,\phi}}=\frac{1}{\sqrt{2}}\big(\ket{0}+ e^{\iu(M\theta+\phi)}\ket{1}\big)\otimes e^{-\iu\lambda_xN\bar{x}\hat{P}}e^{\iu\lambda_pN\bar{p}\hat{X}}
        \ket{\Psi(0,0)}+\frac{e^{\iu\phi}}{\sqrt{2}}\ket{1}\otimes\ket{\delta_M}.
        \label{SI-eq:single_stage_state}
    \end{align}

    Measuring the control qubit in the diagonal basis gives a Bernoulli outcome.  If we encode the outcome as $Y_\phi=+1$ for $\ket{+}$ and $Y_\phi=-1$ for $\ket{-}$, then
    \begin{align}
        \mathbb{E}[Y_\phi]  =  \cos(M\theta+\phi)+\Delta_{M,\phi}, \qquad |\Delta_{M,\phi}|\leq \epsilon_\mathrm{bd}.
        \label{SI-eq:single_stage_mean}
    \end{align}
    In particular, the two analysis phases $\phi_c=0$ and $\phi_s=-\pi/2$ estimate the two quadratures
    \begin{align}
        \mathbb{E}[Y_{\phi_c}]  =  \cos(M\theta)+\Delta_{M,c}, \qquad \mathbb{E}[Y_{\phi_s}] = \sin(M\theta)+\Delta_{M,s}.
    \end{align}
    We next show that the gain sequence unwraps the phase with high probability.  Let
    \begin{align}
        \widehat{C}_l =  \frac{2t_{l,c}-\nu_l}{\nu_l}, \qquad \widehat{S}_l = \frac{2t_{l,s}-\nu_l}{\nu_l}. 
    \end{align}
    Hoeffding's inequality for bounded $\pm1$ random variables gives, for every $\eta>0$,
    \begin{align}
        \Pr\big[ |\widehat{C}_l-\mathbb{E}\widehat{C}_l|>\eta \big]
        \leq  2e^{-\frac{1}{2}\nu_l\eta^2},
        \qquad
        \Pr\big[ |\widehat{S}_l-\mathbb{E}\widehat{S}_l|>\eta \big] \leq 2e^{-\frac{1}{2}\nu_l\eta^2}.
    \end{align}
    That's to say, for $l=0,\ldots, L -1  = \Omega( 2 \log_r N)$
    \begin{align}
        \nu_l = \left \lceil   \frac{8 r^2}{(1-2r\epsilon_{\mathrm{bd}})^2} \log\Big(\frac{4(L+1)}{\delta_{\mathrm{u}}}\Big) \right\rceil, 
    \end{align}
    ensures
    \begin{align}
        \Pr \big[ |\widehat{C}_l-\mathbb{E}\widehat{C}_l|> \frac{1}{2} \big(\frac{1}{r}- 2\epsilon_\mathrm{bd} \big)\big] \leq \frac{\delta_{\mathrm{u}}}{2(L+1)}, \qquad  \Pr \big[ |\widehat{S}_l-\mathbb{E}\widehat{S}_l|>  \frac{1}{2} \big(\frac{1}{r}- 2\epsilon_\mathrm{bd} \big) \big] \leq \frac{\delta_{\mathrm{u}}}{2(L+1)}. 
    \end{align}
    This indicates that  
    \begin{align}
        \nu_{\rm u}  =  \sum_{l=0}^{L-1} \nu_l = \mathcal{O}(\log(N)\log (1/\delta_{\rm u})). 
    \end{align}
    Since the ideal vector $(\mathbb{E}\widehat{C}_l,\mathbb{E}\widehat{S}_l)$ has almost unit norm (up to $2\epsilon_{\mathrm{bd}}$), a sufficiently small Euclidean quadrature error implies a comparable wrapped-phase error: for $l=0,\ldots,L$
    \begin{align}
        \Pr \big[ \|(\widehat{C}_l,\widehat{S}_l)-(\mathbb{E}\widehat{C}_l,\mathbb{E}\widehat{S}_l)\|_2 >  \frac{1}{r} \big] \leq  \frac{\delta_{\mathrm{u}}}{L+1},
    \end{align}
    and therefore
    \begin{align}
        \Pr \big[ |\widehat{\alpha}_l-M_l\theta| > \frac{1}{r}    \big] \leq \frac{\delta_{\mathrm{u}}}{L+1}, 
        \label{SI-eq:phase_concentration}
    \end{align}
    Provided the bulk condition for sufficiently large $d$, $\epsilon_\mathrm{bd}$ is exponentially negligible, and $\hat{\alpha}_l$ estimates $M_l\theta$ well within $\frac{1}{r}$. Let $\mathcal{E}$ be the event that Eq.~\eqref{SI-eq:phase_concentration} holds for all stages.  By a union bound,
    \begin{align}
        \Pr(\mathcal{E}^c)\leq \delta_{\mathrm{u}}.
    \end{align}
    Conditioned on $\mathcal{E}$, the unwrapping rule in Algorithm~\ref{SI-alg:metrology_qswitch} selects the correct branch at every stage.  Indeed, the initial gain satisfies $M_0(\theta_+-\theta_-)<\pi$, so $M_0\theta$ is injective on the prior interval up to a choice of branch cut.  Hence the first stage produces
    \begin{align}
        |\hat{\theta}_0-\theta| \leq  \frac{1}{M_0r}.
    \end{align}
    Suppose inductively that stage $l-1$ has selected the correct branch and
    \begin{align}
        |\hat{\theta}_{l-1}-\theta| \leq  \frac{1}{M_{l-1}r}.
    \end{align}
    Since $\frac{M_l}{M_{l-1}}\leq r$, this implies
    \begin{align}
        |\hat{\theta}_{l-1}-\theta| \leq
        \frac{1}{M_l} .
    \end{align}
    To estimate the correct branch of the $l$-th wrapped phase estimate, we need
    \begin{align}
        \frac{1}{M_{l-1}r}  \leq \frac{\pi}{M_l}, 
    \end{align}
    of the true $\theta$.  Every incorrect branch is separated from the correct branch by an integer multiple of $\frac{2\pi}{M_l}$, and hence is farther away from $\hat{\theta}_{l-1}$.  Thus the nearest-branch rule chooses the correct branch.  By induction, the final stage is unwrapped correctly.

    Conditioned on correct unwrapping, the final estimator is simply the local quadrature phase estimator at gain $M_L=N^2$.  The delta method applied to
    \begin{align}
        \widehat{\alpha}_L=\arg(\widehat{C}_L+\iu\widehat{S}_L)
    \end{align}
    gives
    \begin{align}
        \Var(\widehat{\alpha}_L|\mathcal{E}) \leq \frac{1}{\nu},
    \end{align}
    because the gradient of $\arg(C+\iu S)$ has unit norm on the ideal unit circle and the two empirical quadratures have variances $O(\nu^{-1})$.  Dividing by $M_L=N^2$ yields
    \begin{align}
        \Var(\hat{\theta}|\mathcal{E}) \leq \frac{1}{\nu N^4}.
    \end{align}
    The boundary error in Eq.~\eqref{SI-eq:single_stage_mean} shifts the final quadratures by at most $\epsilon_\mathrm{bd}$, and therefore contributes a squared bias bounded by
    \begin{align}
        \Bias(\hat{\theta},\theta|\mathcal{E})^2 \leq \frac{\epsilon_\mathrm{bd}^2}{N^4}.
    \end{align}
    Hence
    \begin{align}
        \mathbb{E}\big[ (\hat{\theta}-\theta)^2|\mathcal{E} \big]
        \leq \frac{1}{\nu N^4}  +   \frac{\epsilon_\mathrm{bd}^2}{N^4}.
    \end{align}

    Finally, Algorithm~\ref{SI-alg:metrology_qswitch} clips its estimates to the prior interval $\Theta$, so even on the failure event $\mathcal{E}^c$ the squared error is at most $(\theta_+-\theta_-)^2$.  Therefore
    \begin{align}
        \MSE(\hat{\theta},\theta)  &= \Pr(\mathcal{E}) \mathbb{E}\big[ (\hat{\theta}-\theta)^2|\mathcal{E} \big] +   \Pr(\mathcal{E}^c) \mathbb{E}\big[  (\hat{\theta}-\theta)^2|\mathcal{E}^c \big]   \leq  \frac{1}{\nu N^4} + \frac{\epsilon_\mathrm{bd}^2}{N^4}  +   \Delta \theta^2\delta_{\mathrm{u}}.
    \end{align}
    To ensure that the MSE caused by phase unwrapping is upper bounded by the MSE caused by phase estimation, 
    \begin{align}
        \delta_{\mathrm{u}} = \Theta(\frac{1}{\nu N^4}). 
    \end{align}
    This leads to the number of repetitions in each round of phase unwrapping $l=0,\ldots, L -1  = \Omega( 2 \log_r N)$
    \begin{align}
        \nu_l = \Theta(\log(\nu N^4 L)).
    \end{align}
    The number of rounds of phase unwrapping is 
    \begin{align}
        L = \Theta(\log N). 
    \end{align}
    The additional low-gain repetitions for phase unwrapping is 
    \begin{align}
        \nu_{\rm u} = \sum_{l=0}^{L-1} \nu_l = \Theta(\log(N) \log(\nu N)). 
    \end{align}
    This proves the theorem. 
\end{proof}

\begin{remark}[Operational removal of the $O(N^{-2})$ prior]
    A single final-gain experiment with $M=N^2$ has likelihood periodicity $\frac{2\pi}{N^2}$ in $\theta$, and therefore requires prior knowledge of width $O(N^{-2})$ to be interpreted globally.  Algorithm~\ref{SI-alg:metrology_qswitch} does not assume this prior.  Instead, it manufactures it by measuring the same parameter $\theta$ at gains $M_0,M_1,\ldots,M_L$, where each stage narrows the posterior or confidence interval to one fringe of the next stage.  The cost of this prior acquisition is logarithmic in $N$ up to logarithmic failure-probability factors, rather than $O(N^2)$ repetitions of a definite-order coarse estimator.
\end{remark}

\begin{remark}[The regime of negligible bias]
    In order that the bias is negligible in Theorem~\ref{SI-thm:upper_bound_with_ico} compared to the MSE $\frac{1}{\nu N^4}$, we require both the bulk confinement condition to ensure the validity of the DGW framework and the pre-asymptotic condition to ensure $\frac{\poly(d)e^{-\frac{\pi d}{16}}}{N^4}$ is dominated by $\frac{1}{\nu N^4}$.
    
    To strictly satisfy the bulk confinement condition in Eq.~\eqref{SI-eqn:condition}, it is sufficient that all high-gain unattenuated stages satisfy the bulk condition and all lower-gain attenuated stages with smaller displacements automatically satisfy the same bulk condition: 
    \begin{align}
        N \max\lbrace |\bar{x}| , |\bar{p}| \rbrace \leq \sqrt{\frac{\pi d}{8}} .
    \end{align}
    Since $\bar{x}$ and $\bar{p}$ are bounded within constant ranges, the final unattenuated stage remains in the bulk whenever
    \begin{align}
        d \geq \frac{8}{\pi} \max \lbrace |x_{\max}|^2,|p_{\max}|^2 \rbrace N^2. 
    \end{align}
    This implies that when $d = \Omega(N^2)$, the DGW framework is secured. 
    
    To satisfy the pre-asymptotic condition,  we require that 
    \begin{align}
        \frac{\poly(d)e^{-\frac{\pi d}{16}}}{N^4}\leq\frac{1}{\nu N^4}
    \end{align}
    This implies that when $\nu = \mathcal{O}(e^{\frac{\pi d}{16}}/\poly(d^{}) )$, the MSE dominates the bias. Though this is a non-asymptotic regime, the size of the regime grows exponentially large with respect to $d$, meaning that we can achieve the desired scaling of $\frac{1}{\nu N^4}$ in a reasonably large regime. 
\end{remark}

\begin{remark}[Attainability of the QCRB]
    The QCRB is a lower bound and its attainability requires a separate local asymptotic argument: after coarse localization to an identifiable branch, a regular one-parameter model typically admits an efficient estimator with $\mathbb{E}_\theta(\hat\theta-\theta)^2=(1+o(1))/(\nu I_Q(\theta))$ as $\nu\to\infty$. For the order bound used here, however, it suffices to have constant-factor attainability, $\mathbb{E}_\theta(\hat\theta-\theta)^2=O(1/(\nu I_Q(\theta)))$. For example, in the Bernoulli phase model $p_\phi(0)=(1+\cos(M\phi))/2$, $p_\phi(1)=(1-\cos(M\phi))/2$, restricted to a single branch, $I(\phi)=M^2$, and constant-factor CRB attainability only requires $\nu \inf_{\phi\in\Theta}\sin^2(M\phi)=\Omega(1)$, while sharp $1+o(1)$ attainability requires this quantity to diverge. Hence the exponential repetition scale $\nu=\mathcal{O}(e^{\pi d/ 16}/\poly(d))$ is imposed to ensure that the boundary error remains negligible compared to the statistical error, while the intrinsic QCRB-attainability requirement may be substantially weaker. 
\end{remark}

\subsection{Quantum metrology without quantum switch}
\label{SI-sec:metrology_without_ico}

We next derive a lower bound for DCO strategies using the same fixed causal order in every shot. The argument allows arbitrary known interleaving controls, including controls that inject energy after the first unknown oracle call.

\begin{theorem}[Lower Bound for MSE without Quantum Switch]
    \label{SI-thm:lower_bound_without_ico}
    Define $\hat{X}$, $\hat{P}$, and $\hat{H}$ as in Definition~\ref{SI-def:DGW_framework}, with $d$ exponential in the number of qubits $n$.  Consider any quantum metrology protocol with tunable access to the displacements $e^{\iu \lambda_p p_j \hat{X}}$ and $e^{-\iu \lambda_x x_k \hat{P}}$ in the same fixed causal order in every shot where $x_j\in[x_{\min}, x_{\max}]$, $\Delta x=x_{\min}-x_{\max}$ and $p_j\in[p_{\min},p_{\max}]$, $\Delta p=p_{\min}-p_{\max}$ such that $x_{\min},x_{\max}\in[0,\gamma_x]$ and $p_{\min},p_{\max}\in[0,\gamma_p]$, whose probe state $\ket{\Psi}$ has bounded energy $E$, and which has no access to any process with indefinite causal order. For estimating $\theta=-\bar{\omega} = \bar{x}\bar{p}$, the MSE of any unbiased estimator $\hat{\theta}$ is lower bounded by
    \begin{align}
        \MSE(\hat{\theta},\theta) \geq  \frac{\min\{\bar{x}^2,\bar{p}^2\}}{4\nu E N^2}.
    \end{align}
    Moreover, the maximal MSE of any biased estimator $\hat{\theta}$ over the prior interval $\Theta_0=\theta_+-\theta_-$ is lower bounded by
    \begin{align}
        \MSE_{\Theta}(\hat{\theta},\theta) \geq  \min \Big\{ \frac{\min\{\bar{x}^2,\bar{p}^2\}}{16\nu E N^2}, \frac{\min\lbrace \bar{x}^2\Delta p^2, \bar{p}^2 \Delta x^2 \rbrace}{16N^2}, \frac{\Delta\theta^2}{16} \Big \}.
    \end{align}
\end{theorem}

\begin{proof}
    We lower bound the estimation error for $\theta=\bar{x}\bar{p}$ by giving the estimator additional classical side information. Suppose that all displacements are known except the actual first displacement. Without loss of generality, we assumes it to be the position displacement amplitude $x_1$.  Providing side information can only decrease the MSE, so any lower bound in this easier problem also lower bounds the original problem.

    With all other parameters known, the target parameter depends on $x_1$ through
    \begin{align}
        \frac{\partial\theta}{\partial x_1} = \frac{\bar{p}}{N}.
    \end{align}
    Thus any unbiased estimator of $\theta$ induces an unbiased estimator of $x_1$, with
    \begin{align}
        \MSE(\hat{\theta},\theta)  = \frac{\bar{p}^2}{N^2} \MSE(\hat{x}_1,x_1).
    \end{align}
    In the tunable setting, the unknown parameter $x_1$ is encoded through $e^{-\iu\lambda_x x_1\hat{P}}$ for $0\leq\lambda_x\leq 1$. The corresponding generator is $\lambda_x\hat{P}$.  For any probe state satisfying
    \begin{align}
        \bra{\Psi}\hat{H}\ket{\Psi}\leq E,
    \end{align}
    the quantum Fisher information obeys
    \begin{align}
        F_Q(x_1)  =   4\lambda_x^2  \big(  \bra{\Psi}\hat{P}^2\ket{\Psi}  -  \bra{\Psi}\hat{P}\ket{\Psi}^2  \big)  \leq   4\lambda_x^2\bra{\Psi}\hat{P}^2\ket{\Psi}  \leq   4\bra{\Psi}\hat{H}\ket{\Psi}   \leq   4E .
    \end{align}
    The quantum Cram\'er-Rao bound therefore gives
    \begin{align}
        \MSE(\hat{x}_1,x_1)  \geq  \frac{1}{4\nu E}.
    \end{align}
    Converting back to $\theta$, we obtain
    \begin{align}
        \MSE(\hat{\theta},\theta)  \geq  \frac{\bar{p}^2}{4\nu E N^2}.
    \end{align}

    By symmetry, we may instead reveal all parameters except the first one with momentum displacement amplitude $p_1$.  Then
    \begin{align}
        \frac{\partial\theta}{\partial p_1}  =  \frac{\bar{x}}{N},
    \end{align}
    and the unknown parameter is encoded through
    \begin{align}
        V_1(\lambda_p)=e^{\iu\lambda_p p_1\hat{X}},  \qquad 0\leq\lambda_p\leq 1.
    \end{align}
    Since the generator is $\lambda_p\hat{X}$ and $\lambda_p\leq 1$, the same energy argument gives
    \begin{align}
        F_Q(p_1)\leq 4E,   \qquad   \MSE(\hat{p}_1,p_1)\geq \frac{1}{4\nu E}.
    \end{align}
    Therefore
    \begin{align}
        \MSE(\hat{\theta},\theta)  =  \frac{\bar{x}^2}{N^2} \MSE(\hat{p}_1,p_1) \geq  \frac{\bar{x}^2}{4\nu E N^2}.
    \end{align}
    Taking the lower one of the two side-information reductions (as the first one is either a position displacement or a momentum displacement) yields
    \begin{align}
        \MSE(\hat{\theta},\theta)   \geq  \frac{\min\{\bar{x}^2,\bar{p}^2\}}{4\nu E N^2}.
    \end{align}
 
    Any biased estimator $\hat{\theta}$ on $\Theta=[\theta_0-\Delta\theta,\theta_0+\Delta\theta]$, natively induces a biased estimator $\hat{x}_1$ on $\Theta_x=[x_0-\epsilon_x,x_0+\epsilon_x]$ via $\hat{x}_1-x_0 = \frac{N}{\bar{p}}(\hat{\theta} - \theta_0)$ where $\epsilon_x = \min\lbrace \frac{N\Delta \theta}{2\bar{p}}, \frac{\Delta x}{2} \rbrace $, meaning their MSEs on the prior interval scale exactly by the squared scaling coefficient:
    \begin{align}
        \MSE_\Theta(\theta_0)  = \frac{\bar{p}^2}{N^2} \MSE_{\Theta_x}(x_0). 
    \end{align}
    For any possibly biased estimator $\hat{x}_1$, the MSE is bounded by the quantum Cram\'er-Rao bound
    \begin{align}
        \MSE(\hat{x}_1,x_1) \geq \frac{(1+\partial_{x_1}\Bias(x_1))^2}{\nu F_Q(x_1)} + \Bias(x_1)^2. 
    \end{align}
    By the mean value theorem, for some $\xi\in\Theta_x$
    \begin{align}
        \Bias(x_0+\epsilon_x)- \Bias(x_0-\epsilon_x) =  2\epsilon_x \partial_{x_1}\Bias(x_1)\big|_{x_1=\xi}.
    \end{align}
    We analyze the two possible regimes for the bias derivative at $\xi$:
    \begin{itemize}
        \item \textbf{Case 1:} If $\partial_{x_1}\Bias(x_1)\big|_{x_1=\xi}\geq -\frac{1}{2}$, then the numerator of the quantum Cram\'er-Rao bound is bounded away from zero:
        \begin{align}
            \MSE(\hat{x}_1,\xi) \geq \frac{1}{4\nu F_Q(\xi)} \geq \frac{1}{4 \nu F_Q^{\Theta_x}(x_0)},
        \end{align}
        where 
        \begin{align}
            F_Q^{\Theta_x}(x_0) = \max_{x_1\in\Theta_x} F_Q(x_1). 
        \end{align}
        \item \textbf{Case 2} If $\partial_{x_1}\Bias(x_1)\big|_{x_1=\xi}\leq -\frac{1}{2}$, the bias changes rapidly across the interval, 
        \begin{align}
            |\Bias(x_0+\epsilon_x)- \Bias(x_0-\epsilon_x)| \geq \epsilon_x,  
        \end{align}
        meaning the squared bias itself must be large at the boundaries
        \begin{align}
            \max\lbrace|\Bias(x_0+\epsilon_x)|, |\Bias(x_0-\epsilon_x)| \rbrace \geq \frac{\epsilon_x}{2}.
        \end{align}
        Therefore, 
        \begin{align}
            \max\lbrace \MSE(\hat{x}_1,x_0+\epsilon_x), \MSE(\hat{x}_1,x_0-\epsilon_x )\rbrace \geq \frac{\epsilon_x^2}{4}. 
        \end{align}
    \end{itemize}
    Taking the minimum of these two exhaustive cases guarantees: 
    \begin{align}
        \MSE_{\Theta_x}(x_0) \geq \min\Big\lbrace \frac{1}{4 \nu F_Q^{\Theta_x}(x_0)} , \frac{\epsilon_x^2}{4} \Big\rbrace. 
    \end{align}
    To bound the quantum Fisher information, recall that the unknown parameter $x_1$ is encoded via the position translation $e^{-\iu x_1 \hat{P}}$. Because $\hat{P}$ is the Hermitian generator of this translation, the quantum Fisher information evaluated on any probe state $\ket{\Psi}$ constrained by $\bra{\Psi}  \hat{H} \ket{\Psi}   \leq E$ is rigorously bounded by its variance:

    \begin{align}
        F_Q^{\Theta_x}(x_0) = \max_{x_1\in\Theta_x} F_Q(x_1) \leq 4E. 
    \end{align}
    Substituting this into the inequality for the MSE gives:
    \begin{align}
        \MSE_{\Theta_x}(x_0) \geq \min\Big\lbrace \frac{1}{16\nu E} , \frac{\epsilon_x^2}{4} \Big\rbrace.
    \end{align}
    Converting this back to the target parameter $\theta$ and substituting $\epsilon_x = \min \lbrace \frac{N(\theta_+-\theta_-)}{2\bar{p}},\frac{x_{\max}-x_{\min}}{2}\rbrace$, we obtain
    \begin{align}
        \MSE_{\Theta}(\theta_0) = \frac{\bar{p}^2}{N^2} \MSE_{\Theta_x}(x_0) \geq \min \Big\lbrace\frac{\bar{p}^2 }{16\nu E N^2}, \frac{\bar{p}^2\Delta x^2}{16N^2} , \frac{\Delta\theta^2}{16}\Big\rbrace. 
    \end{align}    
    By symmetry, we can instead grant the estimator classical knowledge of all parameters except $p_1$. The parameter $p_1$ is encoded via $e^{\iu p_1 \hat{X}}$, with conjugate generator $\hat{X}$. An identical proof structure yields the alternative bound:
    \begin{align}
        \MSE_{\Theta}(\theta_0)\geq  \min \Big\lbrace\frac{\bar{x}^2 }{16\nu E N^2}, \frac{\bar{x}^2\Delta p^2}{16N^2} ,\frac{\Delta\theta^2}{16}\Big\rbrace. 
    \end{align}
    Taking the lower of these two inequalities yields the final result
    \begin{align}
        \MSE_\Theta(\theta_0)\geq  \min \Big\lbrace\frac{ \min\lbrace\bar{x}^2, \bar{p}^2\rbrace}{16\nu E N^2}, \frac{\min\lbrace \bar{x}^2\Delta p^2, \bar{p}^2\Delta x^2\rbrace }{16N^2} , \frac{\Delta\theta^2}{16}\Big\rbrace
    \end{align}
\end{proof}

\begin{remark}
    Allowing attenuation does not invalidate the DCO lower bound. In the side-information reduction, the remaining unknown parameter is encoded through $\lambda_x\hat{P}$ or $\lambda_p\hat{X}$, with $0\leq\lambda_x,\lambda_p\leq1$, so the same energy bound on the quantum Fisher information continues to apply.
    The best definite causal order strategy therefore uses the unattenuated queries for coarse estimation, and still cannot beat the Heisenberg MSE scaling $O(1/N^2)$ for the independent-average product.
\end{remark}

\subsection{Separation between ICO and DCO}

Combining Theorems~\ref{SI-thm:upper_bound_with_ico} and~\ref{SI-thm:lower_bound_without_ico}, we obtain the following formal version of Theorem~\ref{thm:finite_dimensional_separation} in the main text.

\begin{theorem}[Unbounded separation across finite-dimensional instances]
\label{SI-thm:finite_dimensional_separation}
Fix parameter intervals $[x_{\min},x_{\max}]$, $\Delta x = x_{\max}- x_{\min}$ and $[p_{\min},p_{\max}]$, $\Delta p = p_{\max} - p_{\min}$ such that $x_{\min}, x_{\max},p_{\min}, p_{\max} = \mathcal{O}(1)$. Let $N$ be a large enough integer. Assume tunable oracle access to a set of $N$ position translations and a set of $N$ momentum translations whose amplitudes satisfy $x_j\in[x_{\min},x_{\max}]$ and $p_j\in[p_{\min},p_{\max}]$. The goal is to estimate $\bar{\theta} = \bar{x} \bar{p}$ where $\bar{x} = \frac{1}{N}\sum_j x_j$ and $\bar{p} = \frac{1}{N}\sum_j p_j$ given prior interval $\Theta=[\theta_{\min},\theta_{\max}]$, $\Delta\theta=\theta_{\max}-\theta_{\min}$. For all finite odd dimensions satisfying $d \geq  \frac{8}{\pi} N^2\max\lbrace x_{\max}^2,p_{\max}^2\rbrace$, when the total number of shots satisfies $\nu=\mathcal{O}(e^{\pi d/16}/\poly(d))$, there is a separation on the required initial probe energy to achieve the same accuracy between quantum metrology protocols with ICO and with DCO, as is suggested by the following two statements. 

First, there is a quantum metrology protocol with quantum switches and initial probe energy at most $1+\mathcal{O}(\poly(d) e^{-\pi d/4})$, given tunable access to two sets of phase-space translation oracles, each containing $N$ elements, whose estimator satisfies
\begin{align}
    \MSE(\hat{\theta}_{\mathrm{ICO}},\theta) \leq \frac{2}{\nu N^4} + \frac{\mathcal{O}(\poly(d)e^{-\frac{\pi d}{16}})}{N^4}.
    \label{SI-eqn:ico_bound}
\end{align}
Within the finite-sample window $\nu = \mathcal{O}(e^{\pi d/16}/ \poly(d))$, one can ensure $\nu \cdot \poly(d)e^{-\pi d/16}=o(1)$,
and hence 
\begin{align}
    \MSE(\hat{\theta}_{\mathrm{ICO}},\theta) \leq \frac{2 + o(1)}{\nu N^4}. 
\end{align}

Second, for all quantum metrology protocols with definite causal order and initial probe energy at most $E$, given tunable access to the same two sets of phase-space translation oracles, each containing $N$ elements, their unbiased estimator satisfy
\begin{align}
    \MSE(\hat{\theta}_{\mathrm{DCO}},\theta) \geq  \frac{\min\lbrace \bar{x}^2,\bar{p}^2 \rbrace}{4\nu E N^2}.
    \label{SI-eqn:dco_unbiased_bound}
\end{align}
and their biased estimator satisfy
\begin{align}
    \sup_{\theta\in\Theta}  \MSE(\hat{\theta}_{\rm DCO},\theta) \geq \min\Big\lbrace \frac{\min\lbrace \bar{x}^2,\bar{p}^2 \rbrace}{16\nu E N^2}, \frac{\min\lbrace \bar{x}^2\Delta p^2, \bar{p}^2 \Delta x^2 \rbrace }{16N^2},  \frac{\Delta\theta^2}{16}
    \Big\rbrace.
    \label{SI-eqn:dco_biased_bound}
\end{align}

Consequently, when $\nu=\mathcal{O}(e^{\pi d/16}/\poly(d))$, regarding the coefficients of $\MSE$ with respect to $\nu$, these two strategies are separated by
\begin{align}
    \frac{\MSE(\hat{\theta}_{\mathrm{DCO}},\theta)}{\MSE(\hat{\theta}_{\mathrm{ICO}},\theta)} \geq  (1-o(1))\frac{\min\lbrace \bar{x}^2,\bar{p}^2 \rbrace N^2}{32 E}.
\end{align}
In particular, to achieve the same mean square error with the same total number of shots for the finite-sample window $\nu=\mathcal{O}(e^{\pi d/16}/\poly(d))$, quantum metrology protocols with DCO requires initial probe state $E \geq \frac{\min\lbrace \bar{x}^2,\bar{p}^2\rbrace N^2}{32}$ while quantum metrology protocols with ICO requires initial probe state $E\leq  1+\mathcal{O}(\poly(d) e^{-\pi d/4})$. For any prescribed factor $R$, choosing a finite $N$ sufficiently large and a finite odd $d=\Omega(N^2)$ therefore makes this ratio at least $R$. 
\end{theorem}

\section{Compatibility with Previous Finite-dimensional Bunds}
\label{SI-sec:comparison}

We now explain the compatibility of our energy separation with the finite-dimensional bound in~\cite[Appendix~\rom{5}]{Kong_2026} by presenting the energy-constrained interpretation of their finite-dimensional argument. The following proposition gives a formal version of Proposition~\ref{prop:lower_bound_with_ico} in the main text. Its resource is the largest intermediate probe energy encountered
during the protocol.

\begin{proposition}[Energy-resolved fixed-dimensional bound]\label{SI-prop:lower_bound_with_ico}
    Let $\{x_j\}$ and $\{p_j\}$ ($j=1,\dots,N$) be unknown parameters encoded into a $d$-dimensional finite quantum system via target and auxiliary operations generated by $e^{\iu p_j \hat{X}}$ and $e^{-\iu x_j \hat{P}}$, respectively. A control system of dimension $d_C \le |S_N|$ with basis states $\{\ket{\pi}_C\}$ governs the permutation $\pi$ from the permutation group $S_{2N}$ of all possible unitaries, yielding a block-diagonal total unitary $\hat{W} = \sum_{\pi\in S_{2N}} \proj{\pi}_C \otimes \hat{U}_\pi(x,p)$ where $\hat{U}_\pi(x,p)$ denotes the unitary under causal order $\pi$. Let the maximal energy of the intermediate state among all causal order and all steps be $\max_{\pi,i}E_{\pi,i} = \Tr(\hat{V}_{\pi,i}\rho_\pi \hat{V}_{\pi,i}^\dagger\hat{H} )$ for $\hat{H}=\hat{X}^2+\hat{P}^2$ where $\hat{V}_{\pi,i}$ denotes the cumulative unitary evolution preceding the $j$-th application of $\hat{P}$. For the composite target parameter $A = \bar{x} \bar{p}$ (where $\bar{x} = \frac{1}{N}\sum_j x_j$ and $\bar{p} = \frac{1}{N}\sum_j p_j$), the mean squared error under this indefinite causal order protocol is bounded from below by
    \begin{align}
        \MSE(A)=\Omega\left(\frac{1}{\nu N^2 \max_{\pi,i} E_{\pi,i}}\right).
    \end{align}
    Consequently, at every fixed finite $d$, give that $\max_{\pi,i}E_{\pi,i}\leq d$, the bound recovers the Heisenberg scaling $\Omega(\nu^{-1} N^{-2} d^{-1})$ in the asymptotic limit $N\to\infty$ in~\cite[Appendix~\rom{5}]{Kong_2026}. Our $N^{-4}$ behavior is compatible with this bound only inside the finite-sample bulk window, which for bounded initial state energy $E$ and bounded displacement steps $N=O(\sqrt{d})$.
\end{proposition}

\begin{proof}
    Estimating $A = \bar{x}\bar{p}$ is a multiparameter quantum estimation problem governed by the Quantum Cramer-Rao Bound (QCRB). The variance $\Delta^2 A$ is bounded by the Quantum Fisher Information Matrix (QFIM), $\mathcal{F}$, via error propagation: $\Delta^2 A \ge \nabla A^T \mathcal{F}^{-1} \nabla A$, where the gradient vector is $\nabla A = (\partial_{\bar{x}} A, \partial_{\bar{p}} A)^T = (\bar{p}, \bar{x})^T$.

    To determine the ultimate precision limit, we evaluate the local generators for $\bar{x}$ and $\bar{p}$. We formalize the variation with respect to the mean parameters by shifting the individual parameters as $x_j = \bar{x} + \delta x_j$, giving the derivative operator $\partial_{\bar{x}} = \sum_{j=1}^N \partial_{x_j}$. The total generator for $\bar{x}$ is the block-diagonal operator:
    \begin{align}
        \hat{\mathcal{H}}_{\bar{x}} = \iu \hat{W}^\dagger \partial_{\bar{x}} \hat{W} = \sum_{\pi\in S_{2N}}\proj{\pi}_C \otimes \big( \iu \hat{U}_\pi^\dagger \partial_{\bar{x}} \hat{U}_\pi \big) = \sum_{\pi\in S_N} \proj{\pi}_C \otimes \hat{h}_\pi^{(\bar{x})}.
    \end{align}
    For a fixed causal order $\pi$, $\hat{U}_\pi(x,p)$ comprises $N$ instances of $\hat{X}$ and $N$ instances of $\hat{P}$. Here, we define the Heisenberg 
    \begin{align}
        \tilde{X}_j = \hat{V}_{\pi,j}^\dagger X \hat{V}_{\pi,j},\quad \tilde{P}_j = \hat{V}_{\pi,j}^\dagger P \hat{V}_{\pi,j}, \quad \tilde{H}_i = \hat{V}_{\pi,j}^\dagger H \hat{V}_{\pi,j}. 
    \end{align}
    Applying $\partial_{\bar{x}}$ yields exactly $N$ terms via the product rule:
    \begin{align}
        \hat{h}_\pi^{(\bar{x})} = \sum_{j=1}^N \tilde{P}_j = \sum_{j=1}^N \hat{V}_{\pi,j}^\dagger \hat{P} \hat{V}_{\pi,j}. 
    \end{align}
    For any initial state $\ket{\Psi}$
    \begin{align}
        \Var_\Psi(\hat{\mathcal{H}}_{\bar{x}} ) & = \bra{\Psi} \hat{\mathcal{H}}_{\bar{x}} ^2 \ket{\Psi} - |\bra{\Psi} \hat{\mathcal{H}}_{\bar{x}}  \ket{\Psi}|^2 \leq \bra{\Psi} \hat{\mathcal{H}}_{\bar{x}} ^2 \ket{\Psi} 
    \end{align}
    Considering that $\hat{\mathcal{H}}_{\bar{x}}$ is block diagonalized, let 
    \begin{align}
        \rho_\pi = \frac{\Tr_C ((\proj{\pi}\otimes \id) \proj{\Psi} )}{\Tr((\proj{\pi}\otimes \id) \proj{\Psi} )}
    \end{align}
    Then
    \begin{align}
        \bra{\Psi} \sum_{j=1}^N \proj{\pi}\otimes (\hat{h}_\pi^{(\bar{x})})^2\ket{\Psi} & \leq \max_{\pi} \Tr(\rho_\pi (\hat{h}_\pi^{(\bar{x})})^2) = \max_\pi \sum_{ij} \Tr(\rho_\pi \tilde{P}_i\tilde{P}_j ) \\
        & \leq  \max_\pi \sum_{ij} |\Tr(\rho_\pi \tilde{P}_i\tilde{P}_j )|    \leq \max_\pi \sum_{ij} \sqrt{\Tr(\rho_\pi \tilde{P}_i^2)} \sqrt{\Tr(\rho_\pi \tilde{P}_j^2)} \\
        & \leq N^2 \max_{\pi,i} \Tr(\rho_\pi \tilde{P}_i^2) \leq N^2 \max_{\pi,i} \Tr(\rho_\pi \tilde{H}_i)  \leq N^2 \max_{\pi,i} E_{\pi,i} 
    \end{align}
    Here, $E_{\pi,i}=\Tr(\rho_\pi \tilde{H}_i)$ is the energy of the state among all orders and steps. 
    We can also apply the same reasoning to $\hat{\mathcal{H}}_{\bar{p}}$. Therefore, 
    \begin{align}
        \Var_\Psi(\hat{\mathcal{H}}_{\bar{x}} )  \leq  N^2 \max_{\pi,i} E_{\pi,i}, \qquad \Var_\Psi(\hat{\mathcal{H}}_{\bar{p}} )  \leq  N^2 \max_{\pi,i} E_{\pi,i}. 
    \end{align}
    These limits constrain the single-parameter QFI (the diagonal elements of $\mathcal{F}$), which are bounded by the variance of their total generators over the initial state $|\Psi\rangle$: 
    \begin{align}
        \mathcal{F}_{\bar{x}\bar{x}} \le 4\Var_{\Psi}(\hat{\mathcal{H}}_{\bar{x}}) \leq 4N^2 \max_{\pi,i} E_{\pi,i}, \qquad
        \mathcal{F}_{\bar{p}\bar{p}} \le 4\Var_{\Psi}(\hat{\mathcal{H}}_{\bar{p}}) \leq 4N^2 \max_{\pi,i} E_{\pi,i}.
    \end{align}
    To bound the multiparameter error from below, we rely on the spectral properties of the positive semi-definite matrix $\mathcal{F}$. Its maximum eigenvalue is bounded by its trace:
    \begin{align}
        \lambda_{\max}(\mathcal{F}) \leq \mathrm{Tr}(\mathcal{F}) = \mathcal{F}_{\bar{x}\bar{x}} + \mathcal{F}_{\bar{p}\bar{p}} \leq 8N^2  \max_{\pi,i} E_{\pi,i}. 
    \end{align}
    Applying the matrix inequality 
    \begin{align}
        \mathcal{F}^{-1} \geq \lambda_{\max}(\mathcal{F})^{-1} \id,
    \end{align}
    yields the definitive lower bound on the variance:
    \begin{align}
        \Delta^2 A \ge \frac{|\nabla A|^2}{\nu \lambda_{\max}(\mathcal{F})} \geq \frac{\bar{p}^2 + \bar{x}^2}{8\nu N^2\max_{\pi,i} E_{\pi,i}}. 
    \end{align}
    To recover the original bound in~\cite[Appendix~\rom{5}]{Kong_2026}, we simply use the fact that 
    \begin{align}
        E_{\pi,i} = \Tr(\rho_\pi \tilde{H}_\pi) \leq \|\tilde{H}\|_\infty = \mathcal{O}(d),  
    \end{align}
    and thus 
    \begin{align}
        \Delta^2 A \geq \Omega\big(\frac{1}{8\nu N^2 d}\big). 
    \end{align}
\end{proof}

\begin{remark}
    The two lower bounds use different resource constraints. The DCO bound depends on the probe energy entering the first unknown oracle, whereas the bound for coherent superpositions of causal orders depends on the largest intermediate energy. This distinction explains how the finite-sample energy separation is compatible with the fixed-dimensional QFI bound.
\end{remark}

\section{Physics of DGW Framework}
\label{SI-sec:DGW_realization}

Now we discuss the physical foundations of the DGW framework. The scalable architectures of Discrete Continuous-Variable (DCV)~\cite{Kitaev_2008} and Quasi-Continuous Variable (QCV)~\cite{Opatrny_2017} quantum computation were proposed to bridge the gap between continuous-variable (CV) phase-space formalisms and realistic quantum hardware. A closely related finite-dimensional oscillator is provided by the discrete Kravchuk framework~\cite{Hakioglu_2000}. These descriptions arise in collective spins of large atomic ensembles~\cite{Opatrny_2017,Opatrny_2023} and in coupled spatial modes of engineered optical waveguide arrays~\cite{Weimann_2016}. We first describe the local spin--oscillator correspondence, and then discuss the Hamiltonians, free operations, and state preparation.

\subsection{Local spin--oscillator correspondence}

A $d$-dimensional Hilbert space can be identified with a spin-$j$ Hilbert space, where $d=2j+1$. The angular momentum operators satisfy
\begin{align}
    [\hat{J}_a,\hat{J}_b]=\iu\sum_c\epsilon_{abc}\hat{J}_c.
\end{align}
We center the local phase space around the $+x$ pole of the Bloch sphere and consider the macroscopic polarization regime $\langle\hat{J}_x\rangle\approx j$. The effective quadratures are
\begin{align}
    \hat{X}_{\mathrm{eff}}:=\frac{\hat{J}_z}{\sqrt{j}}, \qquad \hat{P}_{\mathrm{eff}}:=-\frac{\hat{J}_y}{\sqrt{j}}.
\end{align}
The commutation relation on these states satisfies
\begin{align}
    [\hat{X}_{\mathrm{eff}},\hat{P}_{\mathrm{eff}}] \ket{\Psi}   =\iu\frac{\hat{J}_x}{j}  \ket{\Psi}  \approx\iu \ket{\Psi} . 
\end{align}
Given the commutation relation, the two transverse-spin components act approximately as oscillator position and momentum on states localized near this pole.

In the $\hat{J}_z$ basis, where $\hat{J}_z\ket{m}_z=m\ket{m}_z$, we can write
\begin{align}
    \hat{J}_x=\frac{\hat{J}_++\hat{J}_-}{2}, \qquad  \hat{J}_y=\frac{\hat{J}_+-\hat{J}_-}{2\iu},
\end{align}
with
\begin{align}
    \hat{J}_+\ket{m}_z
    &=\sqrt{(j-m)(j+m+1)}\ket{m+1}_z,\\
    \hat{J}_-\ket{m}_z
    &=\sqrt{(j+m)(j-m+1)}\ket{m-1}_z.
\end{align}
Suppose these operators act on a state $\ket{\Psi}=\sum_m\psi(m)\ket{m}_z$ whose envelope is slowly varying near $m=0$. Their action gives
\begin{align}
    {}_z\bra{m}\hat{J}_y\ket{\Psi}  \approx\iu j \frac{\d\psi}{\d m}(m), \quad {}_z\bra{m}\hat{J}_z\ket{\Psi} = m \psi(m).
\end{align}
In the scaled coordinate $x=\frac{m}{\sqrt{j}}$, the effective quadratures approach multiplication by $x$ and the derivative $-\iu\partial_x$.

The normalization of the DGW quadratures must also be taken into account. Identifying $\ket{x_m}$ with $\ket{m}_z$ for $m=-j,\ldots,j$, the DGW position operator is exactly $\hat{X}=\sqrt{\frac{2\pi}{d}}\hat{J}_z$, while its momentum operator acts locally on a smooth DGW envelope as
\begin{align}
    \bra{x_m}\hat{P}\ket{\Psi}  \approx-\iu\sqrt{\frac{d}{2\pi}}\frac{\d\psi}{\d m}(m).
\end{align}
Consequently, noting that $d=2j+1$, the correspondence with the spin quadratures is
\begin{align}
    \hat{X}
    =\sqrt{\frac{2\pi j}{d}}\,\hat{X}_{\mathrm{eff}},
    \qquad
    \hat{P}
    \approx\sqrt{\frac{d}{2\pi j}}\,\hat{P}_{\mathrm{eff}}.
\end{align}
These scale factors can be accommodated by calibrating the displacement strengths and squeezing the prepared wavepacket to match its DGW width. In particular, our resource Hamiltonian remains
\begin{align}
    \hat{H}=\hat{X}^2+\hat{P}^2
    \approx\frac{2\pi j}{d}\hat{X}_{\mathrm{eff}}^2
    +\frac{d}{2\pi j}\hat{P}_{\mathrm{eff}}^2
\end{align}
in the local oscillator regime.

This correspondence is necessarily local. The rigorous DGW construction uses a periodic phase space, whereas a collective spin has a spherical phase space and a waveguide array has a finite set of spatial modes. Our bulk condition keeps the ideal DGW wavepacket away from its phase-space boundaries. A spin realization additionally requires low excitation around the chosen pole, so that the curvature corrections remain small. The DGW boundary-error bound and the physical spin-approximation error therefore describe separate requirements.

\subsection{Hamiltonians}

\subsubsection{Quasi-Continuous Variables Architecture}

The QCV architecture uses the collective spin states of large atomic ensembles~\cite{Opatrny_2017,Braverman_2019,Opatrny_2023,Brattley_2024}. For $M$ two-level atoms, the fully symmetric spin sector has $j=M/2$ and dimension $d=M+1$. Equivalently, the Schwinger representation expresses the same angular momentum algebra using two bosonic modes with fixed total occupation $2j$~\cite{Schwinger_2015}.

To mimic the continuous harmonic oscillator, QCV constructs the Hamiltonian using the quadratic form of the transverse-spin operators~\cite{Opatrny_2017}:
\begin{align}
    \hat{H}_{\mathrm{QCV}}
    =\frac{1}{2j}\left[j(j+1)\id-\hat{J}_x^2\right]
    =\frac{1}{2j}\left(\hat{J}_y^2+\hat{J}_z^2\right)
    =\frac12\left(\hat{X}_{\mathrm{eff}}^2+\hat{P}_{\mathrm{eff}}^2\right).
\end{align}
Here we have used conservation of the total spin, $\hat{J}_x^2+\hat{J}_y^2+\hat{J}_z^2=j(j+1)\id$. Its eigenstates are the $\hat{J}_x$ eigenstates $\ket{m_x}_x$. Letting $n=j-m_x$ denote the excitation number around the $+x$ pole, the energies are
\begin{align}
    E_n=\frac{j(j+1)-(j-n)^2}{2j}
    =\frac12+n-\frac{n^2}{2j}.
\end{align}
At low excitation, $E_n\approx n+1/2$, while the quadratic correction produces an anharmonic spectrum at higher excitation. The minimum energy is attained at both $m_x=j$ and $m_x=-j$; the local QCV description considered here uses the sector near $m_x=j$.

\subsubsection{Canonical Kravchuk Oscillator Architecture}

The canonical Kravchuk oscillator architecture~\cite{Hakioglu_2000} instead preserves an exactly equidistant energy spectrum. Its eigenstate wavefunctions are the normalized Kravchuk functions. To obtain the spectrum $E_n=n$, with the ground-state energy set to zero, its Hamiltonian is
\begin{align}
    \hat{H}_{\mathrm{Krav}}=j\id-\hat{J}_x.
\end{align}
In the $\hat{J}_z$ basis, this is a tight-binding hopping Hamiltonian~\cite[Eq.~(13)]{Hakioglu_2000}:
\begin{align}
    {}_z\!\bra{m}\hat{H}_{\mathrm{Krav}}\ket{\Psi}  =j\psi(m) -\frac12\sqrt{j(j+1)-m(m+1)}\psi(m+1) -\frac12\sqrt{j(j+1)-m(m-1)} \psi(m-1).
\end{align}
The relation between the two Hamiltonians is particularly simple:
\begin{align}
    \hat{H}_{\mathrm{QCV}}  =\hat{H}_{\mathrm{Krav}}+\frac{1}{2}\id  -\frac{\hat{H}_{\mathrm{Krav}}^2}{2j}.
\end{align}
On the subspace with Kravchuk excitation number at most $n_{\max}$, the norm of the last term is at most $n_{\max}^2/(2j)$. Thus, up to the zero-point energy, the two Hamiltonians agree at sufficiently low excitation. For a slowly varying envelope, the discrete Laplacian
\begin{align}
    \Delta\psi(m)=\psi(m+1)-2\psi(m)+\psi(m-1),
\end{align}
approaches $j^{-1}\partial_x^2$ in the coordinate $x=\frac{m}{\sqrt{j}}$. The local oscillator limit is therefore
\begin{align}
    \hat{H}_{\mathrm{Krav}}    \approx-\frac12\partial_x^2+\frac12x^2-\frac12.
\end{align}

\subsection{Free operations}

To manipulate this finite phase space, we use the discrete counterparts of CV Gaussian operations: phase-space displacements, phase-space rotations, and discrete Bogoliubov transformations. We call them free operations in the sense of the available control toolbox; the initial probe energy is still measured with respect to $\hat{H}=\hat{X}^2+\hat{P}^2$.

The first elementary operations are linear spin rotations generated by $\hat{J}_y$ and $\hat{J}_z$. In the low-excitation sector, these realize translations of the effective quadratures:
\begin{align}
    \hat{D}(\delta x,\delta p)  :=\exp\!\left[-\iu\delta x\,\hat{P}_{\mathrm{eff}} +\iu\delta p\,\hat{X}_{\mathrm{eff}}\right]   =\exp\!\left[\frac{\iu}{\sqrt{j}}   \left(\delta x\,\hat{J}_y+\delta p\,\hat{J}_z\right)\right].
\end{align}
The second elementary operation is the phase-space rotation generated by the longitudinal spin operator,
\begin{align}
    \hat{R}_x(\theta):=e^{-\iu\theta\hat{J}_x}.
\end{align}
Its action on the effective quadratures is exact:
\begin{align}
    \hat{R}_x(\theta)\hat{X}_{\mathrm{eff}}\hat{R}_x(\theta)^\dagger  &=\hat{X}_{\mathrm{eff}}\cos\theta+\hat{P}_{\mathrm{eff}}\sin\theta,\\
    \hat{R}_x(\theta)\hat{P}_{\mathrm{eff}}\hat{R}_x(\theta)^\dagger  &=-\hat{X}_{\mathrm{eff}}\sin\theta+\hat{P}_{\mathrm{eff}}\cos\theta.
\end{align}
In particular, the quarter rotation $\hat{F}_j=e^{-\iu\pi\hat{J}_x/2}$ exchanges the quadratures:
\begin{align}
    \hat{F}_j\hat{X}_{\mathrm{eff}}\hat{F}_j^\dagger  =\hat{P}_{\mathrm{eff}},   \qquad   \hat{F}_j\hat{P}_{\mathrm{eff}}\hat{F}_j^\dagger   =-\hat{X}_{\mathrm{eff}}.
\end{align}
The third elementary operations are quadratic phase gates generated by one-axis twisting of the transverse-spin operators:
\begin{align}
    \hat{U}_X(\lambda) :=e^{-\iu\lambda\frac{\hat{J}_z^2}{2j}}  =e^{-\iu\frac{\lambda}{2}\hat{X}_{\mathrm{eff}}^2},  \qquad   \hat{U}_P(\lambda) :=e^{-\iu\lambda\frac{\hat{J}_y^2}{2j}}   =e^{-\iu\frac{\lambda}{2}\hat{P}_{\mathrm{eff}}^2}.
\end{align}
These act as quadratic phase gates and shears in the local harmonic regime. Composing them with rotations gives the usual single-mode Gaussian transformations. Experimentally, linear spin rotations can be driven by microwave fields or Raman transitions, while quadratic generators arise from twisting interactions~\cite{Opatrny_2017}.

In the Kravchuk framework, the discrete fractional Fourier transformation is generated by the linear Hamiltonian~\cite{Hakioglu_2000}:
\begin{align}
    U_{\mathrm{DFrFT}}(\alpha)  =e^{-\iu\alpha\hat{H}_{\mathrm{Krav}}}  = e^{-\iu\alpha j}\hat{R}_x(-\alpha).
\end{align}
It is therefore a spin rotation up to a global phase. At $\alpha=\pi/2$, it gives the Fourier--Kravchuk transformation, whose oscillator limit is a Fourier transformation; it differs from the cyclic discrete Fourier transformation used in the DGW construction. Weimann \textit{et al.}~\cite{Weimann_2016} implemented the discrete fractional Fourier transformation through optical propagation in engineered arrays of coupled waveguides, where the propagation distance controls the transformation order.

\subsection{Ground state and preparation}

The Kravchuk vacuum and the QCV ground state in the chosen $+x$ sector are both the spin coherent state $\ket{j}_x$. In the $\hat{J}_z$ basis, it has the exact binomial expansion
\begin{align}
    \ket{j}_x   =\frac{1}{2^j}\sum_{m=-j}^j   \sqrt{\binom{2j}{j+m}}\,\ket{m}_z.
\end{align}
For $|m|=O(\sqrt{j})$, Stirling's formula gives
\begin{align}
    \psi_0(m):={}_z\!\bra{m}j\rangle_x    \approx \frac{1}{(\pi j)^{\frac{1}{4}}} e^{-\frac{m^2}{2j}} . 
\end{align}
Thus, its central wavefunction approaches a Gaussian oscillator vacuum. The same result follows from the highest-weight condition
\begin{align}
    (\hat{J}_z-\iu\hat{J}_y)\ket{j}_x=0,  \qquad  (\hat{X}_{\mathrm{eff}}+\iu\hat{P}_{\mathrm{eff}})\ket{j}_x=0.
\end{align}
In the coordinate $x=\frac{m}{\sqrt{j}}$, this becomes $(x+\partial_x)\phi_0(x)\approx0$, with solution $\phi_0(x)\propto e^{-\frac{1}{x}x^2}$. Moreover,
\begin{align}
    \hat{F}_j\ket{j}_x=e^{-\iu \frac{\pi}{2} j}\ket{j}_x,
\end{align}
so the spin vacuum is invariant under the quarter rotation up to a phase.

In atomic ensembles, optical pumping prepares a polarized spin state, and a subsequent global spin rotation aligns the polarization with the $+x$ axis~\cite{Opatrny_2017,Opatrny_2023}. A fixed squeezing matches the DGW width $\sigma=\sqrt{d}$, and phase-space displacements then prepare approximate DGW states centered at $(n_0,k_0)$, provided the states remain in the local oscillator regime.

Alternatively, discrete Gaussian wavefunctions can be prepared algorithmically on digital multi-qubit arrays~\cite{Kitaev_2008,Rattew_2021,Xie_2025}. Such preparation is performed to a prescribed finite accuracy, with the computational cost depending on the number of qubits and the required precision.

\section{Physics of Attenuation}
\label{SI-sec:attenuation}

We assume tunable attenuation factors in oracle accesses to phase-space translations, introduced to remove the phase-wrapping ambiguity of the Ramsey experiment. The attenuated low-gain stages identify the correct phase branch, after which the final unattenuated stage uses the full local sensitivity. The estimation procedure is described in Section~\ref{SI-sec:metrology_with_qswitch}; here we discuss how the attenuation can be realized physically.

The required control is a known change in the interaction time or coupling strength. Furthermore, this attenuation can be passed on from one layer to the next layer, meaning we only need to attenuate the initial continuous phase space signal rather than the final discrete phase space signal.  For example, if the unknown displacement amplitudes multiply the known generators through
\begin{align}
    \hat{H}_{\mathrm{int}}(t) =f(t)\left(\gamma_x\hat{P}-\gamma_p\hat{X}\right), 
\end{align}
where $\int_0^T f(t)=\d t=\lambda$, then the interaction implements $G(\lambda\gamma)$. The calibrated envelope $f(t)$ controls the attenuation without requiring knowledge of the unknown amplitudes. We first describe a driven cavity and then discuss opto-mechanical, electro-optical, and superconducting implementations of this control.

\subsection{General framework for phase-space displacement}

Here we consider a single cavity mode. Let $a$ and $a^\dagger$ denote its annihilation and creation operators. In units with $\hbar=1$, the uncoupled cavity mode is described by
\begin{align}
    H_{\mathrm{cavity}}=\omega_c a^\dagger a.
\end{align}
The cavity is coupled to an input port through which a classical coherent microwave or optical pulse with complex envelope $\beta_{\mathrm{in}}(t)$ is injected. In the standard input-output normalization, this produces the drive Hamiltonian
\begin{align}
    H_{\mathrm{input}}(t)  =\iu\sqrt{\kappa_{\mathrm{in}}}  \left(\beta_{\mathrm{in}}(t)e^{-\iu\omega_d t}a^\dagger  -\beta_{\mathrm{in}}^*(t)e^{\iu\omega_d t}a\right),
    \label{SI-eq:cavity_injection_lab}
\end{align}
where $\kappa_{\mathrm{in}}$ is the coupling rate to the input port and $\omega_d$ is the drive frequency. In a rotating frame at resonance, $\omega_d=\omega_c$, the displacement part becomes
\begin{align}
    H_{\mathrm{disp}}(t)=\iu\left(\epsilon(t)a^\dagger-\epsilon^*(t)a\right), 
\end{align}
where $\epsilon(t)=\sqrt{\kappa_{\mathrm{in}}}\beta_{\mathrm{in}}(t)$. If the pulse is short compared with the relevant decoherence times and unwanted interactions are negligible, the evolution is, up to a global phase,
\begin{align}
    U_{\mathrm{disp}}(T)=e^{\alpha a^\dagger-\alpha^*a}, 
\end{align}
where $ \alpha=\int_0^T\epsilon(t)\d t$. Thus, rescaling the drive amplitude, restricting the interaction time or reducing the coupling strength by $\lambda$ gives the attenuated displacement $\lambda\alpha$. The following examples control these parameters through modulation, optical transmissivity, or tunable coupler.

\subsection{Opto-mechanical system}

The mechanism for strength control in an opto-mechanical system is that motion of a cavity mirror modulates the cavity frequency and hence its effective coupling to an external field. We begin with the canonical radiation-pressure Hamiltonian~\cite{Law_1995,Aspelmeyer_2014}, where the cavity mode $\hat{a}$ couples to a mechanical mode $\hat{b}$ via radiation pressure:
\begin{align}
    \hat{H}_{\mathrm{bare}} =\omega_c\hat{a}^\dagger\hat{a} +\omega_m\hat{b}^\dagger\hat{b}  - g_0\hat{a}^\dagger\hat{a}(\hat{b}^\dagger+\hat{b}).
\end{align}
Suppose the mechanical oscillator is driven coherently at frequency $\Omega_m$. Neglecting its fluctuations relative to the classical motion, we replace $\hat{b}$ by $\beta e^{-\iu\Omega_m t}$. Choosing the drive phase so that $\beta$ is real gives~\cite{Law_1995,Heinrich_2010}
\begin{align}
    \hat{H}_S(t)     =\left[\omega_c+\delta\cos(\Omega_m t)\right]\hat{a}^\dagger\hat{a},
\end{align}
where $\delta=-2g_0\beta$. Now couple the cavity to an external mode $\hat{c}$ at frequency $\omega_c$ 
\begin{align}
    \hat{H}_B=\omega_c\hat{c}^\dagger\hat{c},
\end{align}
through
\begin{align}
    \hat{H}_{S-B}=g(\hat{a}^\dagger\hat{c}+\hat{a}\hat{c}^\dagger).
\end{align}
In the interaction picture with respect to $\hat{H}_S(t)+\hat{H}_B$, this becomes
\begin{align}
    \widetilde{H}_{\mathrm{int}}(t)
    =g\left[
    \hat{a}^\dagger\hat{c}\,
    e^{\iu(\delta/\Omega_m)\sin(\Omega_m t)}
    +\mathrm{h.c.}\right].
\end{align}
Using the Jacobi--Anger expansion, as in frequency-modulated interactions~\cite{Liao_2015,Liao_2016,Wang_2019},
\begin{align}
    e^{\iu z\sin\theta}  =\sum_{n=-\infty}^{\infty}J_n(z)e^{\iu n\theta},
\end{align}
we obtain
\begin{align}
    \widetilde{H}_{\mathrm{int}}(t)   = g\sum_{n=-\infty}^{\infty}    J_n\!\left(\frac{\delta}{\Omega_m}\right)  \left(\hat{a}^\dagger\hat{c}\,e^{\iu n\Omega_m t}     +\hat{a}\hat{c}^\dagger e^{-\iu n\Omega_m t}\right).
\end{align}
When $\Omega_m$ is large compared with the relevant coupling rates and the nonzero sidebands are off resonance, averaging over the modulation gives the leading effective Hamiltonian
\begin{align}
    \hat{H}_{\mathrm{eff}} \approx g_{\mathrm{eff}}(\hat{a}^\dagger\hat{c}  +\hat{a}\hat{c}^\dagger),
\end{align}
where 
\begin{align}
    g_{\mathrm{eff}}=gJ_0\!\left(\frac{\delta}{\Omega_m}\right). 
\end{align}
Thus, the effective coupling can be attenuated by tuning the modulation amplitude or frequency. Its leading term vanishes at a zero of $J_0$. 

To realize a displacement, the external mode supplies a coherent field whose depletion and fluctuations can be neglected. Replacing $\hat{c}$ by its coherent amplitude gives a Hamiltonian linear in $\hat{a}$ and $\hat{a}^\dagger$, as in the general framework above. In this regime, changing $g_{\mathrm{eff}}$ by a calibrated factor changes the displacement amplitude by the same factor.

\subsection{Electro-optical system}

In an electro-optical system, a classical voltage drives a nonlinear optical crystal and changes its refractive index~\cite{Tsang_2010,Tsang_2011}. When the crystal is embedded in an interferometer, the resulting phase shift controls its transmissivity, giving a tunable beam splitter~\cite{Ma_2011,Ma_2012,Chan_2013}.

An electro-optical modulator (EOM) relies on the Pockels effect. For a crystal of thickness $d_{\mathrm{cr}}$, the refractive index changes approximately as
\begin{align}
    n(V)=n_0-\frac{n_0^3r}{2d_{\mathrm{cr}}}V,
\end{align}
where $r$ is the Pockels coefficient. After fixing the phase bias and voltage convention, the relative phase shift can be written as
\begin{align}
    \phi(V)=\pi\frac{V}{V_\pi},
\end{align}
where $V_\pi$ is the calibrated half-wave voltage.

To obtain variable transmissivity, the EOM is embedded in a Mach--Zehnder interferometer containing two balanced beam splitters. With a suitable port-phase convention, its input-output transformation is
\begin{align}
    \begin{pmatrix}
        \hat{a}_{\mathrm{out}}\\
        \hat{b}_{\mathrm{out}}
    \end{pmatrix}
    =
    \begin{pmatrix}
        \cos(\phi/2)&\iu\sin(\phi/2)\\
        \iu\sin(\phi/2)&\cos(\phi/2)
    \end{pmatrix}
    \begin{pmatrix}
        \hat{a}_{\mathrm{in}}\\
        \hat{b}_{\mathrm{in}}
    \end{pmatrix}.
\end{align}
The intensity transmissivity from the first input port to the first output port is therefore
\begin{align}
    \eta(V)=\cos^2\!\left(\frac{\phi(V)}{2}\right)     =\cos^2\!\left(\frac{\pi V}{2V_\pi}\right).
\end{align}
For a coherent input in the first port and vacuum in the second, the first output has coherent amplitude $\cos(\phi/2)\beta_{\mathrm{in}}$. Operating with $0\leq\phi\leq\pi$ gives the amplitude attenuation factor $\lambda=\sqrt{\eta}$.

For our displacement oracle, the interferometer attenuates the incoming coherent driving field before it interacts with the probe. The drive envelope is then multiplied by $\lambda$, and so is the induced displacement. 

\subsection{Superconducting system}

In a superconducting system, an external magnetic flux controls the effective Josephson energy of a SQUID~\cite{Makhlin_1999,Blais_2003,Averin_2003}, which can serve as a tunable coupler between two LC oscillators~\cite{Peropadre_2013}. Tunable and switchable coupling between superconducting resonators has also been demonstrated experimentally~\cite{Baust_2015}.

Consider a symmetric SQUID with two Josephson junctions of energy $E_J$. Let $\varphi_1$ and $\varphi_2$ be the dimensionless phase differences across the junctions and define
\begin{align}
    \varphi_+:=\frac{\varphi_1+\varphi_2}{2},
    \qquad
    \varphi_-:=\frac{\varphi_1-\varphi_2}{2},
    \qquad
    \Phi_0:=\frac{h}{2e}.
\end{align}
For negligible loop inductance, flux quantization fixes $\varphi_-=\pi\Phi_{\mathrm{ext}}/\Phi_0$ modulo $\pi$. The Josephson potential is
\begin{align}
    U_{\mathrm{SQUID}} =-E_J\cos\varphi_1-E_J\cos\varphi_2 =-2E_J\cos\!\left(\frac{\pi\Phi_{\mathrm{ext}}}{\Phi_0}\right)
    \cos\varphi_+.
\end{align}
For identical junction capacitances $C$, the capacitive energy is
\begin{align}
    K_{\mathrm{SQUID}}
    =\frac{C}{2}\left(\frac{\Phi_0}{2\pi}\right)^2
    (\dot{\varphi}_1^2+\dot{\varphi}_2^2)
    =C\left(\frac{\Phi_0}{2\pi}\right)^2
    (\dot{\varphi}_+^2+\dot{\varphi}_-^2).
\end{align}
At fixed flux, expanding around the stable point $\varphi_+=0$ in a region with positive effective Josephson energy gives
\begin{align}
    U_{\mathrm{SQUID}}
    \approx\mathrm{const}
    +E_J\cos\!\left(\frac{\pi\Phi_{\mathrm{ext}}}{\Phi_0}\right)
    \hat{\varphi}_+^2.
\end{align}
When the SQUID connects two resonator modes, its phase contains contributions from both modes. For a representative coupling geometry, write
\begin{align}
    \hat{\varphi}_+
    =\xi_a(\hat{a}+\hat{a}^\dagger)
    -\xi_b(\hat{b}+\hat{b}^\dagger),
\end{align}
where $\xi_a$ and $\xi_b$ are their zero-point phase amplitudes at the coupler. Expanding $\hat{\varphi}_+^2$ produces local frequency shifts and a cross term. Absorbing the local terms into the resonator frequencies and applying the rotating wave approximation near resonance gives
\begin{align}
    \hat{H}_{\mathrm{int}}(\Phi_{\mathrm{ext}}) \approx g(\Phi_{\mathrm{ext}}) (\hat{a}^\dagger\hat{b}+\hat{a}\hat{b}^\dagger),
\end{align}
where 
\begin{align}
    g(\Phi_{\mathrm{ext}})=g_0\cos\!\left(\frac{\pi\Phi_{\mathrm{ext}}}{\Phi_0}\right),
\end{align}
where $g_0=-2E_J\xi_a\xi_b$. Here the mode amplitudes are treated as fixed in the weak-coupling approximation; in a device, the flux dependence of the coupling is calibrated experimentally.

When one resonator supplies a coherent driving field, the interaction again reduces to a Hamiltonian linear in the probe-mode operators. Tuning the external flux then changes the effective drive strength, while tuning the pulse duration changes its interaction time. Either control realizes a calibrated attenuation of the displacement within the regime of the effective Hamiltonian.

\end{document}